\documentclass[11pt]{article}
\usepackage[english]{babel}
\usepackage[document]{ragged2e}
\usepackage[skip=10pt plus1pt, indent=40pt]{parskip}
\usepackage{longtable}
\usepackage[top=0.75in, bottom=0.5in, left=0.8in, right=0.8in,includefoot,heightrounded]{geometry}
\usepackage[utf8]{inputenc}
\usepackage[procnames]{listings}
\usepackage{pdflscape}
\usepackage{upgreek}
\usepackage{fancyhdr}

\usepackage{mathtools}
\usepackage{dsfont}
\usepackage{amssymb, amsmath, amsbsy,amsthm,amscd}
\usepackage{bigints}
\usepackage{stmaryrd}
\usepackage{cases}
\usepackage{accents}

\newtheorem{defn0}{Definition}[section]
\newtheorem{prop0}[defn0]{Proposition}
\newtheorem{thm0}[defn0]{Theorem}
\newtheorem{lemma0}[defn0]{Lemma}
\newtheorem{corollary0}[defn0]{Corollary}
\newtheorem{example0}[defn0]{Example}
\newtheorem{remark0}[defn0]{Remark}
\newtheorem{conjecture0}[defn0]{Conjecture}
\newtheorem{ansatz0}[defn0]{Ansatz}
\newtheorem{notation0}[defn0]{Notation}

\newenvironment{definition}{ \begin{defn0}}{\end{defn0}}
\newenvironment{proposition}{\bigskip \begin{prop0}}{\end{prop0}}
\newenvironment{theorem}{\bigskip \begin{thm0}}{\end{thm0}}
\newenvironment{lemma}{\bigskip \begin{lemma0}}{\end{lemma0}}
\newenvironment{corollary}{\bigskip \begin{corollary0}}{\end{corollary0}}

\newenvironment{remark}{ \bigskip \begin{remark0}}{\end{remark0}}

\newenvironment{notation}{\begin{notation0}}{\end{notation0}}

\newcommand{\transv}{\mathrel{\text{\tpitchfork}}}
\newcommand{\bil}[2]{(#1\mid #2)} 
\mathtoolsset{showonlyrefs}

\numberwithin{equation}{section}
\numberwithin{figure}{section}

\usepackage{overpic}
\usepackage{subcaption}
\usepackage{graphicx}

\usepackage{orcidlink}
\usepackage{hyperref}
\definecolor{ao(english)}{rgb}{0.0, 0.5, 0.0}
\definecolor{DBrown}{RGB}{161, 104, 17}
\definecolor{RawSienna}{RGB}{125, 73, 49}
\definecolor{RedViolet}{RGB}{146, 43, 62}
\definecolor{Cerulean}{RGB}{29,172,214}
\definecolor{blueKepler}{RGB}{45,107,156}
\definecolor{redKepler}{RGB}{199,86,18}
\definecolor{grey}{RGB}{129,129,129}
\hypersetup{
    colorlinks=true,
    linkcolor = black,
    citecolor = Cerulean,
    filecolor=magenta,      
    urlcolor=blue,
}

\newcommand\Alpha{\mathrm{A}}
\newcommand{\constantValue}{\lambda}

\newcommand{\pS}{\overline p}
\newcommand{\rS}{\overline r}
\newcommand{\thetaS}{\overline \theta}
\newcommand{\RS}{\overline R}
\newcommand{\ThetaS}{\overline \Theta}

\usepackage{authblk}

\makeatletter
\newcommand\blfootnote[1]{%
  \begingroup
  \renewcommand\thefootnote{}\footnote{#1}%
  \addtocounter{footnote}{-1}%
  \endgroup
}
\newcommand{\tpitchfork}{%
  \vbox{
    \baselineskip\z@skip
    \lineskip-.52ex
    \lineskiplimit\maxdimen
    \m@th
    \ialign{##\crcr\hidewidth\smash{$-$}\hidewidth\crcr$\pitchfork$\crcr}
  }%
}
\makeatother

\usepackage{biblatex}
\title{Second-species dynamics in the restricted planar circular three body problem: chaos, final motions and periodic orbits}
\author[1,3]{Marcel Guardia \orcidlink{0000-0002-4802-3151}}
\author[2]{Jos\'e Lamas \orcidlink{0000-0002-1809-1823}}
\author[2,3]{Tere M-Seara \orcidlink{0000-0001-8421-8717}}

\affil[1]{Departament de Matem\`atiques i Inform\`atica, Universitat de Barcelona, Gran Via, 585, 08007 Barcelona, Spain}
\affil[2]{Departament de Matem\`atiques \& IMTECH, Universitat Polit\`ecnica de Catalunya, Diagonal 647, 08028 Barcelona, Spain}
\affil[3]{Centre de Recerca Matem\`atica, Campus de Bellaterra, Edifici C, 08193, Barcelona, Spain}

\date{}

\begin{document}
\justifying
\maketitle

\begin{abstract}
\blfootnote{\textit{E-mail addresses}: \href{mailto:guardia@ub.edu}{guardia@ub.edu} (M.Guardia), \href{mailto:jose.lamas.rodriguez@upc.edu}{jose.lamas.rodriguez@upc.edu} (J.Lamas), \href{mailto:tere-m.seara@upc.edu}{tere-m.seara@upc.edu} (Tere M-Seara).}

Consider the Restricted Planar Circular Three Body Problem (RPC3BP), which models the motion of a massless particle (Asteroid) under the gravitational influence of two massive bodies (the primaries) moving on circular orbits. By considering the ratio between the masses of the primaries to be arbitrarily small, we construct orbits with close encounters with the smaller primary (Jupiter) that realize any combination of past and future final motions (in the sense of Chazy’s), including oscillatory motions. We also obtain arbitrarily large ejection-collision orbits with Jupiter and ejection-collision orbits between the two primaries (Sun and Jupiter), as well as arbitrarily large periodic orbits that pass arbitrarily close to Jupiter. Our approach combines singular perturbation theory and Levi-Civita regularization near Jupiter, and McGehee regularization near infinity and near the Sun, together with a global analysis that leads to transverse intersections of invariant manifolds.
\end{abstract}
\tableofcontents
\section{Introduction}\label{sec: Introduction}
In his \textit{M\'ethodes nouvelles de la m\'ecanique c\'eleste}~\cite{poincare1987methodes}, Poincar\'e pointed out the fundamental role of collision and near-collision orbits in celestial mechanics. In particular, he introduced the concept of \textit{second-species} solutions: periodic orbits that, while avoiding actual collision, pass arbitrarily close to it. Their construction typically involves singular perturbation methods and careful analysis of the flow near collision via regularization techniques. In contrast, \textit{first-species} solutions correspond to periodic orbits that are continuation of those of the  $2$ body problem that arise through regular perturbations of integrable systems. 
In recent years, a variety of techniques (ranging from variational methods to geometric and perturbative approaches) have been developed to construct second-species solutions and to analyze their structure (see for instance \cite{MR2331205,MR2245344,MR1805879,MR1335057}).

A natural question concerning these near-collision behaviors is how they relate to the broader picture of long-term dynamics of the $3$ body problem. These dynamics were studied extensively by Chazy, who classified all complete solutions (those defined for all positive or negative time) based on their asymptotic behavior. In the RPC3BP, this classification is reduced to 4 possible asymptotic behaviors. To provide the classification, let $q$ denote the position of the Asteroid.
%
\begin{theorem}[Chazy, 1922, see also \cite{MR2269239}]\label{thm:chazyJ}
Every solution of the RPC3BP defined for all time belongs to one of the following four classes.
\begin{itemize}
    \item Hyperbolic ($\mathcal{H}^\pm$) $\colon \underset{t\to \pm \infty}{\lim} \|q(t)\| = \infty$ and $\underset{t\to \pm \infty}{\lim}\|\dot q(t)\| = c > 0$.
    \item Parabolic ($\mathcal{P}^\pm)$ $\colon \underset{t\to \pm \infty}{\lim} \|q(t)\| = \infty$ and $\underset{t\to \pm \infty}{\lim}\|\dot q(t)\| = 0$.
    \item Bounded ($\mathcal{B}^\pm$) $ \colon \underset{t\to \pm \infty}{\limsup}\, \|q(t)\| < \infty$.
    \item Oscillatory ($\mathcal{OS}^\pm$) $\colon \underset{t\to \pm \infty}{\limsup}\, \|q(t)\| = \infty$ and $\underset{t\to \pm \infty}{\liminf}\,\|q(t)\| < \infty$.
\end{itemize}
\end{theorem}
Note that this classification applies both when $t\to+\infty$ or $t\to-\infty$. To distinguish both cases we add a superindex $+$ or $-$ to each of the cases, e.g $\mathcal{P}^+$ and $\mathcal{P}^-$.

While this classification excludes collisions, it provides a natural framework in which to understand global dynamics, and it remains central to the description of the long-term motion in the $3$ body problem. Among these classes, oscillatory motions are particularly subtle: they exhibit infinitely many excursions far away from the primaries without tending to infinity. The first example of oscillatory motions was given by Sitnikov in \cite{MR0127389}. Since then, more recent works have proven their existence in several models by showing their connection to chaotic behavior (see for instance \cite{MR3455155,MR3583476,moser2001stable}).

The existence of orbits exhibiting oscillatory behavior (and transitions between different types of final motions) has led to renewed interest in how near-collision dynamics fit into the Chazy framework. In our previous work \cite{guardia2024oscillatorymotionsparabolicorbits}, we analyzed the different types of final motions for the 
RPC3BP (including oscillatory motions) passing near the massive primary, referred to as the Sun. These orbits could be studied through classical perturbation theory because the system, also near collision, is a regular perturbation of the Sun-Asteroid $2$ body problem. Using a slight abuse of terminology in line with Poincaré’s original classification, we refer to these orbits as \emph{first-species} solutions to emphasize their distinction with the ``more singular'' behavior of the orbits considered in this article.

The present paper focuses on the behavior near the smaller primary, referred to as Jupiter. Here, the problem possesses a singular perturbation regime in which the dynamics can no longer be viewed as a small perturbation of a Kepler problem between the Sun and the Asteroid. As a result, the analytical framework used in \cite{guardia2024oscillatorymotionsparabolicorbits} is no longer applicable, and a different class of geometric and regularization techniques is required. The orbits constructed here (including oscillatory motions, symbolic dynamics, and ejection-collision trajectories) although similar in nature to those in \cite{guardia2024oscillatorymotionsparabolicorbits}, arise through fundamentally different mechanisms and are  referred to as \textit{second-species} solutions in this broader sense.

\subsection{Main results}\label{subsec: main results}
We consider the RPC3BP and normalize the masses of the primaries to be $1-\mu$ and $\mu$, with $\mu \in (0,1/2]$ denoting the mass ratio. Throughout this paper we consider $\mu> 0$ small, and accordingly we refer to the more massive primary as the Sun and the smaller one as Jupiter. Taking the appropriate units, the RPC3BP is a Hamiltonian system with respect to
\begin{equation}\label{eqn: non autonomous Hamiltonian}
    \mathrm{H}_\mu(Q,P,t) = \frac{|P|^2}{2} - \frac{1-\mu}{|Q+\mu Q_0(t)|} - \frac{\mu}{|Q-(1-\mu)Q_0(t)|},
\end{equation}
where $Q,P\in \mathds{R}^2$ and $-\mu Q_0(t)$ and $(1-\mu)Q_0(t)$ with $Q_0(t) = (\cos t,\sin t)$ are the positions of the primaries. 

The symplectic change to rotating coordinates $(\hat q,\hat p)$ makes the Hamiltonian \eqref{eqn: non autonomous Hamiltonian} autonomous and of the form 
\begin{equation}\label{eqn: Hamiltonian in synodical cartesian coordinates centered at CM J}
\begin{aligned}
    \hat{H}_\mu(\hat{q},\hat{p}) = \frac{|\hat{p}|^2}{2} - \left(\hat{q}_1\hat{p_2} -\hat{q}_2\hat{p}_1\right) - \frac{1-\mu}{|\hat{q} + (\mu,0)|}- \frac{\mu}{|\hat{q} - (1-\mu,0)|},
\end{aligned}
\end{equation}
where $\hat{q},\hat{p}\in \mathds{R}^2$ and $(-\mu,0)$ and $(1-\mu,0)$ are the new position of the primaries, which are now fixed.

For any $\mu\in (0,1/2]$ it is known (see \cite{MR3455155,MR0573346}) that 
\[
X^{+}\cap Y^{-}\neq \emptyset\quad \text{where}\quad X,Y=\mathcal{H},\mathcal{P},\mathcal{B},\mathcal{OS}.
\]
Note that, when $\mu=0$, the RPC3BP is reduced to a Kepler problem and therefore $\mathcal{OS}^\pm=\emptyset$ and $\mathcal{H}^+=\mathcal{H}^-$, $\mathcal{P}^+=\mathcal{P}^-$, $\mathcal{B}^+=\mathcal{B}^-$. 

We define the Sun and Jupiter as the collision sets
\begin{equation}\label{eqn: Collision sets}
    \mathcal{S} = \left\{(\hat{q},\hat{p})\in \mathds R^4\colon \hat{q} = (-\mu,0)\right\},\quad \mathcal{J}=\left\{(\hat{q},\hat{p})\in\mathds R^4\colon \hat{q} = (1-\mu,0)\right\}.
\end{equation}
Let us define also the ejection and collision orbits.
\begin{definition}\label{def: EC-orbits}
    An ejection or collision orbit $\hat \gamma(t) = (\hat q(t), \hat p(t))$ associated to the Hamiltonian \eqref{eqn: Hamiltonian in synodical cartesian coordinates centered at CM J} belongs to one of the following families.
    \begin{itemize}
        \item Ejection orbits from the Sun ($\mathcal S^-$) $\colon$ there exists $t_0\in \mathds{R}$ such that $\underset{t\to t_0^+}{\lim} \hat{q}(t) = (-\mu,0)$. 
        \item Ejection orbits from Jupiter ($\mathcal J^-$) $\colon$ there exists $t_0\in \mathds{R}$ such that $\underset{t\to t_0^+}{\lim} \hat{q}(t) = (1-\mu,0)$.
        \item Collision orbits to the Sun ($\mathcal S^+$) $\colon$ there exists $t_1\in \mathds{R}$ such that $\underset{t\to t_1^-}{\lim} \hat{q}(t) = (-\mu,0)$.
        \item Collision orbits to Jupiter ($\mathcal J^+$) $\colon$ there exists $t_1\in \mathds{R}$ such that $\underset{t\to t_1^-}{\lim} \hat{q}(t) = (1-\mu,0)$. 
        \end{itemize}
        Moreover, we define
        \begin{itemize}
        \item Ejection-collision orbit with the Sun if it belongs to $\mathcal S^- \cap \mathcal S^+$.
        \item Ejection-collision orbit with Jupiter if it belongs to $\mathcal J^-\cap \mathcal J^+$.
        \item Ejection-collision orbit between the primaries if it belongs to $\mathcal J^-\cap \mathcal S^+$ or $\mathcal J^+\cap \mathcal S^-$.
    \end{itemize}
\end{definition}
The main results of the present paper are the following. The first one constructs two different ``types'' of ejection-collision orbits with Jupiter and between the primaries.
\begin{theorem}\label{thm: ejection-collision orbits}
For any open set $U\subset \left[\frac{1-\sqrt 3}{3},\frac{1+\sqrt{3}}{3}\right]$ there exists $\mu_0 > 0$ such that, for any $\mu \in (0,\mu_0)$, the following statements hold:
\begin{enumerate}
    \item For any $\hat h \in U$, there exists a sequence of trajectories $\{\hat z_k(t)\}_{k\in \mathds{N}}$, $\hat z_k(t) = (\hat q_k(t),\hat p_k(t)) \in \mathcal{J}^- \cap \mathcal{J}^+$ in the energy level $\hat H_\mu(\hat q,\hat p) = \hat h$ satisfying
    \begin{equation}\label{eqn: J+- close to inf}
        \underset{k\to \infty}{\limsup}\; \left(\underset{t\in (t_0^k,t_1^k)}{\sup}\;|\hat q_k(t)|\right) = +\infty,
    \end{equation}
    where $t_0^k,t_1^k$ are as in Definition \ref{def: EC-orbits}.
    \item There exists $\eta > 0$ (independent of $\mu$) such that, for any $\hat h\in(-\eta\mu,\eta\mu) \subset U$, one can find 
    \begin{itemize}
    \item 
    a sequence of trajectories $\{\hat z_k^1(t)\}_{k\in \mathds{N}}$, $\hat z_k^1(t) = (\hat q_k^1(t),\hat p_k^1(t)) \in \mathcal{S}^- \cap \mathcal{J}^+$,
    \item 
    a sequence of trajectories $\{\hat z_k^2(t)\}_{k\in \mathds{N}}$, $\hat z_k^2(t) = (\hat q_k^2(t),\hat p_k^2(t)) \in \mathcal{J}^-\cap \mathcal S^+$,
    \end{itemize}
    in the energy level $\hat H_\mu(\hat q,\hat p) = \hat h$ which satisfy
    \begin{equation}\label{eqn: SJ close to inf}
    \underset{k\to \infty}{\limsup}\; \left(\underset{t\in (t_{0,i}^k,t_{1,i}^k)}{\sup}\;|\hat q_k^i(t)|\right) = +\infty,\quad i=1,2,
    \end{equation}
    where $(t_{0,1}^k,t_{1,1}^k)$ and $(t_{0,2}^k,t_{1,2}^k)$ are as in Definition \ref{def: EC-orbits}.
    \item 
    There exists $\eta > 0$ (independent of $\mu$) such that, for any $\hat h\in(-\eta\mu,\eta\mu) \subset U$, one can find 
    \begin{itemize}
    \item 
    a trajectory $\hat z^3(t) = (\hat q^3(t), \hat p^3(t)) \in \mathcal S^- \cap \mathcal J^+$,
    \item 
    a trajectory $\hat z^4(t) = (\hat q^4(t), \hat p^4(t)) \in \mathcal J^- \cap \mathcal S^+$,
    \end{itemize}
    in the energy level $\hat H_\mu(\hat q,\hat p) = \hat h$ such that
    \begin{equation}\label{eqn: SJ bounded}
        \underset{t\in (t_{0,i},t_{1,i})}{\sup}\; |\hat q^i(t)|  = 1-\mu,\quad i=3,4,
    \end{equation}
    where $(t_{0,3},t_{1,3})$ and $(t_{0,4},t_{1,4})$ are as in Definition \ref{def: EC-orbits}.
\end{enumerate}
\end{theorem}

\begin{remark}\label{def: large-ballistic}
We say that an ejection-collision trajectory is \emph{large} if it belongs to a sequence of trajectories for which \eqref{eqn: J+- close to inf} holds. 
An ejection-collision trajectory between the primaries is \emph{ballistic} if it satisfies \eqref{eqn: SJ bounded}.
\end{remark}

The second main theorem shows the existence of hyperbolic sets (whose dynamics is conjugated to the infinite symbols shift) which are unbounded and contain Jupiter at their closure.

\begin{theorem}\label{thm:chaosJ} 
There exist $\eta > 0$ and $\mu_0 > 0$ such that, for any $\mu \in (0,\mu_0)$, 
there exists $\hat h \in \left(-1-\eta\mu^{\frac{1}{12}},-1+\eta\mu^{\frac{1}{12}}\right)$ 
and a section $\Pi \subset \{\hat H_\mu(\hat q,\hat p) = \hat h\}$ transverse to the flow of \eqref{eqn: Hamiltonian in synodical cartesian coordinates centered at CM J} where the induced Poincar\'e map
\[
\mathbb{P} : V \subset \Pi \to \Pi
\]
has an invariant set $\mathcal{X}$ which is homeomorphic to $\mathbb{N}^\mathbb{Z}$
and whose  dynamics $\mathbb{P}_{|_{\mathcal X}}: \mathcal{X}\to \mathcal{X}$ is topologically conjugated to the shift $\sigma: \mathbb{N}^\mathbb{Z}\to \mathbb{N}^\mathbb{Z}$, $(\sigma\omega)_k=\omega_{k+1}$.
Moreover, this invariant set $\mathcal X$ satisfies 
\begin{equation*}
    \overline{\mathcal X}\cap \mathcal J \neq \emptyset,\quad \overline{\mathcal X}\cap \mathcal J^\pm \neq \emptyset,\quad \overline{\mathcal X}\cap \mathcal P^\pm \neq \emptyset.
\end{equation*}
\end{theorem}
A  consequence of this result is the following theorem, which provides the existence of any combination of past and future asymptotic behaviors (in the sense of Chazy) accumulating to Jupiter.
\begin{theorem}\label{thm: Existence of parabolic, oscillatory and periodic orbits J} 
There exists $\mu_0 > 0$  such that, for any $\mu\in (0,\mu_0)$,
\[
\overline{X^+\cap Y^-}\cap \mathcal{J} \neq\emptyset\qquad \text{where}\qquad  X,Y=\mathcal{H},\mathcal{P},\mathcal{B},\mathcal{OS}.
\]
Moreover, 
\begin{itemize}
 \item 
 There exist trajectories $(\hat{q}(t),\hat{p}(t))$ of \eqref{eqn: Hamiltonian in synodical cartesian coordinates centered at CM J} which belong to $\mathcal{OS}^-\cap \mathcal{OS}^+$ and get arbitrarily close to collision with Jupiter. Namely, they satisfy
    \begin{equation*}
        \underset{t\to \pm \infty}{\limsup}\, |\hat q(t)|= \infty \;\;\;\; \text{and} \;\;\;\; \underset{t\to \pm \infty}{\liminf}\, |\hat q(t) - (1-\mu,0)| = 0.
    \end{equation*} 
    In particular, this also implies that they are oscillatory in their velocity:
    \begin{equation}\label{def:oscillatoryspeedjup}
            \underset{t\to \pm \infty}{\limsup}\, \left|\frac{d}{dt} \hat q(t)\right| = \infty \;\;\;\; \text{and} \;\;\;\; \underset{t\to \pm \infty}{\liminf}\, \left|\frac{d}{dt} \hat q(t) \right|<\infty.
    \end{equation}
 \item For any $\varepsilon>0$,  there exists a periodic trajectory $(\hat{q}(t),\hat{p}(t))$ of \eqref{eqn: Hamiltonian in synodical cartesian coordinates centered at CM J} satisfying
    \begin{equation*}
       \underset{t\in \mathds{R}}{\sup}\;|\hat q(t)|\geq \varepsilon^{-1} \;\;\;\; \text{and} \;\;\;\;\underset{t\in \mathds{R}}{\inf}\,|\hat q(t)- (1-\mu,0)|\leq \varepsilon.
    \end{equation*}
\end{itemize}
\end{theorem}
Theorem \ref{thm: ejection-collision orbits} is proved in Section \ref{sec: Proof of ECO SJ}. Theorems \ref{thm:chaosJ} and \ref{thm: Existence of parabolic, oscillatory and periodic orbits J} are proved  together in Section \ref{sec: Triple intersection}. 
\subsection{Literature and previous results}\label{subsec: literature and previous results} 
The study of near-collision orbits has been the subject of several works, particularly in relation to second-species solutions. After Poincaré’s initial classification \cite{poincare1987methodes}, a geometric framework for studying these solutions was developed through regularization techniques and the analysis of invariant manifolds associated with the collision set. In the context of the planar $3$ body problem, Levi-Civita and McGehee regularizations \cite{MR1555161, MR0359459} have been used to describe the flow near binary collisions. This approach has enabled the construction of near-collision orbits in both planar and spatial models~\cite{MR2331205,  MR2245344,MR1805879,Negrini,font2001consecutive,MR1335057}. One can also rely on KAM Theory to build punctured tori (invariant tori containing collisions for the regularized flow), see \cite{MR0967629, MR1849229,MR1919782,MR3417880}.



Ejection and collision orbits have also received considerable attention. Saari, and later Fleisher and Knauf, proved that the set of initial conditions leading to collision has zero Lebesgue measure \cite{fleischer2019improbability1, fleischer2019improbability2, MR0295648, MR0321386}, though it may still be topologically large. The question of whether collision orbits form a dense set within an open set of the phase space was first posed by Siegel and later formulated as a conjecture by Alekseev \cite{MR0629685}. Although it remains open, a partial answer was given in \cite{MR3951693}, where the authors prove that the set of orbits leading to collision in the RPC3BP is $\mu^\alpha$-dense (for some $\alpha>0$) in some open set of phase space.

In the RPC3BP, ejection-collision orbits were first constructed by Lacomba and Llibre \cite{MR0949626,MR0682839}. Their results were later generalized in \cite{MR4518121,MR4110029,olle2018ejection}.
Using computer assisted proofs, ballistic ejection-collision orbits (see Remark \ref{def: large-ballistic}) have been obtained in \cite{MR4576879F}, and \cite{capinski2024oscillatory} provides oscillatory orbits in velocity (in the sense of \eqref{def:oscillatoryspeedjup}) accumulating to collision. 



In the full $3$ body problem, triple collisions are also possible and have been extensively study since the work of McGehee \cite{MR0359459} (see also \cite{MR0586428,MR1013560,MR0571374,MR0636961,MR0640127}). The analysis of these singularities has provided a framework for constructing various types of trajectories in the $3$ body problem, including those with oscillatory behavior in both position and velocity (by Moeckel \cite{MR1000223, MR2350333}). These constructions rely on trajectories passing arbitrarily close to triple collision and therefore the total angular momentum has to be very close to zero. This is not the case of the present paper, where the primaries perform circular motion and therefore their angular momenta are not small. 
 
Regarding the combination of past and future final motions, it can be traced back to the work done by Sitnikov for the now known as the Sitnikov problem \cite{MR0127389}, where he showed that all of them were possible  (including oscillatory motions). A decade later, Moser \cite{moser2001stable} gave a new proof, relating these motions to chaotic behavior via the construction of Smale horseshoes. Moser’s approach has since been adapted to other versions of the restricted $3$ body problem \cite{MR3455155, MR0573346} (see \cite{paradela2022oscillatory} for similar results using other methods). In the non-restricted setting, similar results were given by Alekseev \cite{MR0249754, MR0276949, MR0276950}, and have been recently extended in \cite{guardia2022hyperbolic}.

\subsection{Main ideas for the proofs of Theorems \ref{thm: ejection-collision orbits}, \ref{thm:chaosJ}  and \ref{thm: Existence of parabolic, oscillatory and periodic orbits J}}\label{subsec: main ideas for the proofs}

The orbits constructed in Theorems \ref{thm: ejection-collision orbits}, \ref{thm:chaosJ} and \ref{thm: Existence of parabolic, oscillatory and periodic orbits J} rely on a combination of invariant manifold theory and singular perturbation techniques applied to certain objects of the RPC3BP. These include both the analysis close to the primaries (the Sun and Jupiter) as well as the behavior of the system at infinity. After suitable regularizations and compactifications, these objects can be understood within the context of different regularized flows.  To analyze the dynamics in these different regions of the phase space, we perform several changes of coordinates and time scalings. The McGehee regularizations near the Sun and at infinity (see  \cite{McGeheeInf,MR0359459}) were deeply studied in \cite{guardia2024oscillatorymotionsparabolicorbits}. In contrast, in the present paper the analysis near Jupiter is achieved through a Levi-Civita transformation \cite{MR1555161}, which is described in Section \ref{sec: local analysis at collision with J} below. 

At a fixed energy level, the McGehee regularization transforms the Sun into an invariant torus containing two normally hyperbolic invariant circles, each foliated by equilibrium points. As detailed in Section \ref{sec: local analysis close to collision with S}, the ejection and collision orbits correspond to unstable and stable manifolds associated to these circles, and their behavior can be analyzed as perturbations of the flow induced by the unperturbed Hamiltonian $\hat H_0$ in \eqref{eqn: Hamiltonian in synodical cartesian coordinates centered at CM J}. 

After compactification, the ``parabolic infinity'' (described in Section \ref{sec: The invariant manifolds of infinity}) becomes a degenerate periodic orbit at each energy level. Though degenerate, it is known to possess well-defined stable and unstable manifolds for any value of $\mu \in [0,1/2]$ (see \cite{MR0362403}).

Near Jupiter we consider the Levi-Civita regularization (see Section \ref{sec: local analysis at collision with J} below). Then, at a given energy level, the collision set becomes a circle of regular points (located close to a saddle), and the associated ejection and collision orbits are generated by the forward and backward flow of this circle. Unlike the cases of the Sun and the parabolic infinity, the dynamics close to Jupiter does not arise as a small $\mu$-perturbation of the Sun-Asteroid $2$ body problem. Rather, after a scaling and a singular perturbation analysis, the dynamics is governed by the Jupiter-Asteroid $2$ body problem and therefore a separate analysis is required.

We show that the invariant manifolds of infinity intersect transversally with the ejection and collision orbits associated to Jupiter. 
By combining these intersections with a local analysis near collisions and infinity, we construct the different types of motions provided by Theorems \ref{thm: ejection-collision orbits}, \ref{thm:chaosJ}  and \ref{thm: Existence of parabolic, oscillatory and periodic orbits J}.

Let us describe the strategy more precisely:
\begin{enumerate}
\item We prove that the stable manifold of infinity intersects transversally the ejection manifold $\mathcal J^-$ and the unstable manifold of infinity intersects transversally the collision manifold $\mathcal J^+$ (see Section \ref{sec: The distance between the invariant manifolds and EC orbits}).

\item 
Relying on the local analysis close to infinity (at the $C^1$ level) done in \cite{moser2001stable}, we prove that the ejection and collision manifolds $\mathcal J^\mp$ intersect transversally and that these intersections can be arbitrarily far away from Jupiter. 
Building on our results in  \cite{guardia2024oscillatorymotionsparabolicorbits}, where we prove an analogous result for the ejection and collision manifolds with $\mathcal S$, we prove the existence of transverse intersections of ejection and collision manifolds associated to $\mathcal S$ and $\mathcal J$ which can be arbitrarily far away from both primaries (see Section \ref{subsec: first part of theorem eco}). 
This proves the first two items of Theorem \ref{thm: ejection-collision orbits}.

\item 
Through the analysis of the invariant manifolds of collision with $\mathcal S$ done in \cite{guardia2024oscillatorymotionsparabolicorbits}, we prove that the ejection and collision manifolds associated to $\mathcal S$ and $\mathcal J$ intersect transversally inside the region of the phase space enclosed by Jupiter's orbit (see Section \ref{subsec: second part of thm eco}). 
This gives the third item of Theorem \ref{thm: ejection-collision orbits}.

\item 
 From the local analysis close to Jupiter (at the $C^1$ level) done in \cite{MR1335057} (which we recall in Section \ref{sec: local analysis at collision with J} below), we prove that the stable and unstable manifolds of infinity intersect transversally close to the collision set $\mathcal J$. 
Then, we construct hyperbolic sets with symbolic dynamics which contain the homoclinic points to infinity in its closure (but not containing Jupiter in its closure). 
This leads to oscillatory motions passing close to $\mathcal J$ (and combination of past and future different final motions), but not to oscillatory motions which have Jupiter at its closure, and it does not imply Theorems \ref{thm:chaosJ} and  \ref{thm: Existence of parabolic, oscillatory and periodic orbits J}.

\item 
To prove Theorems \ref{thm:chaosJ} and  \ref{thm: Existence of parabolic, oscillatory and periodic orbits J} we have to further analyze the invariant manifolds of infinity and the ejection/collision manifolds $\mathcal J^\mp$. That is, we prove that there exists an energy level for which the transverse intersection of the invariant manifolds of infinity belongs to one of the ejection fibers of $\mathcal J^-$. Then, making use of the local analysis close to $\mathcal J$ and the tools developed by Moser in \cite{moser2001stable}, one can construct the behaviors provided by Theorems \ref{thm:chaosJ} and  \ref{thm: Existence of parabolic, oscillatory and periodic orbits J} (see Section \ref{sec: Triple intersection}).
\end{enumerate}

\section{The Kepler problem}\label{sec: Unperturbed case mu=0}
We start by considering the Hamiltonian \eqref{eqn: Hamiltonian in synodical cartesian coordinates centered at CM J} with $\mu = 0$, which is the classical Kepler problem in rotating coordinates
\begin{equation}\label{eqn: Hamiltonian in cartesian synodical coordinates mu=0}
    \hat{ H}_0(\hat q,\hat{p}) = \frac{|\hat p|^2}{2} - (\hat q_1 \hat p_2 - \hat q_2 \hat p_1) - \frac{1}{|\hat q|}.
\end{equation}
This Hamiltonian has two first integrals: the angular momentum and the energy of the Asteroid in non-rotating coordinates centered at $\mathcal S$, defined respectively as
\begin{equation}\label{eqn: F0 and Theta in (hat q, hat p)}
    \hat \Theta(\hat q, \hat p) = \hat q_1 \hat p_2 - \hat q_2\hat p_1,\quad \mathrm H_0(\hat q,\hat p) = \frac{|\hat p|^2}{2} - \frac{1}{|\hat q|}.
\end{equation}
%
In this model, we are interested in orbits hitting $\mathcal{J}$, defined in \eqref{eqn: Collision sets} (for $\mu = 0$ they are regular orbits). 
Relying on \eqref{eqn: Hamiltonian in cartesian synodical coordinates mu=0}, for a fixed $\hat h\in \mathds R$, the collision set $\mathcal{J}$ at the energy level $\{\hat H_0 = \hat h\}$ becomes
\begin{equation*}
    \mathcal{J}_{\hat h} = 
    \left\{(\hat q,\hat p) \in \mathds{R}^4\colon \quad\hat q=(1,0), \quad \hat p_1^2 + (\hat p_2 - 1)^2 = 2\hat{h} + 3\right\}.
\end{equation*}
%
The following lemma, whose proof is straightforward,
gives a characterization of the ejection and collision orbits $\mathcal{J}^-$, $\mathcal{J}^+$ (see Definition \ref{def: EC-orbits}) in terms of $\hat{h}$ and $\hat{\Theta}$.
\begin{lemma}\label{lemma: Characterization of the keplerian orbits in terms of h and Theta}
     An orbit $\hat \gamma(t) \in \mathcal{J}^- \cup \mathcal{J}^+$ of the equations of motion associated to \eqref{eqn: Hamiltonian in cartesian synodical coordinates mu=0} in the hypersurface $\{\hat H_0(\hat q,\hat p) = \hat h\}$ belongs to one of the following three classes.
    \begin{itemize}
        \item Elliptic if one of the following conditions holds:
        \begin{itemize}
        \item $\hat{h} \in \left(-\frac 3 2, -\sqrt{2}\right)$ and $\hat \Theta \in \left[1-\sqrt{2\hat h+3},  1+\sqrt{2\hat h+3}\right]$.
        \item $\hat h \in \left[-\sqrt 2, \sqrt 2\right)$ and $\hat \Theta \in \left[1-\sqrt{2\hat h+3},-\hat h\right)$.
        \end{itemize}
        \item Parabolic if $\hat{h} \in \left[-\sqrt{2},\sqrt{2}\right]$ and $\hat{\Theta} = -\hat{h}$.
        \item Hyperbolic if $\hat{h}\in \left(-\sqrt{2},+\infty\right)$ and $\hat{\Theta} > -\hat{h}$.
    \end{itemize}
\end{lemma}

We study the parabolic orbits in Lemma \ref{lemma: Characterization of the keplerian orbits in terms of h and Theta}. It is convenient to work in the non-rotating polar coordinates defined by
%
%
\begin{equation}\label{eqn: Change from sidereal cartesian to sidereal polar centered at CM}
\begin{aligned}
    Q_1 &=  \check{r}\cos\check{\theta},&\quad  P_1 &= \check{R}\cos\check{\theta} - \frac{\check{\Theta}}{\check{r}}\sin\check{\theta},\\
    Q_2&= \check{r}\sin\check\theta,&\quad  P_2 &= \check{R}\sin\check{\theta} + \frac{\check{\Theta}}{\check{r}}\cos\check{\theta}.
\end{aligned}
\end{equation}
In these coordinates, the Hamiltonian $\mathrm H_0(Q,P)$ in \eqref{eqn: non autonomous Hamiltonian} becomes
\begin{equation}\label{eqn:Hamiltonian rotating polar coordinates for mu = 0 J}
    \check{\mathcal {H}}_0 (\check{r},\check{R},\check{\Theta}) = \frac 1 2\left(\check{R}^2 + \frac{\check{\Theta}^2}{\check{r}^2}\right) - \frac{1}{\check{r}}.
\end{equation}
From \eqref{eqn: Hamiltonian in cartesian synodical coordinates mu=0}, \eqref{eqn: F0 and Theta in (hat q, hat p)} and Lemma \ref{lemma: Characterization of the keplerian orbits in terms of h and Theta}, the parabolic  orbits crossing the orbit of Jupiter at a fixed angular momentum 
\begin{equation}\label{eqn:checkTheta-hatTheta}
\check \Theta=\check\Theta_0 = - \hat h \in \left[-\sqrt 2, \sqrt 2\right]
\end{equation}
are given by
\begin{equation}\label{eqn: set parabolic orbits passing through J}
\mathcal{V}_{\check \Theta_0} = \Big\{(\check r, \check \theta,\check R,\check \Theta) \in (0,+\infty)\times \mathds T \times \mathds R^2 \colon \check{\Theta} = \check \Theta_0\Big\} \cap \Big\{\check {\mathcal H}_0 =0\Big\}.
\end{equation}
The rest of the section is devoted to compute explicitly the parabolic  orbits in $\mathcal{V}_{\check \Theta_0}$ hitting Jupiter. We proceed as follows. First, we exploit the  integrability of $\hat{\mathcal H}_0$ in \eqref{eqn:Hamiltonian rotating polar coordinates for mu = 0 J} to compute the parabolic orbits in $\mathcal V_{\check \Theta_0}$ and to analyze its asymptotic behavior as $t\to \pm\infty$. Then we identify the orbits within this set that hit $\mathcal J$.

The equations of motion associated with the Hamiltonian $\check {\mathcal H}_0$ restricted to the plane $\left\{\check \Theta=  \check \Theta_0, \check{\mathcal{H}}_0 = 0\right\}$ are reduced to
\begin{equation}\label{eqn: Eq motion of F0}
    \begin{aligned}
        \frac{d}{dt}\check r =\; \check R,\qquad
        \frac{d}{dt}\check \theta =\; \frac{\check{\Theta}_0}{\check{r}^2},\qquad
        \frac{d}{dt}\check R =\; \frac{\check{\Theta}_0^2}{\check{r}^3} - \frac{1}{\check{r}^2}.
    \end{aligned}
\end{equation}
Next lemma (see \cite{MR0573346}) computes the trajectories of these equations.
\begin{lemma}\label{lemma: Reparameterization of time unperturbed solution}
Fix $\check \Theta_0 \in [-\sqrt 2, \sqrt 2]$ and let $w(t)$ be the unique analytic function defined by
\begin{equation}\label{eqn: tw}
    t = \frac 1 2\left(\frac 1 3 w^3 + \check\Theta_0^2 w\right),
\end{equation}
such that $w$ is real for real values of $t$, that is,
    \begin{equation}\label{eqn: wt}
    w(t) = \begin{cases} \sqrt[3]{6t}\quad &\text{if } \check \Theta_0 = 0,\\
    \left(3t + \sqrt{9t^2 + \check\Theta_0^6}\right)^{\frac 1 3} - \check\Theta_0^2\left(3t+ \sqrt{9t^2 + \check\Theta_0^6}\right)^{-\frac 1 3}\quad &\text{if } \check \Theta_0 \neq 0.
    \end{cases}
    \end{equation}
Then:
%
\begin{enumerate}
    \item[(i)] If $\check \Theta_0 \neq 0$, the parabolic orbit solution of \eqref{eqn: Eq motion of F0} with initial condition 
    \begin{equation}\label{eqn: initcondJupiter0}
    \check r_h(0,\check \Theta_0) = \frac{\check \Theta_0^2}{2},\quad \check \theta_h(0,\check \theta_0,\check \Theta_0) = \check \theta_0,\quad \check R_h(0,\check \Theta_0) = 0,
    \end{equation}
    is given by
    \begin{equation}\label{eqn: Unperturbed separatrix}
    \begin{aligned}
    \check\gamma_h(t,\check \theta_0,\check{\Theta}_0) &= \left(\check{r}_h(t,\check{\Theta}_0),\check{\theta}_h(t,\check \theta_0,\check{\Theta}_0),\check{R}_h(t,\check{\Theta}_0),\check{\Theta}_0\right) \\
    &= \left(\frac 1 2 \left(w(t)^2 + \check\Theta_0^2\right), \check \theta_0  -i\log\left(\frac{i\check\Theta_0-w(t)}{i\check\Theta_0+w(t)}\right),\frac{2w(t)}{w(t)^2 + \check\Theta_0^2},\check \Theta_0\right) ,\quad \;t\in \mathds R,
    \end{aligned}
    \end{equation}
    and satisfies $\check \gamma_h(t,\check \theta_0,\check \Theta_0) \subset \mathcal P^+ \cap \mathcal P^-$.
    
    \item[(ii)] If $\check \Theta_0 = 0$, $\check\gamma_h(t,\check \theta_0,0)$ above is no longer defined at $t=0$. Instead, there exist two trajectories of \eqref{eqn: Eq motion of F0} given by
    \begin{equation}\label{eqn: unperturbed sun}
        \begin{aligned}
            \check \gamma_h^+(t,\check \theta_0) &= \left(\check r_h^+(t), \check \theta_h^+(t,\check \theta_0), \check R_h^+(t), \check \Theta_h^+(t)\right) = \left(\lambda t^{\frac 2 3}, \check \theta_0, \sqrt{\frac{2}{\lambda}}t^{-\frac 1 3},0\right),\quad \forall\;t>0,\\
            \check \gamma_h^-(t,\check \theta_0) &= \left(\check r_h^-(t), \check \theta_h^-(t,\check \theta_0), \check R_h^-(t), \check \Theta_h^-(t)\right) = \left(\lambda t^{\frac 2 3}, \check \theta_0, -\sqrt{\frac{2}{\lambda}}|t|^{-\frac 1 3},0\right),\quad \forall\;t<0,\\
        \end{aligned}
    \end{equation}
    where $\lambda=\left(\frac 9 2\right)^{\frac 1 3 }$. They satisfy $\check \gamma_h^\pm(t,\check \theta_0)\subset \mathcal S^\mp \cap \mathcal P^\pm$ (see Definition \ref{def: EC-orbits}).
\end{enumerate}
\end{lemma}
%
Note that the singularities in \eqref{eqn: Unperturbed separatrix}, located at $w = \pm i\check\Theta_0$ (or equivalently at $t = \pm i\frac{\check\Theta_0^3}{3}$), are in fact zeroes of the function $\check r_h$ and therefore correspond to collisions of the parabolic orbit with the primary $\mathcal S$, which occur at purely complex values of time if $\check \Theta_0 \neq 0$.
To deal at the same time with the cases $\check \Theta_0\neq  0$ and $\check \Theta_0 = 0$, we introduce 
\begin{equation}\label{eqn: orbits non-rotating mu=0}
    \check \gamma_h^\pm(t,\check \theta_0,\check \Theta_0) = (\check r_h^\pm(t,\check\Theta_0), \check \theta_h^\pm(t,\check \theta_0,\check\Theta_0), \check R_h^\pm(t,\check\Theta_0), \check \Theta_0)=\begin{cases}\check \gamma_h(t,\check \theta_0,\check \Theta_0) \; \text{for} \pm  t \geq0\;\; \text{if }\check \Theta_0 \neq 0,\\ \check \gamma_h^\pm(t,\check\theta_0)\;\text{for} \pm t >0\;\;\text{if } \check \Theta_0=0.\end{cases}
\end{equation}
%
The following corollary provides the asymptotic behavior of the parabolic trajectories $\check \gamma_h^\pm(t,\check \theta_0,\check \Theta_0)$ as $t\to\pm\infty$. 
\begin{corollary}\label{corollary: Behaviour at infinity of the unperturbed separatrix}
    The trajectories $\check\gamma^\pm_h(t,\check \theta_0,\check{\Theta}_0)=\left(\check{r}_h^\pm(t,\check{\Theta}_0),\check{\theta}_h^\pm(t,\check \theta_0,\check{\Theta}_0),\check{R}^\pm_h(t,\check{\Theta}_0),\check{\Theta}_0\right)$ in \eqref{eqn: orbits non-rotating mu=0} satisfy
    \[\check{r}_h^\pm(t,\check\Theta_0) \sim |t|^{\frac 2 3},\quad \check{R}_h^\pm(t,\check\Theta_0) \sim \pm |t|^{-\frac 1 3},\quad \check{\theta}_h^\pm(t,\check \theta_0,\check\Theta_0)- \check{\theta}_0 \sim \check{\Theta}_0 |t|^{-\frac 1 3}\qquad \text{ as  }t\to \pm \infty.\]
\end{corollary}
Recall that for $\mu = 0$, the Hamiltonian \eqref{eqn: Hamiltonian in synodical cartesian coordinates centered at CM J} reduces to a Kepler problem involving only the Sun and the Asteroid. Hence, it is an abuse of language to say that the orbits $\check{\gamma}_h^\pm$ in \eqref{eqn: orbits non-rotating mu=0} collide with $\mathcal J$. To address this ambiguity, we say that these orbits ``hit Jupiter'' if there exist $(t_c,\check \theta_0) \in (0,+\infty)\times \mathds{T}$ such that 
\[
\check r_h^\pm(\pm t_c,\check\Theta_0) = 1,\quad \check \theta_h^\pm(\pm t_c,\pm \check \theta_0,\check\Theta_0) = \pm t_c.
\]
Imposing both conditions on \eqref{eqn: orbits non-rotating mu=0} yields
\begin{equation}\label{eqn: tc thetac hatTheta}
\begin{aligned}
    t_c = t_c(\check \Theta_0) = \frac 1 3\sqrt{2-\check \Theta_0^2} \left(1+\check\Theta_0^2\right),\quad\check \theta_c := \check \theta_0(\check \Theta_0) = t_c(\check \Theta_0) +i\log\left(\frac{i\check\Theta_0 - w( t_c(\check \Theta_0) )}{i\check\Theta_0 + w(t_c(\check \Theta_0))}\right).
\end{aligned}
\end{equation}
Then, the two trajectories
\begin{equation}\label{eqn: parabolic ejection and collision mu=0}
    \check \gamma_c^\pm (t,\check\Theta_0) := \check \gamma_h^\pm(t, \pm\check \theta_c,\check \Theta_0)
\end{equation} 
correspond to the parabolic orbits colliding with $\mathcal J$ at time $t=\pm t_c(\check \Theta_0)$, respectively.

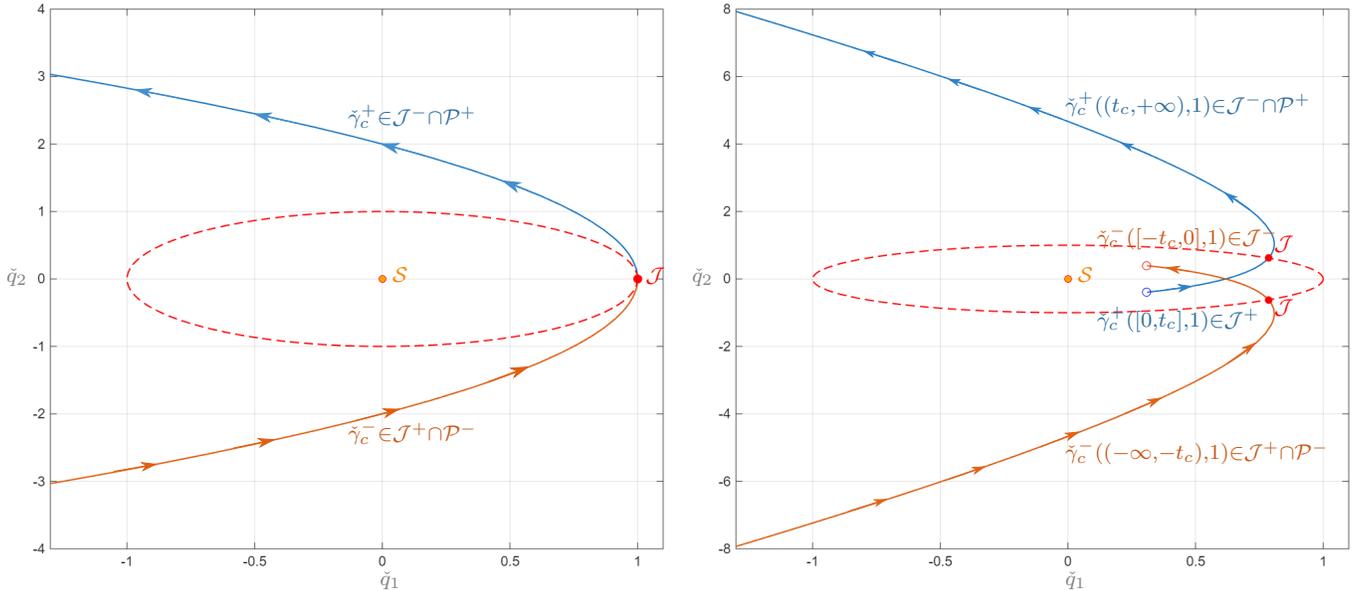
\begin{figure}[h!]
    \centering
    \begin{subfigure}{0.48\textwidth}
       \centering
        \begin{overpic}[width=\linewidth]{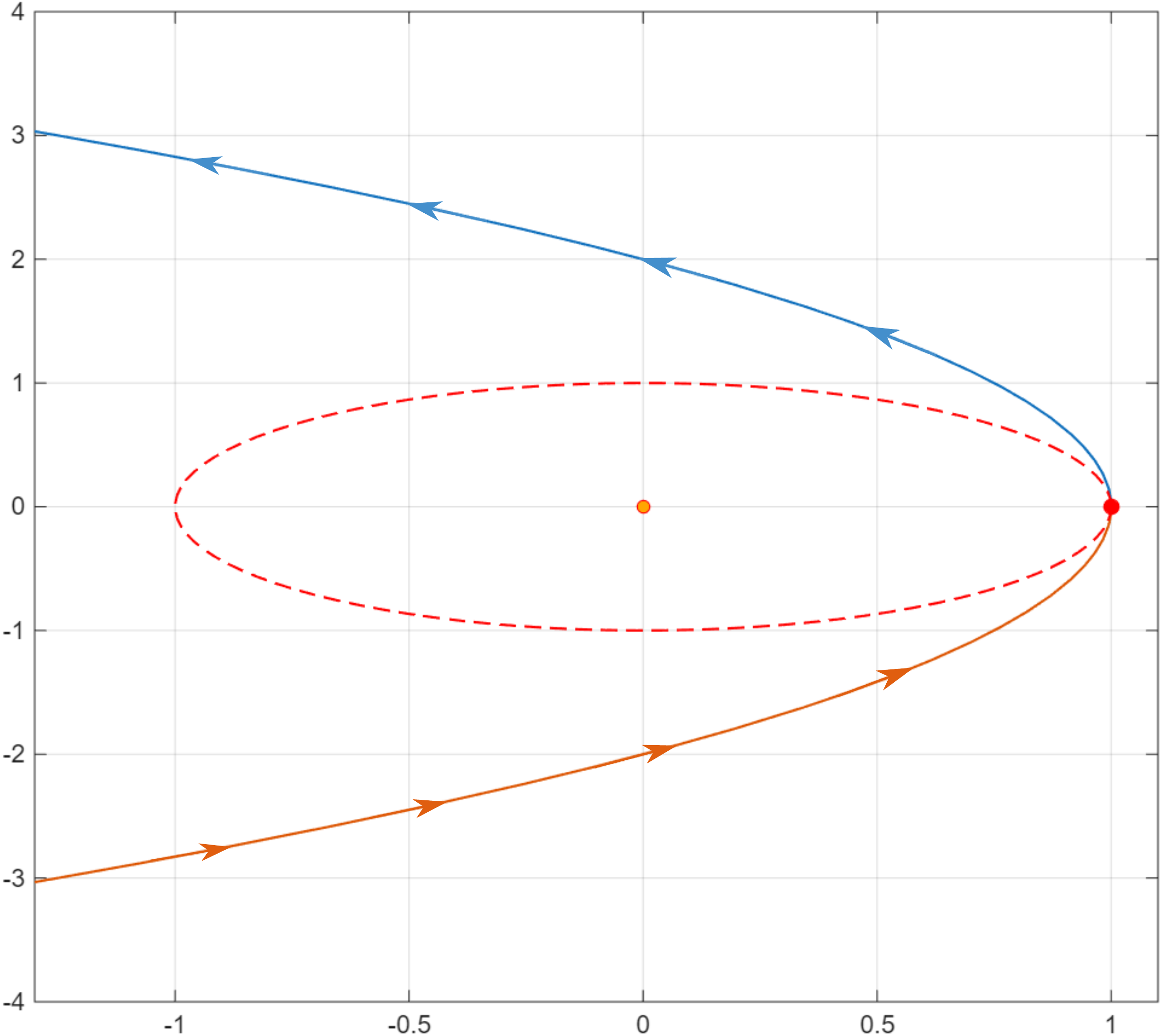}
             \put(50,70){\color{blueKepler}$\scriptstyle \check \gamma_c^+\in\mathcal J^-\cap \mathcal P^+$}
            \put(50,20){\color{redKepler}$\scriptstyle\check \gamma_c^-\in \mathcal J^+\cap \mathcal P^-$}
            \put(97,45){\color{red}$\scriptstyle \mathcal J$}
            \put(57,45){\color{orange}$\scriptstyle \mathcal S$}
            \put(55,-3){\color{grey}$\scriptstyle \check q_1$}
            \put(-4,45){\color{grey}$\scriptstyle \check q_2$}
        \end{overpic}
    \end{subfigure}
    \hfill
    \begin{subfigure}{0.48\textwidth}
        \centering
        \begin{overpic}[width=\linewidth]{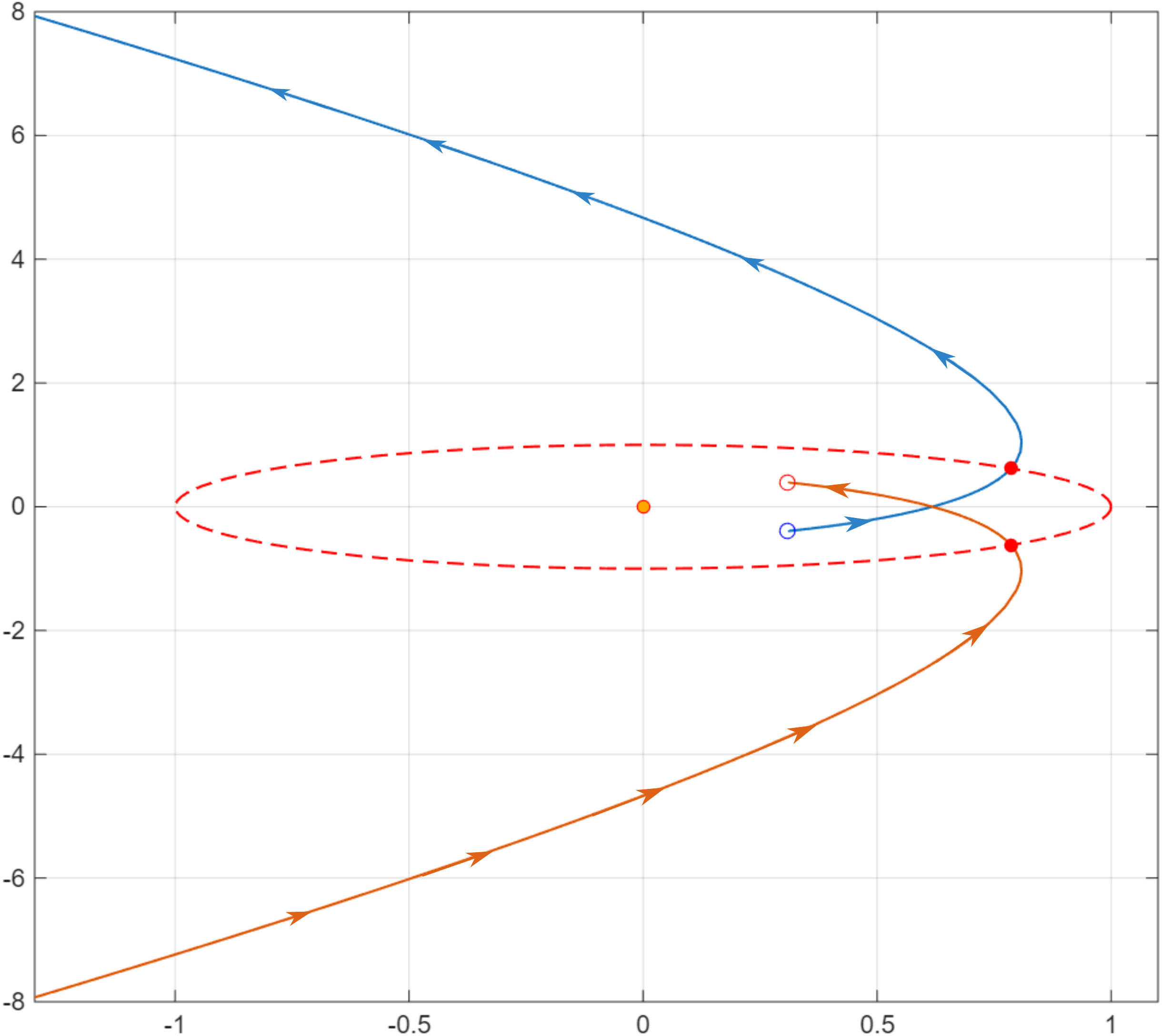}
            \put(60,38){\color{blueKepler}$\scriptstyle \check \gamma_c^+([0,t_c],1) \in \mathcal J^+$}
            \put(60,51){\color{redKepler} $\scriptstyle \check \gamma_c^-([-t_c,0],1) \in \mathcal J^-$}
            \put(55,72){\color{blueKepler}$\scriptstyle \check \gamma_c^+((t_c,+\infty),1)\in\mathcal J^-\cap \mathcal P^+$}
            \put(55,17){\color{redKepler}$\scriptstyle\check \gamma_c^-((-\infty,-t_c),1)\in \mathcal J^+\cap \mathcal P^-$}
            \put(88,50){\color{red}$\scriptstyle \mathcal J$}
            \put(88,40){\color{red}$\scriptstyle \mathcal J$}
            \put(57,45){\color{orange}$\scriptstyle \mathcal S$}
            \put(55,-3){\color{grey}$\scriptstyle \check q_1$}
            \put(-4,45){\color{grey}$\scriptstyle \check q_2$}
        \end{overpic}
    \end{subfigure}
    \caption{
    Examples of the orbits analyzed in Remark \ref{remark: parabolic EC orbits with J}. In the left  and right pictures we consider $\check \Theta_0^*=\sqrt2$ and $\check \Theta_0^* = 1$ respectively.
    }
    \label{fig:KeplerproblemECOparabolic}
\end{figure}

\begin{remark}\label{remark: parabolic EC orbits with J}
We make the following remark regarding the orbits $\check \gamma_c^\pm(t,\check\Theta_0)$ (see Figure \ref{fig:KeplerproblemECOparabolic}).
\begin{itemize}
    \item For $\check \Theta_0^* = \pm \sqrt 2$ we have $t_c(\pm\sqrt 2) = 0$ in \eqref{eqn: tc thetac hatTheta} so the orbits $\check \gamma_c^\pm(t,\check \Theta_0^*)$ correspond to parabolic ejection and collision orbits with $\mathcal J$, respectively. Namely
    \[\check \gamma_c^+(t,\check \Theta_0^*)\subset \mathcal J^-\cap \mathcal P^+, \quad \check \gamma_c^-(t,\check \Theta_0^*) \subset \mathcal J^+ \cap \mathcal P^-.\]
    \item For $\check \Theta_0^*\in (-\sqrt 2, \sqrt 2)\setminus\{0\}$ we have
    \begin{itemize}
        \item For $t\in [0,t_c]$, $\check \gamma_c^+(t,\check \Theta_0^*) \subset \mathcal J^+$ and for $t \geq t_c$, $\check \gamma_c^+(t,\check \Theta_0^*) \subset \mathcal J^-\cap \mathcal P^+$.
        \item For $t\in [-t_c,0]$, $\check \gamma_c^-(t,\check \Theta_0^*) \subset \mathcal J^-$ and for $t\leq -t_c$, $\check \gamma_c^-(t,\check \Theta_0^*) \subset \mathcal J^+\cap \mathcal P^-$.
    \end{itemize}
    \item For $\check \Theta_0^* = 0$ we have
    \begin{itemize}
        \item For $t\in (0,t_c]$, $\check \gamma_c^+(t,0) \subset \mathcal S^- \cap \mathcal J^+$ (and ballistic, see Remark \ref{def: large-ballistic}) and for $t\geq t_c$, $\check \gamma_c^+(t,0) \subset \mathcal J^-\cap \mathcal P^+$.
        \item For $t\in [-t_c,0)$, $\check \gamma_c^-(t,0) \subset \mathcal S^+ \cap \mathcal J^-$ (also ballistic) and for $t\leq -t_c$, $\check \gamma_c^-(t,0) \subset \mathcal J^+\cap \mathcal P^-$.
    \end{itemize}
\end{itemize}
\end{remark}

\section{Local analysis close to $\mathcal J$}\label{sec: local analysis at collision with J}
To carry out a local analysis around the collision set $\mathcal J$ in \eqref{eqn: Collision sets}, we first consider the translation to put $\mathcal J$ at the origin
\begin{alignat*}{3}
& q_1 &&= \hat q_1 - (1-\mu),   \quad   &&q_2 = \hat q_2,\\
& p_1 &&= \hat p_1,\quad  &&p_2 = \hat p_2-(1-\mu).
\end{alignat*}
In these coordinates, the Hamiltonian \eqref{eqn: Hamiltonian in synodical cartesian coordinates centered at CM J} reads
\begin{equation}\label{eqn: Hamiltonian in cartesian coordinates centered at P2}
	H_\mu(q,p) = \frac{|p|^2}{2} - \left( q_1p_2 - q_2 p_1\right) - q_1(1-\mu)  - \frac{1-\mu}{|q+(1,0)|} - \frac{\mu}{|q|} - \frac{(1-\mu)^2}{2},
\end{equation}
and $\mathcal J$ in \eqref{eqn: Collision sets} becomes $\{ q = 0\}$. For $q$ close to $0$, it has the expansion
\[H_\mu(q,p) =   \frac{|p|^2}{2} - \left( q_1p_2 - q_2 p_1\right) - (1-\mu)\left(q_1^2 - \frac{q_2^2}{2}\right) - \frac{\mu}{|q|} - (1-\mu)\left(1+\frac{(1-\mu)}{2}\right) + \mathcal O_3(q).\]
We consider the Hamiltonian
%
\begin{equation}\label{eqn: Hamiltonian + cte}
\begin{aligned}
    G_\mu (q,p) &= H_\mu(q,p)  + (1-\mu)\left(1+\frac{1-\mu}{2}\right) \\
    &= \frac{|p|^2}{2} - \left( q_1p_2 - q_2 p_1\right) - (1-\mu)\left(q_1^2 - \frac{q_2^2}{2}\right) - \frac{\mu}{|q|}+ \mathcal O_3(q).
\end{aligned}
\end{equation}
Both Lemma \ref{lemma: Characterization of the keplerian orbits in terms of h and Theta} and Hamiltonian \eqref{eqn: Hamiltonian + cte} prompt us to consider energy levels  $G_\mu=g$, where $ h \in (-\sqrt 2, \sqrt 2)$ and therefore $g$ satisfies
\begin{equation}\label{eqn: energy i}
    g =  h + (1-\mu)\left(1+\frac{1-\mu}{2}\right) > 0.
\end{equation}
%
The analysis performed in Sections \ref{subsec: Regularization of the system}
and \ref{subsec: Ejection and collision orbits from J} below  will be restricted to a neighborhood of $q = 0$ in the energy level $G_\mu^{-1}(g)$,
\begin{equation}\label{eqn: neighborhood of collision}
    B_{\gamma} = \left\{(q,p) \in G_\mu^{-1}(g) \colon |q| \leq \mu^\gamma \right\},
\end{equation}
where $\gamma\in (0,1)$ is independent of $\mu$. 

\begin{notation}
Throughout the paper we use several coordinate systems. For a coordinate $x$ and $k\in \mathds N$ we write
$f(x)=\mathcal O_k(x)$ if there exist  $C> 0$ independent of $\mu$,  such that there is $d>0$  and a function
$j\in C^k((-d,d))$ with  $\|j\|_{C^k((-d,d))}\leq C$, such that
$f(x)=x^k\,j(x)$.
\end{notation}

\subsection{Levi-Civita regularization of the collision   $\mathcal{J}$}\label{subsec: Regularization of the system}
Following \cite{MR3951693}, we perform the Levi-Civita transformation
\begin{equation}\label{eqn: Change of coordinates Levi-Civita}
\begin{aligned}
\psi\colon \mathds{R}^4 &\to \mathds{R}^4\\
(z_1,z_2,w_1,w_2) &\mapsto \left(2(z_1^2 - z_2^2), 4z_1z_2, \frac{w_1z_1- w_2z_2}{\xi |z|^2}, \frac{w_1z_2 + w_2z_1}{\xi |z|^2}\right),
\end{aligned}
\end{equation}
%
where $\xi$ is related to $g$  
in \eqref{eqn: energy i} as 
\begin{equation}\label{eqn: Energy xi}
\xi = (2g)^{-\frac 1 2}  =(2 h + 3)^{-\frac 1 2} + \mathcal{O}(\mu) > 0,
\end{equation}
to the system of equations associated to Hamiltonian \eqref{eqn: Hamiltonian + cte}.
Applying this change of coordinates and also  the time scaling
\[d\tau = \frac{dt}{\xi^2|z|^2}\]
to the Hamiltonian $G_\mu$ in \eqref{eqn: Hamiltonian + cte}, we obtain a new system which is Hamiltonian with respect to
\begin{equation}\label{eqn: Levi-Civita Hamiltonian in coord (z,w)}
\begin{aligned}
L_\mu(z,w) =&\; \xi^2|z|^2\left(\left(G_\mu-\frac{1}{2\xi^2}\right)\circ \psi\right)\\
=&\;\frac 1 2 \left(|w|^2-|z|^2\right) - \frac{\xi^2\mu}{2}- 2\xi |z|^2(z_1w_2-z_2w_1) - \xi^2|z|^2 f(z),
\end{aligned}
\end{equation}
where
\begin{equation}\label{eqn: f(z) LC}
f(z) = 4(1-\mu)\left((z_1^2-z_2^2)^2 - 2z_1^2z_2^2\right) +  \mathcal{O}_6(z).
\end{equation}
The orbits belonging to the hypersurface $\left\{G_\mu= 1/(2\xi^2)\right\}$, including the ejection and collision ones, now lie in $\left\{L_\mu(z,w) = 0\right\}$. In particular, in coordinates $(z,w)$, the collision manifold $\left\{|q|= 0\right\}$ restricted to the level set $\{L_\mu(z,w) = 0\}$ becomes the circle
\begin{equation}\label{eqn: Collision in coord (z,w)}
\begin{aligned}
    \mathcal{J}_h =& \left\{(z,w)\in L_\mu^{-1}(0)\colon z= 0\right\}=\left\{ (0,0,w_1,w_2) \colon w_1^2+w_2^2 = \xi^2\mu \right\}\\
    =& \left\{(0,0,\xi\mu^{\frac 1 2}\cos\beta,\xi\mu^{\frac 1 2}\sin\beta),\beta\in \mathds{T}\right\},
\end{aligned}
\end{equation}
and $B_\gamma$ defined in \eqref{eqn: neighborhood of collision} now becomes a neighborhood of $z=0$ in $L_\mu^{-1}(0)$ given by
\begin{equation}\label{eqn: neighborhood of collision (z,w)}
    \mathrm{B_\gamma} =\left\{(z,w) \in L_\mu^{-1}(0)\colon |z| \leq \frac{1}{\sqrt 2}\mu^{\frac{\gamma}{2}}\right\}.
\end{equation}
Note that $f(z)$ in \eqref{eqn: f(z) LC} is analytic in $\mathrm B_\gamma$.
%
%
The equations of motion associated to \eqref{eqn: Levi-Civita Hamiltonian in coord (z,w)} have a saddle at the origin with eigenvalues $\pm 1$ with multiplicity $2$. We denote by $W_\mu^s(0)$ and $W_\mu^u(0)$ its stable and unstable invariant manifolds respectively. The following proposition (see \cite{MR1335057}, $\S$ 3.3.1, p. 231) establishes a change of coordinates that straightens both invariant manifolds.

\begin{proposition}\label{prop: Straightening invariant manifolds collision LC}
    Fix $\epsilon > 0$ small enough and denote by $\mathrm B$ a $\epsilon$-neighborhood of the origin. 
    Then, there exists $\mu_0 >0$ such that, for any $0<\mu<\mu_0$, there exists an analytic change of variables
    \begin{equation}\label{eqn: straightening (z,w)-(s,u)}
        \begin{aligned}
            \Xi \colon \mathrm{B}\subset \mathds{R}^4 &\to \Xi(\mathrm B)\subset \mathds R^4\\
            (z,w) &\mapsto  (s, u)
        \end{aligned}
    \end{equation}
    satisfying
    \begin{equation*}
    \begin{aligned}
        s  = s(z,w) = \frac{z-w}{\sqrt 2}+\mathcal{O}_5(z, w),\quad u = u(z, w) =  \frac{z+w}{\sqrt 2} + \mathcal{O}_5(z, w),
    \end{aligned}
    \end{equation*}
   which transforms the equations of motion associated to Hamiltonian \eqref{eqn: Levi-Civita Hamiltonian in coord (z,w)} into equations of the form
    \begin{equation}\label{eqn: Eq motion in coordinates (s,u)}
    \begin{aligned}
       s' &= -s \left(1 + F^{s}(s,u)\right)\\
        u' &= u \left(1 + F^{u}( s,u)\right),
    \end{aligned}
    \end{equation}
    where $F^s$ and $F^u$ are analytic functions such that $F^{*}(s,u) = \mathcal{O}_2(s,u)$, for $*=s,u$.
    In these coordinates the saddle point remains as $(s,u) = (0,0)$ and its local invariant manifolds become
    \begin{equation*}
        \mathrm W_{\mu}^s(0) = \{u=0\},\quad \mathrm W_{\mu}^u(0) = \{s=0\}.
    \end{equation*}
    Moreover the Hamiltonian \eqref{eqn: Levi-Civita Hamiltonian in coord (z,w)} becomes the first integral
    \begin{equation}\label{eqn: energy in coordinates (s,u)}
    \begin{aligned}
        l_\mu(s,u) = L_\mu \circ \Xi^{-1}(s,u) =& -(s_1 u_1+s_2 u_2)  - \frac{\xi^2\mu}{2} \\
        &-\xi\left((s_1+ u_1)^2 + ( s_2+ u_2)^2\right)( s_1 u_2- s_2 u_1)
        +\mathcal{O}_6( s, u).
    \end{aligned}
    \end{equation}
\end{proposition}
In coordinates $(s,u)$, the collision circle defined in $\eqref{eqn: Collision in coord (z,w)}$ becomes
\begin{equation}\label{eqn: Collision in coord (s,u)}
\begin{aligned}
\mathcal{J}_h=\left\{(s,u)\in l_\mu^{-1}(0):\ (s,u)=\mathbf \Gamma(\beta),\ \beta\in\mathds{T}\right\},
\end{aligned}    
\end{equation}
where $\mathbf \Gamma(\beta)$ satisfies
\[
\mathbf \Gamma(\beta)
= \left(-\frac{\xi\mu^{\frac 1 2}}{\sqrt 2}(\cos\beta,\sin\beta), \frac{\xi \mu^{\frac 1 2}}{\sqrt 2}(\cos\beta,\sin\beta)\right) + \mathcal{O}_{C^1}\left(\mu^{\frac 5 2}\right).
\]
Note that $\mathrm{B}_\gamma \subset \mathrm{B}$ (where $\mathrm B_\gamma$ is defined in \eqref{eqn: neighborhood of collision (z,w)}). Hence, there exists a function $W$ satisfying 
\[
W(s,u)= (s_1+u_1)^2 +(s_2+u_2)^2 + \mathcal{O}_6(s,u)
\]
such that the  boundary of 
 the domain $\mathbf B_\gamma =\Xi(\mathrm B_\gamma)$ in \eqref{eqn: straightening (z,w)-(s,u)} is given by 
%
\begin{equation}\label{eqn: neighborhood of collision delta mu (s,u)}
 \partial \mathrm{B}_\gamma =\left\{(s,u)\in l_\mu^{-1}(0): W(s,u)=\mu^\gamma\right\}.
\end{equation}
\subsection{Ejection and collision orbits from $\mathcal J$}\label{subsec: Ejection and collision orbits from J}

As stated in Section \ref{subsec: main ideas for the proofs}, we compare the ejection and collision orbits from Jupiter with the invariant manifolds of infinity (see Section \ref{sec: The invariant manifolds of infinity}) in a common set of coordinates. We do the comparison in rotating polar coordinates centered at Jupiter, given by 
\begin{equation}\label{eqn: (q,p) - (r,theta,R,Theta)}
    \begin{aligned}
    \Upsilon\colon [0,+\infty) \times \mathds{T} \times \mathds{R}^2 &\to \mathds{R}^4\\
    (r,\theta,R,\Theta) &\mapsto (q_1,q_2,p_1,p_2).
    \end{aligned}
\end{equation}
%
%
In these coordinates, the Hamiltonian $H_\mu (q,p)$ in \eqref{eqn: Hamiltonian in cartesian coordinates centered at P2} reads
\begin{equation}\label{eqn: Hamiltonian in polar coordinates centered at J}
    \mathcal H_\mu(r,\theta,R,\Theta) = \frac 1 2 \left(R^2 + \frac{\Theta^2}{r^2}\right) - \Theta - r\cos\theta(1-\mu) - \frac{1-\mu}{\sqrt{r^2  +2r\cos\theta + 1}} - \frac{\mu}{r} - \frac{(1-\mu)^2}{2},
\end{equation}
which is reversible with respect to the symmetry
%
\begin{equation}\label{eqn: Symmetry in synodical polar coordinates centered at J}
    \left(r,\theta,R,\Theta;t\right) \to \left(r,-\theta,-R,\Theta;-t\right), 
\end{equation}
and $\partial \mathrm B_\gamma$ in \eqref{eqn: neighborhood of collision (z,w)} (see also  $\partial B_\gamma$ in  \eqref{eqn: neighborhood of collision}) becomes, for a fixed level set $\{\mathcal H_\mu = h\}$, the section
\begin{equation}\label{eqn: Section Sigma (r,theta,R,Theta)}
    \Sigma_\gamma= \left\{(r,\theta,R,\Theta) \in [0,+\infty)
    \times \mathds{T}\times \mathds{R}^2 \colon r= \mu^\gamma, \mathcal{ H}_\mu(\mu^\gamma,\theta,R,\Theta) =h\right\}.    
\end{equation}    
This section analyzes the ejection and collision orbits from Jupiter and their intersections with  $\Sigma_\gamma$. We proceed as follows. First, in coordinates $(s,u)$, we compute the ejection and collision orbits from the collision circle $\mathcal J_h$ in \eqref{eqn: Collision in coord (s,u)} and, in particular, their intersections with $\partial \mathbf B_\gamma$, defined in \eqref{eqn: neighborhood of collision delta mu (s,u)}. Then, we translate the result to coordinates $(r,\theta,R,\Theta)$.

Next lemma analyzes the dynamics of the collision circle $\mathcal J_h$ through the linearized part of the vector field in \eqref{eqn: Eq motion in coordinates (s,u)} until it reaches the section $\partial \mathbf B_\gamma$. To control the non-linear terms for the full system, we complexify $\beta$ in the parameterization of $\mathcal J_h$ in \eqref{eqn: Collision in coord (s,u)} on the complex strip
\begin{equation}\label{eqn: complexdomainbeta}
    \mathds{T}_{\sigma_0} = \{\beta \in \mathds{C}/(2\pi\mathds{Z})\colon |\mathrm{Im}\beta|\leq \sigma_0\},\qquad \sigma_0\in(0,1).
\end{equation}
\begin{lemma}\label{lemma: Linear ECO-orbits intersect section (z,w)}
Fix a closed interval $I \subset \left(-\sqrt{2}+\frac 3 2, \sqrt{2}+\frac 3 2\right)$, $\xi > 0$ with $1/(2\xi^2) \in I$, $\sigma_0\in (0,1)$ and $\gamma \in (0,1)$. Denote by $(s_{\mathrm{lin}}(\tau,\beta), u_{\mathrm{lin}}(\tau,\beta))$ the ejection (defined for $\mathrm{Re}(\tau) > 0$) and collision (defined for $\mathrm{Re}(\tau) <0$) trajectories of the linearization of \eqref{eqn: Eq motion in coordinates (s,u)} such that $(s_{\mathrm{lin}}(0,\beta),u_{\mathrm{lin}}(0,\beta)) = \mathbf \Gamma(\beta) \in \mathcal J_h$ (see \eqref{eqn: Collision in coord (s,u)}) with $\beta \in \mathds{T}_{\sigma_0}$ in \eqref{eqn: complexdomainbeta}. Then there exists $\mu_0 > 0$ such that, for $0<\mu<\mu_0$,  there exists a time $\tau_{\mathrm{lin}}(\beta)$ of the form

\begin{equation}\label{eqn: Time lineal Levi-Civita}
    \begin{aligned}
        \tau_{\mathrm{lin}}(\beta) = \mathrm{arcsinh}\left(\frac{\mu^{\frac{\gamma-1}{2}}}{\sqrt 2 \xi}\right) + \mathcal{O}_{C^1}\left(\mu^{\frac{3+\gamma}{2}}\right) = \log\left(\frac{\sqrt{2} \mu^\frac{\gamma-1}{2}}{\xi}\right) + \mathcal{O}_{C^1}\left(\mu^{1-\gamma}\right),
    \end{aligned}
\end{equation}
such that the orbits $(s_{\mathrm{lin}}, u_{\mathrm{lin}})$ satisfy $(s_{\mathrm{lin}}(\pm\tau_{\mathrm{lin}}(\beta),\beta), u_{\mathrm{lin}}(\pm\tau_{\mathrm{lin}}(\beta),\beta)) \in \partial \mathbf B_\gamma$ defined in \eqref{eqn: neighborhood of collision delta mu (s,u)}, forming two complex-analytic curves $\mathbf D_{\mathcal{J}}^\mp(\mu)\subset \partial \mathbf  B_\gamma$ of the form
\begin{equation*}
    \mathbf D_{\mathcal J}^\mp(\mu) = \left\{\left(s_{\mathrm{lin}}^\mp(\beta),u_{\mathrm{lin}}^{\mp}(\beta)\right)\colon \beta \in \mathds{T}_{\sigma_0}\right\} \subset \partial \mathbf B_\gamma
\end{equation*}
where
\begin{equation}\label{eqn: curve sulin}
\begin{aligned}
    s_{\mathrm{lin}}^\mp(\beta) &= -\frac{1}{2} \mu^{\frac \gamma 2}\left(\mp 1 + 1\right)(\cos\beta,\sin\beta) + \mathcal{O}_{C^1}\left(\mu^{1-\frac \gamma 2}\right),\\ u_{\mathrm{lin}}^\mp(\beta) &= \frac{1}{2} \mu^{\frac \gamma 2}\left(\pm 1 + 1\right)(\cos\beta,\sin\beta) + \mathcal{O}_{C^1}\left(\mu^{1-\frac \gamma 2}\right).
\end{aligned}
\end{equation}

\end{lemma}
\begin{proof}
    The proof of this lemma is a direct consequence of the integration of the linear part of system \eqref{eqn: Eq motion in coordinates (s,u)}, whose trajectory (with initial condition at a point in \eqref{eqn: Collision in coord (s,u)}) is of the form
    \begin{equation}\label{eqn: slin ulin}
    \begin{aligned}
    s_{\mathrm{lin}}(\tau,\beta) &= \left(-\frac{\xi\mu^{\frac 1 2}}{\sqrt 2}\left(\cos \beta, \sin\beta\right) + \mathcal{O}_{C^1}\left(\mu^{\frac 5 2}\right)\right)e^{-\tau},\\
    u_{\mathrm{lin}}(\tau,\beta) &=\left(\frac{\xi\mu^{\frac 1 2}}{\sqrt 2}\left(\cos \beta, \sin\beta\right) + \mathcal{O}_{C^1}\left(\mu^{\frac 5 2}\right)\right)e^{\tau}.
    \end{aligned}
    \end{equation}
\end{proof}
The next proposition shows that if one considers the full system \eqref{eqn: Eq motion in coordinates (s,u)}, the same is true up to a small error.
\begin{proposition}\label{prop: general ECO-orbits intersect section (z,w)}
Fix a closed interval $I \subset \left(-\sqrt{2}+\frac 3 2, \sqrt{2}+ \frac 3 2\right)$, $\xi > 0$ with $1/(2\xi^2)\in I$,  $\sigma_0 \in (0,1)$ and $\gamma \in \left(\frac{3}{11}, 1\right)$. Denote by $(s_\mu(\tau,\beta),u_\mu(\tau,\beta))$ the ejection (defined for $\mathrm{Re}(\tau) > 0$) and collision (defined for $\mathrm{Re}(\tau) < 0$) trajectories of \eqref{eqn: Eq motion in coordinates (s,u)} such that $(s_\mu(0,\beta),u_\mu(0,\beta)) = \mathbf \Gamma(\beta) \in \mathcal J_h$ (see \eqref{eqn: Collision in coord (s,u)}) with $\beta \in \mathds{T}_{\sigma_0/2}$ in \eqref{eqn: complexdomainbeta}. Then there exists $\mu_0 > 0$ such that, for $0<\mu<\mu_0$, there exists a time $T(\beta)$ satisfying
\begin{equation}\label{eqn: Estimate tilde T}
T(\beta) = \tau_{\mathrm{lin}}(\beta) + \mathcal{O}_{C^1}\left(\mu^{\frac{11\gamma-3}{8}}\right),
\end{equation}
such that $(s_\mu(\pm T(\beta),\beta), u_\mu(\pm T(\beta),\beta)) \in \partial \mathbf B_\gamma$ in \eqref
{eqn: neighborhood of collision delta mu (s,u)}, forming two complex-analytic curves $\mathbf  \Lambda_{\mathcal J}^\mp(\mu)\subset \partial \mathbf B_\gamma$ of the form
\begin{equation}\label{eqn: Intersect eco (s,u)-Sigma}
\begin{aligned}
\mathbf  \Lambda_{\mathcal J}^\mp(\mu) = \left\{\left(s^\mp(\beta), u^\mp(\beta)\right)\colon \beta \in \mathds{T}_{\sigma_0/2} \right\}\subset \partial \mathbf B_\gamma,
\end{aligned}
\end{equation}
where
\begin{equation}\label{eqn: sb ub}
    (s^\mp(\beta),u^\mp(\beta)) = (s_{\mathrm{lin}}^\mp(\beta), u_{\mathrm{lin}}^\mp(\beta)) + \mathcal{O}_{C^1}\left(\mu^{\frac{3(5\gamma-1)}{8}}, \mu^{1-\frac{\gamma }{2}}\right).
\end{equation}
%
\end{proposition}
The proof of this proposition is done in Appendix \ref{appendix: Proof of proposition 3.2}. The following proposition translates the curves $\mathbf \Lambda_\mathcal J^\mp(\mu)$ to polar coordinates $(r,\theta,R,\Theta)$.
\begin{proposition}\label{prop: Parameterization of the ECO orbits in (r,theta,R,Theta)}
Fix a closed interval $I \subset \left(-\sqrt{2}+ \frac 3 2, \sqrt{2}+ \frac 3 2\right)$,  $\xi > 0$ with $1/(2\xi^2) \in I$ and $\gamma \in \left(\frac{3}{11}, 1\right)$. Then there exists $\mu_0 > 0$ such that, for $0<\mu<\mu_0$, the curves $\mathbf \Lambda_{\mathcal J}^\mp (\mu)$ defined in \eqref{eqn: Intersect eco (s,u)-Sigma} for $\beta\in\mathds{T}$ are written, in coordinates $(r,\theta,R,\Theta)$, as graphs of the form
\begin{equation}\label{eqn: Curves ECO-section}
    \begin{aligned}
       \Lambda_{\mathcal J}^\mp(\mu) =& \left\{\left(r,\theta, R_{\mathcal J}^\mp(\theta), \Theta_{\mathcal J}^\mp(\theta)\right)\colon r = \mu^\gamma, \theta \in \mathds{T}\right\} \subset \Sigma_\gamma,
    \end{aligned}
\end{equation}
where $\Sigma_\gamma$ is defined in \eqref{eqn: Section Sigma (r,theta,R,Theta)} and both $R_{\mathcal J}^\mp$, $\Theta_{\mathcal J}^\mp$ are real-analytic functions of the form 
\begin{equation}\label{eqn: (R,Theta) in curves ECO-section}
    \begin{aligned}
        R_{\mathcal J}^\mp(\theta)= \pm \xi^{-1} + \mathcal{O}_{C^1}\left(\mu^{\frac{11\gamma-3}{8}},  \mu^{1-\frac \gamma 2}\right),\quad \Theta_{\mathcal J}^\mp(\theta) = \mathcal{O}_{C^1}\left(\mu^{\frac{19\gamma-3}{8}}, \mu\right).
    \end{aligned}
\end{equation}
\end{proposition}
\begin{proof}
We apply the transformations $\Xi^{-1}$ and $\psi$ in \eqref{eqn: straightening (z,w)-(s,u)} and \eqref{eqn: Change of coordinates Levi-Civita} respectively to translate the curves $\mathbf \Lambda_{\mathcal J}^\mp(\mu)$ in \eqref{eqn: Intersect eco (s,u)-Sigma} into coordinates $(q, p)$ yielding
\begin{equation}\label{eqn: ECO-points (q,p)}
    \begin{aligned}
        q^\mp(\beta) &= \mu^\gamma \left(\cos(2\beta),\sin(2\beta)\right) + \mathcal{O}_{C^1}\left(\mu^{\frac{19\gamma-3}{8}}, \mu\right),\\
        p^\mp(\beta) &= \pm \xi^{-1}\left(\cos(2\beta),\sin(2\beta)\right) + \mathcal{O}_{C^1}\left(\mu^{\frac{11\gamma-3}{8}},  \mu^{1-\frac \gamma 2}\right),
    \end{aligned}
\end{equation} 
for $\beta \in \mathds{T}$. Note that the points $q^\mp(\beta)$ satisfy $|q^\mp(\beta)| = \mu^\gamma $ since $\mathbf \Lambda_{\mathcal J}^\mp(\mu) \subset \partial \mathbf B_\gamma$ (see \eqref{eqn: neighborhood of collision delta mu (s,u)}).

Then, we apply the transformation $\Upsilon^{-1}$ in \eqref{eqn: (q,p) - (r,theta,R,Theta)} to obtain the components $(r,\theta)$ as $r=\mu^\gamma$ and
\begin{equation}\label{eqn: (r,theta) components ECO-curves}
    \begin{aligned}
        \theta^\mp(\beta) = \arctan\left(\frac{q_2^\mp(\beta)}{q_1^\mp(\beta)}\right) = 2\beta + \mathcal{O}_{C^1}\left(\mu^{\frac{11\gamma-3}{8}}, \mu^{1-\gamma}\right),
    \end{aligned}
\end{equation}
which can be inverted as 
\[\beta(\theta^\mp) = \frac{\theta^\mp}{2}  +\mathcal{O}_{C^1}\left(\mu^{\frac{11\gamma-3}{8}}, \mu^{1-\gamma}\right).\]
Relying on the change $\Upsilon^{-1}$ in \eqref{eqn: (q,p) - (r,theta,R,Theta)} we write the momenta $(R,\Theta)$ in terms of $\theta \in \mathds{T}$, leading to \eqref{eqn: (R,Theta) in curves ECO-section} and completing the proof.
\end{proof}

\begin{remark}\label{remark: Angle ECO-orbits}
We make the following remarks regarding this coordinate transformation and notation.
\begin{enumerate}
    \item The sign of the radial velocity $R$ of the curve $ \Lambda_{\mathcal{J}}^-(\mu)$ is positive, and negative for $ \Lambda_{\mathcal J}^+(\mu)$ (see \eqref{eqn: Curves ECO-section}). This prompts us to define the following sections
    \begin{equation}\label{eqn: Transverse sections}
        \begin{aligned}
            \Sigma^>_\gamma =& \left\{(r,\theta,R,\Theta)\colon r= \mu^\gamma, \mathcal{H}_\mu(\mu^\gamma, \theta,R, \Theta) = h, R > 0\right\} \subset \Sigma_\gamma,\\
            \Sigma^<_\gamma =& \left\{(r,\theta,R,\Theta)\colon r= \mu^\gamma, \mathcal{H}_\mu(\mu^\gamma, \theta,R, \Theta) = h, R  < 0\right\}\subset  \Sigma_\gamma,
        \end{aligned}
    \end{equation}
    where $\mathcal H_\mu$ is defined in \eqref{eqn: Hamiltonian in polar coordinates centered at J}. Then $\Lambda_{\mathcal{J}}^{-}(\mu) \subset  \Sigma_\gamma^>$ and $\Lambda_{\mathcal{J}}^{+}(\mu) \subset  \Sigma_\gamma^<$.
  
    \item Equation \eqref{eqn: (r,theta) components ECO-curves} relates the angles from the collision curve $ \Lambda_{\mathcal{J}}^+(\mu)$ with the ones from the ejection curve $ \Lambda_{\mathcal{J}}^-(\mu)$ as follows
    \begin{equation}\label{eqn: theta+-}
    \theta^+ = \theta^-  + \mathcal{O}_{C^1}\left(\mu^{\frac{11\gamma-3}{8}},\mu^{1-\gamma}\right).
    \end{equation}
\end{enumerate}
\end{remark}

\subsection{The transition map close to collision}

Let us decompose the boundary $\partial \mathbf B_\gamma = \boldsymbol{\Sigma}^- \cup \boldsymbol{\Sigma}^+$  (see \eqref{eqn: neighborhood of collision delta mu (s,u)}) where
\begin{equation}\label{eqn: bold sigma+-}
    \begin{aligned}
       \boldsymbol{\Sigma}^- = \{(s,u) \in \partial \mathbf{B}_\gamma \colon |s| \leq |u|\},\quad \boldsymbol{\Sigma}^+ = \{(s,u) \in \partial \mathbf{B}_\gamma \colon |s| \geq |u|\},
    \end{aligned}
\end{equation}
so that $\boldsymbol{\Lambda}_{\mathcal J}^-(\mu) \subset \boldsymbol{\Sigma}^-$ and $\boldsymbol{\Lambda}_{\mathcal J}^+(\mu) \subset \boldsymbol{\Sigma}^+$ (see \eqref{eqn: Intersect eco (s,u)-Sigma}). The following proposition defines a transition map from $\boldsymbol{\Sigma}^+$ to $\boldsymbol{\Sigma}^-$. It sends transverse curves to $\boldsymbol{\Lambda}_{\mathcal J}^+$ to transverse curves to $\boldsymbol{\Lambda}_{\mathcal J}^-$.
\begin{proposition}\label{prop: transition map close to collision}
    Fix $\beta_* \in \mathds{T}$, $\gamma \in \left(\frac{3}{11},1\right)$ and $\iota > 0$ small enough. There exists $\mu_0 > 0$ such that, for $0<\mu<\mu_0$, consider a curve $\boldsymbol{\gamma}_{\mathrm{in}} \subset \boldsymbol{\Sigma}^+$ of the form
    \[\boldsymbol{\gamma}_{\mathrm{in}}(\beta) = (s_{\mathrm{in}}(\beta),u_{\mathrm{in}}(\beta)),\quad \beta\in(\beta_*-\iota, \beta_*+\iota)\]
    where $s_{\mathrm{in}}(\beta),u_{\mathrm{in}}(\beta)$ are $C^1$-functions, which is transverse to $\mathbf{\Lambda}_{\mathcal J}^+(\mu) \subset \boldsymbol{\Sigma}^+$ at
    \[\mathbf{p}_+ = \boldsymbol{\gamma}_{\mathrm{in}}(\beta_*) = (s^+(\beta_*),u^+(\beta_*))\in \mathbf{\Lambda}^+_{\mathcal J}(\mu).
    \]
    Then, the flow associated to \eqref{eqn: Eq motion in coordinates (s,u)} induces a $C^1$ Poincar\'e map $\mathbf{f}\colon \boldsymbol{\Sigma}^+\to \boldsymbol{\Sigma}^-$ that maps the curve $\boldsymbol{\gamma}_{\mathrm{in}}$ to a curve $\boldsymbol{\gamma}_{\mathrm{out}}\subset \boldsymbol{\Sigma}^-$, parameterized as
    \begin{equation}\label{eqn: curve gammaout}
    \boldsymbol{\gamma}_{\mathrm{out}}(\beta) = \mathbf{f}(\boldsymbol{\gamma}_{\mathrm{in}}(\beta)) = \left(\frac{|u_{\mathrm{in}}(\beta)|}{|s_{\mathrm{in}}(\beta)|}\cdot s_{\mathrm{in}}(\beta), \frac{|s_{\mathrm{in}}(\beta)|}{|u_{\mathrm{in}}(\beta)|}\cdot u_{\mathrm{in}}(\beta)\right) + \mathcal{O}_{2}(|s_{\mathrm{in}}(\beta)|).
    \end{equation}
    Moreover, it is transverse to $\mathbf{\Lambda}_{\mathcal J}^-(\mu)\subset \boldsymbol{\Sigma}^-$ at
    \begin{equation}
        \mathbf{p}_- = \boldsymbol{\gamma}_{\mathrm{out}}(\beta_*)= \mathbf{f}(\mathbf p_+) =  (s^-(\beta_*),u^-(\beta_*)) \in \mathbf{\Lambda}_{\mathcal J}^-(\mu).
    \end{equation}

\end{proposition}
The proof of this proposition relies on the analysis carried out in \cite{MR1335057}, which is performed within the following neighborhood
\begin{equation}\label{eqn: neighborhood of collision (s,u)}
    \mathbf V_\varepsilon = \left\{(s,u)\in l_\mu^{-1}(0)\colon |s| \leq \varepsilon, |u|\leq \varepsilon\right\},
\end{equation}
where $\varepsilon > 0$ is considered small enough. The boundary $\partial \mathbf V_\varepsilon$ is the union of the submanifolds $C_s(\varepsilon)$ and $C_u(\varepsilon)$ defined as
\begin{equation}\label{eqn: Cs and Cu}
\begin{aligned}
    C_s(\varepsilon) = \left\{(s,u) \in l_\mu^{-1}(0) \colon |s|= \varepsilon,|u| \leq \varepsilon\right\},\quad
    C_u(\varepsilon) = \left\{(s,u)\in l_\mu^{-1}(0)\colon |s|\leq \varepsilon, |u| = \varepsilon\right\}.
\end{aligned}
\end{equation}
Using the results in \cite{MR1335057} ($\S$ 3.3.2, $\S$ 3.3.4, p. 232), the following lemma provides the transition map from $C_s$ to $C_u$. This transition map is a perturbation of the one given by the linearized part of the vector field in \eqref{eqn: Eq motion in coordinates (s,u)}.

\begin{lemma}\label{lemma: transition map close to J}
    There exists $\varepsilon_0 > 0$ such that, for any $0 < \varepsilon < \varepsilon_0$, the following $C^1$-transition map 
    \begin{equation}\label{eqn: transition map close to J}
    \begin{aligned}
        f_{\varepsilon}\colon C_s(\varepsilon) &\to C_u(\varepsilon)\\
        (s,u) &\mapsto \left(f^s_{\varepsilon}( s,u), f^u_{\varepsilon}( s,u)\right)
    \end{aligned}
    \end{equation}
    maps points $(s,u)\in C_s(\varepsilon)$ to points $\left (f^s_{\varepsilon}( s,u), f^u_{\varepsilon}( s,u)\right)\in C_u(\varepsilon)$ (see \eqref{eqn: Cs and Cu}) as follows
    \begin{equation*}
       f^s_{\varepsilon}(s,u) = \frac{| u|}{\varepsilon} s + \mathcal{O}_{C^1}(\varepsilon^2),\quad  f^u_{\varepsilon}(s,u)=\varepsilon\frac{ u}{| u|} + \mathcal{O}_{C^1}(\varepsilon^2).
    \end{equation*}
\end{lemma}
The rest of the proof consists on translating this result to the sections $\boldsymbol{\Sigma}^+$ and $\boldsymbol{\Sigma}^-$. Note that $\boldsymbol{\gamma}_{\mathrm{in}}\subset \boldsymbol{\Sigma}^+ \subset \partial\mathbf{B}_\gamma$ (see \eqref{eqn: neighborhood of collision delta mu (s,u)} and \eqref{eqn: bold sigma+-}), so $|s_{\mathrm{in}}(\beta)| \geq |u_{\mathrm{in}}(\beta)|$ with $|s_{\mathrm{in}}(\beta)| = \mathcal{O}(\mu^{\frac{\gamma}{2}})$. Hence, if we consider $\varepsilon(\beta)= |s_{\mathrm{in}}(\beta)|$, then $\boldsymbol{\gamma}_{\mathrm{in}}(\beta)\in C_s(\varepsilon(\beta))$ for $\beta \in B_{\iota}(\beta_*)$. Namely, the curve $\boldsymbol{\gamma}_{\mathrm{in}}$ satisfies
\[\boldsymbol{\gamma}_{\mathrm{in}} \subset \underset{\beta \in B_{\iota}(\beta_*)}{\bigcup}C_s(\varepsilon(\beta)).\]
Then, we apply the transition map $f_\varepsilon$ from Lemma \ref{lemma: transition map close to J} to the curve $\boldsymbol{\gamma}_{\mathrm{in}}$, which leads to the curve 
\[\boldsymbol{\gamma}_{\mathrm{out}} = \underset{\beta\in B_\iota(\beta_*)}{\bigcup} f_{\varepsilon(\beta)}(\boldsymbol{\gamma}_{\mathrm{in}}(\beta)) \subset \boldsymbol{\Sigma}^-,\]
yielding \eqref{eqn: curve gammaout} and completing the proof.

\section{Local analysis close to $\mathcal S$}\label{sec: local analysis close to collision with S}
To analyze the dynamics close to the collision set $\mathcal S$ in \eqref{eqn: Collision sets}, we follow the study  performed in \cite{guardia2024oscillatorymotionsparabolicorbits}. Consider the change to rotating polar coordinates centered at $\mathcal S$
%
\begin{equation}\label{eqn: Change from synodical cartesian to synodical polar centered at P1 J}
    \begin{aligned}
        \hat q_1 &= -\mu + \rS\cos\thetaS, &\quad\hat p_1 &= \RS\cos\thetaS - \frac{\ThetaS}{\rS}\sin\thetaS,\\
        \hat q_2&= \rS\sin\thetaS, &\quad \hat p_2 &= \RS\sin\thetaS + \frac{\ThetaS}{\rS}\cos\thetaS,
    \end{aligned}
\end{equation}
where $(\hat q, \hat p)$ are the cartesian coordinates centered at the center of mass. In these coordinates, the Hamiltonian $\hat H_\mu$ in \eqref{eqn: Hamiltonian in synodical cartesian coordinates centered at CM J} becomes
\begin{equation}\label{eqn: Hamiltonian function in rotating polar coordinates centered at S J}
\overline{\mathcal{H}}_\mu(\rS,\thetaS,\RS,\ThetaS) = \frac 1 2 \left(\RS^2 + \frac{\ThetaS^2}{\rS^2}\right) - \frac 1 \rS - \ThetaS - \overline V(\rS,\thetaS,\RS,\ThetaS;\mu),
\end{equation}
where
\begin{equation*}
    \overline V(\rS,\thetaS,\RS,\ThetaS;\mu) = -\mu \left(\frac 1 \rS + \RS\sin \thetaS + \frac{\ThetaS}{\rS}\cos \thetaS - \frac{1}{\sqrt{1+\rS^2-2\rS\cos\thetaS}}\right).
\end{equation*}
Moreover, this system is reversible with respect to the involution
%
\begin{equation}\label{eqn: Symmetry of the RPC3BP in synodical polar coordinates centered at P1 J}
    (\rS,\thetaS,\RS,\ThetaS) \to (\rS,-\thetaS,-\RS,\ThetaS)
\end{equation}
and the collision set $\mathcal{S}$ in \eqref{eqn: Collision sets} becomes $\{\rS = 0\}$.

To regularize \eqref{eqn: Hamiltonian function in rotating polar coordinates centered at S J} at $\rS = 0$ we perform the McGehee transformation \cite{MR0359459}
\begin{equation}\label{eqn:McGehee map Collision J}
    \begin{aligned}
    \overline \zeta \colon \mathds{R}^+ \times \mathds{T} \times \mathds{R}^2 &\to \mathds{R}^+ \times \mathds{T} \times \mathds{R}^2\\
    (\rS,\thetaS,v, u) &\mapsto (\rS,\thetaS,\RS, \ThetaS) = \left(\rS,\thetaS, v\rS^{-\frac 1 2} - \mu\sin\thetaS, u\rS^{\frac 1 2} + \rS^2 - \mu \rS \cos\thetaS\right)
    \end{aligned}
\end{equation}
and the change of time
\begin{equation*}
    dt = \rS^\frac 3 2 d\tau,
\end{equation*} 
so that the equations of motion associated to the Hamiltonian $\overline{\mathcal{H}}_\mu$ in \eqref{eqn: Hamiltonian function in rotating polar coordinates centered at S J} become
\begin{equation}\label{eqn:Motion regularized McGehee Collision J}
    \begin{aligned}
    \rS'&= \rS v\\
    \thetaS' &= u\\
    v'&= \frac{v^2}{2} + u^2 + 2u \rS^{\frac 3 2 } + \rS^3 - 1 + \mu \left[1-\rS^2\left(\cos \thetaS + \frac{\rS-\cos \thetaS}{(1+\rS^2-2\rS\cos\thetaS)^{\frac 3 2}}\right)\right]\\
    u' &= -\frac{u v}{2} -2v \rS^{\frac 3 2} + \mu \rS^2\sin\thetaS\left[1-\frac{1}{(1+\rS^2-2\rS\cos\theta)^{\frac 3 2}}\right],
    \end{aligned}
\end{equation}
where $\phantom{3}^{'}$ denotes $\frac{d}{d\tau}$. Observe that \eqref{eqn:Motion regularized McGehee Collision J} is now regular at $\rS=0$.

The change of variables in \eqref{eqn:McGehee map Collision J} is not symplectic but the Hamiltonian $\overline{\mathcal{H}}_\mu$ in \eqref{eqn: Hamiltonian function in rotating polar coordinates centered at S J} is still a first integral of \eqref{eqn:Motion regularized McGehee Collision J}. 
Moreover, the level set $\left\{\overline{\mathcal{H}}_\mu = h\right\}$ is now given by $(\overline{\mathcal{H}}_\mu-h) \circ \overline \zeta = 0$, where
\begin{equation*}
(\overline{\mathcal{H}}_\mu-h) \circ \overline \zeta(\rS,\thetaS,v,u) = - h+ \frac{v^2+u^2}{2\rS} -\frac{\rS^2}{2} -\frac{1-\mu}{\rS} + \mu\left[-\frac{\mu}{2} +\rS\cos\thetaS - \frac{1}{\sqrt{1+\rS^2-2\rS\cos\thetaS}}\right].
\end{equation*}
We now multiply by $\rS$ to remove the singularity, obtaining
\begin{equation*}
\begin{aligned}
    \overline{M}(\rS,\thetaS,v,u;\mu,h) =& -\rS h + \frac{v^2+u^2}{2} -\frac{\rS^3}{2} -1 + \mu + \mu \rS\left[-\frac \mu 2 + \rS\cos\thetaS - \frac{1}{\sqrt{1+\rS^2-2\rS\cos\thetaS}}\right].
\end{aligned}
\end{equation*}
We study \eqref{eqn:Motion regularized McGehee Collision J} restricted to the manifold  $\overline{M}(\rS,\thetaS,v,u;\mu,h)=0$ since the orbits belonging to the hypersurface $\left\{\overline{\mathcal{H}}_\mu(\rS,\thetaS,\RS,\ThetaS) = h\right\}$, including the ejection and collision ones, now lie in $\overline{M}(\rS,\thetaS,v,u;\mu,h) = 0$. It is convenient to introduce a last change of coordinates
\begin{equation}\label{eqn:polar change of coordinates for u,v into rho alpha J}
    \begin{aligned}
    \zeta\colon  \mathds{R}^+ \times \mathds{T}^2 \times \mathds{R}^+ &\to \mathds{R} \times \mathds{T}\times \mathds{R}^2 \\
    (s,\thetaS,\alpha,\rho) &\mapsto (\rS,\thetaS,v,u) = \left(s^2,\thetaS, \sqrt{2(1-\mu) +\rho}\sin\alpha, \sqrt{2(1-\mu) +\rho}\cos\alpha\right),
\end{aligned}
\end{equation}
such that $\left\{\overline{M}=0\right\}$ becomes
\begin{equation}\label{eqn:Relation rho with (r,theta,H,mu) J}
    0= M(s,\thetaS,\alpha, \rho;\mu,h) = -\rho+ 2s^2 h + s^6 -2\mu s^2\left[-\frac{\mu}{2} +s^2\cos\thetaS - \frac{1}{\sqrt{1+s^4-2s^2\cos\thetaS}}\right].
\end{equation}
Note that we have taken $\rS = s^2$ so the vector field \eqref{eqn:Motion regularized McGehee Collision J} in coordinates $(s,\thetaS,\alpha,\rho)$ is now of class $C^\infty$.    

To study the motion in coordinates $(s,\thetaS,\alpha,\rho)$, we define the $3$-dimensional submanifold
\begin{equation}\label{eqn: 3dimensional submanifold McGehee in coordinates (r,theta,rho,alpha) J}
    \mathcal{M} =\{(s,\thetaS,\rho,\alpha) \in \mathds{R}^+\times \mathds{T} \times \mathds{R}\times \mathds{T}  \colon M(s,\thetaS,\rho,\alpha;\mu,h) = 0\}.
\end{equation}
Using $(s,\thetaS,\alpha)$ as coordinates in $\mathcal{M}$ ($\rho$ can be obtained from \eqref{eqn:Relation rho with (r,theta,H,mu) J}) the collision manifold $\{\rS=0\}$ becomes the invariant torus
\begin{equation}\label{eqn: Collision manifold J}
    \Omega = \{(0,\thetaS,\alpha) \colon \thetaS  \in \mathds{T}, \alpha \in \mathds{T}\} \subset \mathcal{M}
\end{equation}
whose dynamics is given by
\begin{equation*}
\begin{aligned}
    \thetaS' = \sqrt{2(1-\mu)} \cos \alpha,\qquad \alpha' = \frac{\sqrt{2(1-\mu)}}{2} \cos \alpha.
\end{aligned}
\end{equation*}
This system has two circles of critical points
\begin{equation}\label{eqn: Circles of equilibrium points for mu neq 0 J}
    \begin{aligned}
        S^+= \left\{S_{\overline{\theta}}^+ = \left(0, \overline{\theta}, \frac \pi 2\right)\colon \overline{\theta} \in \mathds{T}\right\}, \;\;\;\;\; S^- = \left\{S_{\overline{\theta}}^- = \left(0,\overline{\theta}, -\frac{\pi}{2}\right)\colon \overline{\theta} \in \mathds{T}\right\}.
    \end{aligned}
\end{equation}
Next lemma, whose proof is done in \cite{guardia2024oscillatorymotionsparabolicorbits}, analyzes the stable and unstable invariant manifolds associated to these circles.

\begin{lemma}
The invariant circles $S^{\pm}$ in \eqref{eqn: Circles of equilibrium points for mu neq 0 J} are normally hyperbolic and they have $2$-dimensional stable and unstable manifolds $W_\mu^{u,s}(S^{\pm}) = \underset{\overline{\theta} \in \mathds{T}}{\bigcup}W_\mu^{u,s}(S_{\overline{\theta}}^\pm)$ such that

\begin{itemize}   
    \item $W_\mu^s(S^+)$ and $W_\mu^u(S^-)$ are contained in $\Omega$. Moreover, they coincide along a homoclinic manifold
    \begin{equation*}
        \mathrm M := W_\mu^s(S^+) = W_\mu^u(S^-) \subset \Omega.
    \end{equation*}
    Therefore
    \begin{equation*}
        \Omega = S^+ \cup S^- \cup \mathrm M,
    \end{equation*}
    and $\mathrm M$ is foliated by a family of heteroclinic orbits between $S_{\overline{\theta}}^-$ and $S_{\overline{\theta}}^+$, for $\overline{\theta}\in \mathds{T}$. 

    \item $W_\mu^s(S^-)$ and  $W_\mu^u(S^+)$ belong to $\mathcal{M}\setminus \Omega$. 
\end{itemize}
\end{lemma}

\begin{remark}\label{remark: definition of the invariant manifolds of collision in rotating polar coordinates centered at P1 J}

The definitions of $ W_\mu^s(S^-) $ and $W_\mu^u(S^+) $ can be translated to coordinates $(\rS,\thetaS,\RS,\ThetaS)$ by means of the changes \eqref{eqn:McGehee map Collision J} and \eqref{eqn:polar change of coordinates for u,v into rho alpha J}. Abusing the notation, we denote the collision and ejection manifolds as
\begin{equation}\label{eqn: definition of the invariant manifolds of collision in rotating polar coordinates centered at P1 J}
\begin{aligned}
     W_\mu^s(S^-)  =& \Bigg\{(\rS,\thetaS,\RS,\ThetaS) \in \mathds{R}^+ \times \mathds{T} \times \mathds{R}^2 \colon \exists t_* = t_*(\rS,\thetaS,\RS,\ThetaS) >0 \text{ such that }\\
    &\underset{t \to t_*^-}{\lim} \overline \Phi_t^{\rS}(\rS,\thetaS,\RS,\ThetaS) = 0,\underset{t \to t_*^-}{\lim}\overline \Phi_t^{\RS}(\rS,\thetaS,\RS,\ThetaS) = - \infty\Bigg\},\\
    &\\
    W_\mu^u(S^+)  =& \Bigg\{(\rS,\thetaS,\RS,\ThetaS) \in \mathds{R}^+ \times \mathds{T} \times \mathds{R}^2 \colon \exists t_* = t_*(\rS,\thetaS,\RS,\ThetaS)<0 \text{ such that }\\
    &\underset{t \to t_*^+}{\lim} \overline \Phi_t^{\rS}(\rS,\thetaS,\RS,\ThetaS) = 0, \underset{t \to t_*^+}{\lim}\overline \Phi_t^{\RS}(\rS,\thetaS,\RS,\ThetaS) = + \infty\Bigg\},
\end{aligned}
\end{equation}
where $\overline \Phi_t$ refers to the flow of the equations of motion associated to the Hamiltonian $\overline{\mathcal{H}}_\mu$ in \eqref{eqn: Hamiltonian function in rotating polar coordinates centered at S J}.

We stress that, although invariant (until hitting collision), they are not stable and unstable manifolds of any invariant objects since $S^+$ and $S^-$ collapse to the singular set $\{\rS=0\}$.

\end{remark}
To compute the ejection and collision orbits $\mathcal S^-$ and $\mathcal S^+$ (see Definition \ref{def: EC-orbits}) we recall the results obtained in \cite{guardia2024oscillatorymotionsparabolicorbits} ($\S$ 2.5, pp. 17-18), contained in the following proposition, which provides a parameterization in polar coordinates centered at $\mathcal S$ in \eqref{eqn: Change from synodical cartesian to synodical polar centered at P1 J} of the invariant manifolds of collision $W_\mu^u(S^+)$ and $W_\mu^s(S^-)$ (see Remark \ref{remark: definition of the invariant manifolds of collision in rotating polar coordinates centered at P1 J}).

\begin{proposition}\label{proposition: Perturbed invariant manifolds of collision in synodical polar coordinates J}
    Fix $a,b\in \mathds{R}$ and $\varrho \in[a,b]$. Then, there exist $\delta_0,\mu_0 > 0$ such that, for $0<\delta<\delta_0$, $0<\mu<\mu_0$ and $h=-\hat\Theta_0 = -\mu \varrho$, the invariant manifolds $W_\mu^u(S^+)$, $W_\mu^s(S^-)$ in \eqref{eqn: definition of the invariant manifolds of collision in rotating polar coordinates centered at P1 J}, written in polar coordinates centered at $\mathcal S$ (see \eqref{eqn: Change from synodical cartesian to synodical polar centered at P1 J}),  intersect the section 
    \begin{equation}\label{eqn: SigmahSun}
    \overline \Sigma = \{\rS=\delta^2, \overline{\mathcal{H}}_\mu(\delta^2,\thetaS,\RS,\ThetaS) = h\},
    \end{equation}
    where $\overline{\mathcal{H}}_\mu$ is the Hamiltonian \eqref{eqn: Hamiltonian function in rotating polar coordinates centered at S J}, at two curves $\overline{\Delta}_{S^+}^u(\mu)$ and $\overline{\Delta}_{S^-}^s(\mu)$ that can be written as graphs with respect to $\bar \theta$ as
    \begin{equation}\label{eqn: definition of the intersection of the invariant manifolds of collision with delta2 as graphs J}
    \begin{aligned}
        \overline{\Delta}_{S^+}^u(\mu) = \left\{(\thetaS,\ThetaS_{S^+}^u(\thetaS,\delta,\mu)), \thetaS\in\mathds{T}\right\},\quad \overline \Delta_{S^-}^s(\mu) = \left\{(\thetaS,\ThetaS_{S^-}^s(\thetaS,\delta,\mu)),\thetaS\in\mathds{T}\right\},
    \end{aligned}
    \end{equation}
    which depend smoothly on $\mu$. The expression for $\ThetaS_{S^+}^u(\thetaS)$ is given by
    \begin{equation}\label{eqn: Value of the angular momentum of Wu(S+) at r=2 J}
    \begin{aligned}
        \ThetaS_{S^+}^u(\thetaS,\delta,\mu) = \mu I_{S^+}^u(\thetaS) + \mathcal{O}(\mu^2)
    \end{aligned}
    \end{equation}
    with
    \begin{equation*}
        I_{S^+}^u\left(\thetaS\right) = \mathcal{I}_{S^+}^u\left(\thetaS + \frac{\sqrt{2}}{3} \delta^3\right)
    \end{equation*}
    where
    \begin{equation*}
    \begin{aligned}
        \mathcal{I}_{S^+}^u(\alpha) =\constantValue\bigintss_{\frac{\sqrt{2}}{3}\delta^3}^0 \frac{s^{\frac 2 3}\sin\left(\alpha - s\right)}{\left(1+\constantValue^2 s^{\frac 4 3} - 2\constantValue s^{\frac 2 3}\cos\left(\alpha - s\right)\right)^{\frac 3 2}}\;ds + \sqrt{\frac 2 \constantValue} \bigintss_{\frac{\sqrt{2}}{3}\delta^3}^0\frac{\cos\left(\alpha - s\right)}{s^{\frac 1 3}}\;ds
    \end{aligned}
    \end{equation*} 
    with $\lambda = (\frac{9}{2})^\frac {1}{ 3}$. 
    
    The expression for $\ThetaS_{S^-}^s(\thetaS,\delta,\mu)$ comes from the symmetry \eqref{eqn: Symmetry of the RPC3BP in synodical polar coordinates centered at P1 J} as
    \begin{equation}\label{eqn: Theta S-}
    \ThetaS_{S^-}^s(\thetaS,\delta,\mu) = \ThetaS_{S^+}^u(-\thetaS,\delta,\mu).
    \end{equation}
\end{proposition}
\begin{remark}
The sign of the radial velocity $\RS$ for the curve $\overline \Delta_{S^+}^u(\mu)$  is positive, and negative for $\overline \Delta_{S^-}^s(\mu)$. This prompts us to define domains in the section $\overline \Sigma$ in \eqref{eqn: SigmahSun}
\begin{equation}\label{eqn: sectionSun}
\begin{aligned}
    \overline \Sigma^> &= \left\{\left(\rS, \thetaS, \RS, \ThetaS\right)\colon  \rS = \delta^2, \overline{\mathcal H}_\mu(\delta^2,\thetaS,\RS,\ThetaS) =h,  \RS > 0\right\}\subset \overline\Sigma,\\  
    \overline \Sigma^< &= \left\{\left(\rS,\thetaS, \RS, \ThetaS\right)\colon \rS = \delta^2,\overline{\mathcal H}_\mu(\delta^2,\thetaS,\RS,\ThetaS) = h, \RS < 0\right\}\subset \overline \Sigma.
\end{aligned}
\end{equation}
where $\overline{\mathcal H}_\mu$ corresponds to the Hamiltonian \eqref{eqn: Hamiltonian function in rotating polar coordinates centered at S J}, so that $\overline \Delta_{S^+}^u(\mu) \subset \overline \Sigma^>$ and $\Delta_{S^-}^s(\mu) \subset \overline \Sigma^<$.
\end{remark}

\section{The invariant manifolds of infinity}\label{sec: The invariant manifolds of infinity}
To study the behavior 
of the infinity ``invariant set'' and the dynamics close to them, we introduce the so-called \textit{McGehee coordinates at infinity} (see for instance \cite{McGeheeInf}).
This will lead to the existence of the corresponding stable and unstable invariant manifolds.  In Section \ref{subsec: Parameterization of the invariant manifolds of infinity}, we provide a parameterization of these manifolds.

We express the Hamiltonian \eqref{eqn: Hamiltonian in synodical cartesian coordinates centered at CM J} in (synodical) polar coordinates centered at the center of mass $(\hat r, \hat \theta,\hat R,\hat \Theta)$, yielding
\begin{equation}\label{eqn: Hamiltonian Polar Rotating Coordinates centered at CM J}
    \hat{\mathcal H}_\mu(\hat r,\hat{\theta},\hat{R},\hat{\Theta}) = \frac{1}{2}\left(\hat{R}^2 + \frac{\hat{\Theta}^2}{\hat{r}^2}\right) - \frac {1}{\hat{r}} -\hat{\Theta} - \hat{V}(\hat{r},\hat{\theta};\mu),
\end{equation}
where
\begin{equation}\label{eqn: potential in polar rotating coordinates centered at CM J}
    \hat{V}(\hat{r},\hat{\theta};\mu) = \frac{1-\mu}{\left(\hat{r}^2 + 2\hat{r}\mu\cos\hat{\theta} + \mu^2\right)^{\frac{1}{2}}} + \frac{\mu}{\left(\hat{r}^2 - 2\hat{r}(1-\mu)\cos\hat{\theta} + (1-\mu)^2\right)^{\frac 1 2}} - \frac {1}{\hat{r}}.
\end{equation}
We consider the change of coordinates $\hat{r} = 2x^{-2}$, in which the parabolic infinity $\{\hat r = \infty, \hat R = 0\}$ becomes
\begin{equation}\label{eqn: Infinity set in mcgehee coordinates at infinity J}
    \Alpha = \left\{(x,\hat{\theta},\hat{R},\hat{\Theta}) \colon \mathds{R}^+\times \mathds{T}\times \mathds{R}^2\colon x = 0, \hat{R} = 0\right\},
\end{equation}
and the equations of motion associated to the Hamiltonian $\hat{\mathcal{H}}$ in \eqref{eqn: Hamiltonian Polar Rotating Coordinates centered at CM J} now read
\begin{equation}\label{eqn: Eq motion in McGehee coordinates of infinity J}
    \begin{aligned}
        \frac{dx}{dt} =-\frac{\hat{R}x^3}{4},\quad \frac{d\hat{\theta}}{dt} = \frac{\hat{\Theta}}{4}x^4 - 1,\quad \frac{d\hat{R}}{dt} = - \frac{x^4}{4} + \frac{\hat{\Theta}^2}{8}x^6 + \frac{x^3}{4}\partial_{x}\hat V(x,\hat{\theta};\mu),\quad
        \frac{d\hat{\Theta}}{dt} = \partial_{\hat{\theta}} \hat V(x,\hat{\theta};\mu),
    \end{aligned}
\end{equation}
where
\begin{equation*}
    \hat{V}(x,\hat{\theta};\mu) = \frac{x^2}{2}\left(\frac{1-\mu}{\left(1+x^2\mu\cos\hat{\theta} + \frac{x^4}{4}\mu^2\right)^{\frac 1 2}} + \frac{\mu}{\left(1- x^2(1-\mu)\cos\hat{\theta} + \frac{x^4}{4}(1-\mu)^2\right)^{\frac 1 2}}-1\right).
\end{equation*}
From \eqref{eqn: Eq motion in McGehee coordinates of infinity J}, one obtains that the manifold $\Alpha$ in \eqref{eqn: Infinity set in mcgehee coordinates at infinity J} is foliated by periodic orbits as $\Alpha= \underset{\hat{\Theta}_0 \in \mathds{R}}{\bigcup}\Alpha_{\hat{\Theta}_0}$ with
\begin{equation*}
    \Alpha_{\hat{\Theta}_0} = \left\{(x,\hat{\theta},\hat{R},\hat{\Theta}) 
    \colon x = 0, \hat{R} = 0, \hat{\Theta}=\hat{\Theta}_0, \, \hat \theta \in \mathds{T} \right\}.
\end{equation*}
In \cite{MR0362403} it was proven that these periodic orbits have stable and unstable manifolds, which we denote by $W_\mu^s(\Alpha_{\hat{\Theta}_0})$ and $W_\mu^u(\Alpha_{\hat{\Theta}_0})$ respectively. Moreover, these manifolds depend analytically on $\hat\Theta_0$. Note that the rates of convergence of the invariant manifolds $W_\mu^{s,u}(\Alpha_{\hat{\Theta}_0})$ are polynomial in $t$ and not exponential as in the case of  hyperbolic objects. For this reason, in \cite{MR3927089} and \cite{MR3455155}, the set $\Alpha$ in \eqref{eqn: Infinity set in mcgehee coordinates at infinity J} is referred to as a ``normally parabolic'' invariant manifold.

\subsection{Parameterization of the invariant manifolds of infinity}\label{subsec: Parameterization of the invariant manifolds of infinity}

For $\mu = 0$, Lemma \ref{lemma: Characterization of the keplerian orbits in terms of h and Theta} shows that the parabolic ejection and collision orbits with $\mathcal{J}$ lie in the plane $\{\hat H_0 = \hat h , \hat{\Theta}_0 = -\hat h\}$ for $\hat h \in [-\sqrt{2},\sqrt{2}]$ and $\hat\Theta_0 = \check \Theta_0$ (see \eqref{eqn:checkTheta-hatTheta}), where $\hat H_0$ corresponds to the Hamiltonian 
\eqref{eqn: Hamiltonian in cartesian synodical coordinates mu=0}. Lemma \ref{lemma: Reparameterization of time unperturbed solution} (and Remark \ref{remark: parabolic EC orbits with J}) analyzes these orbits and thus, the parabolic ejection trajectories $\check\gamma^{+}_c(t,\hat\Theta_0)$ in \eqref{eqn: parabolic ejection and collision mu=0} (defined for $t\geq t_c$ with $t_c$ as in \eqref{eqn: tc thetac hatTheta}) provide a parametrization in the non-rotating polar coordinates $(\check r,\check \theta,\check R,\check \Theta)$ (see \eqref{eqn: Change from sidereal cartesian to sidereal polar centered at CM}) of the stable manifold $W_0^{s}(\Alpha_{\hat\Theta_0})$. Analogously, the parabolic collision trajectories $\check\gamma^{-}_c(t,\hat\Theta_0)$ (defined for $t\leq-t_c$) parameterize the unstable manifold $W_0^{u}(\Alpha_{\hat\Theta_0})$.

The following lemma (whose proof is straightforward) gives a graph parameterization of the invariant manifolds $W_0^{s,u}(\Alpha_{\hat{\Theta}_0})$ in coordinates $(\hat r, \hat \theta,\hat R,\hat \Theta)$.
\begin{lemma}\label{lemma: Invariant manifolds of infinity as graphs for mu=0}
    For $\mu = 0$ and for any $\hat{\Theta}_0\in \left[-\sqrt{2},\sqrt{2}\right]$, the invariant manifolds of infinity $W_0^{s,u}(\Alpha_{\hat\Theta_0})$ can be written as graphs of the form
    \begin{equation}
        W_0^{s,u}(\Alpha_{\hat\Theta_0}) = \Bigg\{\left(\hat{r},\hat{\theta}, \hat{R}_0^{s,u}(\hat{r},\hat{\Theta}_0), \hat{\Theta}_0^{s,u}(\hat{r},\hat{\Theta}_0)\right) \colon  (\hat{r},\hat{\theta}) \in \left(\hat{\Theta}_0^2/2,+\infty\right)\times \mathds{T}\Bigg\},
    \end{equation}
    such that
    \begin{equation}\label{eqn: (R,Theta) infty mu=0}
    \begin{aligned}
        \hat{R}_0^s(\hat{r};\hat{\Theta}_0) &= \frac{1}{\hat{r}} \sqrt{2\hat{r} - \hat{\Theta}_0^2},&\quad \hat{\Theta}_0^s(\hat{r},\hat{\Theta}_0) &= \hat{\Theta}_0,\\
        \hat{R}_0^u(\hat{r};\hat{\Theta}_0) &= -\frac{1}{\hat{r}} \sqrt{2\hat{r} - \hat{\Theta}_0^2},&\quad \hat{\Theta}_0^u(\hat{r},\hat{\Theta}_0) &= \hat{\Theta}_0.\\
    \end{aligned}
    \end{equation}
\end{lemma}
The following proposition shows that, for $\mu > 0$ small enough, the invariant manifolds $W_\mu^{s,u}(\Alpha_{\hat \Theta_0})$ can be extended to reach a neighborhood of $\mathcal J$, and provides a graph parameterization for them in coordinates $(\hat r, \hat \theta, \hat R,\hat \Theta)$.

\begin{proposition}\label{prop: Parameterization of the invariant manifolds of infinity}
 Fix  $\nu \in (0,\frac 1 3)$, $m > 0$, $\hat\Theta_0 \in [-\sqrt 2 + m , \sqrt 2- m]$ and $\kappa > 0$. There exists $\mu_0 > 0$ such that, for $0<\mu<\mu_0$, the invariant manifolds $W_\mu^{s,u}(\Alpha_{\hat{\Theta}_0})$ associated to the Hamiltonian $\mathcal{\hat H}_\mu$ in \eqref{eqn: Hamiltonian Polar Rotating Coordinates centered at CM J} can be written as graphs of the form
    \begin{equation}\label{eqn: invariant manifolds of infty as graphs in terms of hat theta}
        W_\mu^{s,u}(\Alpha_{\hat\Theta_0}) = \Bigg\{\left(\hat r, \hat \theta, \hat R_\infty^{s,u}(\hat r, \hat \theta,\hat\Theta_0), \hat{\Theta}_\infty^{s,u}(\hat r, \hat \theta, \hat \Theta_0)\right)\colon (\hat r, \hat \theta) \in [1-\mu+\kappa\mu^\nu,+\infty)\times \mathds{T}\Bigg\}
    \end{equation}
    such that
    \begin{equation}\label{eqn: (R,Theta) as graphs inf muneq0}
        \begin{aligned}
            \hat R_\infty^{s,u}(\hat r, \hat \theta,\hat\Theta_0) = \hat R_0^{s,u}(\hat r,\hat\Theta_0) + \mathcal{O}(\mu^{1-2\nu}),\qquad \hat \Theta_\infty^{s,u}(\hat r, \hat \theta,\hat\Theta_0) = \hat \Theta_0 + \mathcal{O}(\mu^{1-2\nu}),
        \end{aligned}
    \end{equation}
    where $\hat R_0^{s,u}(\hat r,\hat \Theta_0)$ are defined in \eqref{eqn: (R,Theta) infty mu=0}. Moreover
    \begin{equation}
        \partial_{\hat\theta} \hat R_\infty^{s,u}(\hat r, \hat \theta,\hat\Theta_0) =\mathcal{O}(\mu^{1-3\nu}), \quad  \partial_{\hat\theta} \hat \Theta_\infty^{s,u}(\hat r, \hat \theta,\hat\Theta_0) = \mathcal{O}(\mu^{1-3\nu}).
    \end{equation}
\end{proposition}
The proof of this proposition is done in Appendix \ref{appendix: proposition 5.2}.

To compare the invariant manifolds of infinity with the ejection and collision curves $ \Lambda_{\mathcal J}^+(\mu)$ and $ \Lambda_{\mathcal J}^-(\mu)$ in \eqref{eqn: Curves ECO-section}, we express the invariant manifolds $W_\mu^{s,u}(\Alpha_{\hat \Theta_0})$ in \eqref{eqn: invariant manifolds of infty as graphs in terms of hat theta} into rotating polar coordinates centered at Jupiter (see \eqref{eqn: (q,p) - (r,theta,R,Theta)}) and we prove that they intersect the section $\Sigma_\nu$ in \eqref{eqn: Section Sigma (r,theta,R,Theta)}. 
\begin{proposition}\label{prop: Parameterization of the invariant manifolds of infinity in (r,theta,R,Theta)}
    Fix $\nu \in (0,\frac 1 3)$, $\hat \Theta_0 \in \left(\frac{1-\sqrt{3}}{2}, \frac{1+\sqrt{3}}{2}\right)$ and $h = - \hat \Theta_0$. Then there exists $\mu_0> 0$ such that, for $0<\mu<\mu_0$ , the invariant manifold $W_\mu^{u}(\Alpha_{\hat\Theta_0})$ intersects the section $ \Sigma_\nu^<$ in \eqref{eqn: Transverse sections} in a curve that can be written in coordinates $(r,\theta,R,\Theta)$ as a graph of the form
    \begin{equation}\label{eqn: curve unstable inf cap r=mudelta2}
    \begin{aligned}
    \Lambda_\infty^u(\hat \Theta_0, \mu) &= \left\{\left(r, \theta, R_\infty^u(\theta,\hat{\Theta}_0), \Theta_\infty^u(\theta,\hat{\Theta}_0)\right)\colon r= \mu^\nu, \theta \in \left(-\frac \pi 4, \frac \pi 4\right)\right\} = W_\mu^u(\Alpha_{\hat\Theta_0})\;\bigcap\;  \Sigma_\nu^<,
    \end{aligned}
    \end{equation}
    where $R_\infty^u$ and $\Theta_\infty^u$ are of the form
    \begin{alignat}{3}\label{eqn: R_infu and Theta_infu}
    &R_\infty^u(\theta,\hat{\Theta}_0) &&= f(\theta,\hat \Theta_0) + \mathcal{O}(\mu^\nu),\quad &&\partial_{\theta}R_\infty^u(\theta,\hat\Theta_0) = \partial_\theta f(\theta,\hat \Theta_0) + \mathcal{O}(\mu^{1-3\nu}),\\
    &\Theta_\infty^u(\theta,\hat{\Theta}_0) &&= \mu^\nu\partial_\theta f(\theta,\hat\Theta_0) + \mathcal{O}(\mu^{2\nu}),\quad &&\partial_{\theta}\Theta_\infty^u(\theta,\hat\Theta_0) =\mu^\nu\partial^2_{\theta}f(\theta,\hat\Theta_0) + \mathcal{O}(\mu^{1-2\nu}),
    \end{alignat}
    and $f(\theta,\hat\Theta_0) = -\cos\theta \sqrt{2-\hat{\Theta}_0^2}+ \sin\theta\left(\hat{\Theta}_0 - 1\right) < 0$.
    
    Moreover, due to the symmetry \eqref{eqn: Symmetry in synodical polar coordinates centered at J}, the curve $\Lambda_\infty^s(\hat\Theta_0,\mu) = W_\mu^s(\Alpha_{\hat\Theta_0})\cap \Sigma_\nu^>$ can be written as
    \begin{equation}\label{eqn: curve stable inf cap r=mudelta2}
        \begin{aligned}
        \Lambda_\infty^s(\hat \Theta_0, \mu) &= \left\{\left(r, \theta, R_\infty^s(\theta,\hat{\Theta}_0), \Theta_\infty^s(\theta,\hat{\Theta}_0)\right)\colon r= \mu^\nu, \theta \in \left(-\frac \pi 4, \frac \pi 4\right)\right\},
        \end{aligned}
    \end{equation}
    where
    \begin{equation}\label{eqn: R_infs and Theta_infs}
    \begin{aligned}
        R_\infty^s(\theta,\hat{\Theta}_0) =  -R_\infty^u(-\theta,\hat{\Theta}_0),\qquad \Theta_\infty^s(\theta,\hat{\Theta}_0) = \Theta_\infty^u(-\theta,\hat{\Theta}_0).
    \end{aligned}
    \end{equation}
\end{proposition}
\begin{proof}
The first step is to identify, for a $\kappa>0$ on an interval  to be determined,  the points in $W_\mu^u(\Alpha_{\hat\Theta_0})\cap \{\hat r  = 1-\mu+\kappa\mu^\nu\}$ that also lie on the section $\Sigma_\nu$ defined in \eqref{eqn: Section Sigma (r,theta,R,Theta)}. We do it in cartesian coordinates $(q,p)$, introduced in \eqref{eqn: (q,p) - (r,theta,R,Theta)}, so we obtain two points (which depend on $\kappa$ and $\mu$) satisfying $\hat r = 1-\mu+\kappa\mu^\nu$ and $r = \mu^\nu$. Namely
\begin{equation*}
    \begin{aligned}
        q_1^2 + q_2^2 = r^2 = \mu^{2\nu},\quad \left(q_1+(1-\mu)\right)^2 + q_2^2 = \hat r^2 =  (1-\mu+\kappa\mu^\nu)^2.
    \end{aligned}
\end{equation*}
Solving for $(q_1,q_2)$ gives two points
\begin{equation*}
\begin{aligned}
    q_1(\kappa,\mu) = \mu^\nu\left(\kappa - \mu^\nu \frac{1-\kappa^2}{2(1-\mu)}\right),\quad q_2^\pm(\kappa,\mu) = \pm \mu^\nu \sqrt{1-\kappa^2} \cdot\sqrt{\frac{1-\mu+\kappa\mu^\nu}{1-\mu} - \mu^{2\nu}\frac{1-\kappa^2}{4(1-\mu)^2}},
\end{aligned}
\end{equation*}
whose arguments are respectively given by
\begin{equation*}
    \theta^\pm(\kappa,\mu) = \pm \arctan\left(\frac{\sqrt{1-\kappa^2}}{\kappa - \mu^\nu \frac{1-\kappa^2}{2(1-\mu)}}\sqrt{\frac{1-\mu+\kappa\mu^\nu}{1-\mu} - \mu^{2\nu}\frac{1-\kappa^2}{4(1-\mu)^2}}\right)
\end{equation*}
These arguments are well defined for any $\kappa \in\left(\kappa_0(\mu),1\right) \subset(0,1)$ with 
\[\kappa_0(\mu) := \frac{\sqrt{(1-\mu)^2+\mu^{2\nu}}- (1-\mu)}{\mu^\nu}=\frac{1}{2}\mu^\nu+\mathcal{O}(\mu^{3\nu}).\]
The derivative satisfies
\begin{equation*}
    \partial_{\kappa}\theta^\pm(\kappa,\mu)
    = \mp \frac{1-\mu+\kappa\mu^\nu}{\sqrt{1-\kappa^2}\sqrt{(1-\mu)^2 + \kappa\mu^\nu - \kappa\mu^{1+\nu} - \frac{\mu^{2\nu}}{4}(1-\kappa^2)}},
\end{equation*}
so that $\partial_\kappa\theta^+(\kappa,\mu)<0$ and $\partial_\kappa\theta^-(\kappa,\mu)>0$ for $\kappa\in(\kappa_0(\mu),1)$. Hence $\theta^+(\cdot,\mu)$ is a  diffeomorphism from $(\kappa_0(\mu),1)$ onto $\left(0,\frac{\pi}{2}\right)$, and $\theta^-(\cdot,\mu)$ is a diffeomorphism from $(\kappa_0(\mu),1)$ onto $\left(-\frac{\pi}{2},0\right)$.

Therefore, the Inverse Function Theorem yield $C^1$ inverse functions for $\theta^+(\kappa,\mu)$ and $\theta^-(\kappa,\mu)$, which we denote by $\kappa^+(\theta,\mu)$ and $\kappa^-(\theta,\mu)$, defined on $\left(0,\frac{\pi}{2}\right)$ and $\left(-\frac{\pi}{2},0\right)$, respectively. These functions satisfy $\lim_{\theta\to 0^{\pm }}\kappa^\pm(\theta,\mu)=1$ and $\lim_{\theta\to 0^{\pm }}\partial_\theta\kappa^\pm(\theta,\mu)=0$ and, therefore, one can define the smooth function
\begin{equation*}
    \kappa(\theta,\mu):=
    \begin{cases}
        \kappa^-(\theta,\mu), & \theta\in\left(-\frac{\pi}{2},0\right],\\[1mm]
        \kappa^+(\theta,\mu), & \theta\in\left[0,\frac{\pi}{2}\right).
    \end{cases}
\end{equation*}
From now on, since it is enough for our purposes,  we restrict to $\theta\in\left(-\frac{\pi}{4},\frac{\pi}{4}\right)$.

Now we consider the transformation from the coordinates $(r,\theta,R,\Theta)$ in \eqref{eqn: (q,p) - (r,theta,R,Theta)} to the ones centered at the center of mass $(\hat r,\hat \theta,\hat R,\hat \Theta)$ given by
\begin{equation}\label{eqn: (R,Theta) - (hat R, hat Theta)}
    \begin{aligned}
    \hat \theta(r,\theta) =& \arctan\left(\frac{r\sin\theta}{r\cos\theta+(1-\mu)}\right),\\
    R(r,\theta,\hat R,\hat \Theta) =&\, \hat{R} \cdot \frac{r+(1-\mu)\cos\theta}{\sqrt{r^2 + 2r(1-\mu)\cos\theta + (1-\mu)^2}} \\
    &+(1-\mu)\sin\theta\left(\frac{\hat{\Theta}}{r^2 + 2r(1-\mu)\cos\theta + (1-\mu)^2} -1\right),\\
    \Theta(r,\theta, \hat R,\hat \Theta) =&\, \hat \Theta \cdot\frac{r(r+(1-\mu)\cos\theta)}{r^2 + 2r(1-\mu)\cos\theta + (1-\mu)^2}\\
    &- r(1-\mu)\left(\hat{R}\frac{\sin\theta}{\sqrt{r^2 + 2r(1-\mu)\cos\theta + (1-\mu)^2}} +\cos\theta\right).
    \end{aligned}
\end{equation}
Relying on \eqref{eqn: (R,Theta) - (hat R, hat Theta)} and Proposition \ref{prop: Parameterization of the invariant manifolds of infinity} we have, for some fixed $m > 0$ (independent of $\mu)$ and
\[\hat\Theta_0\in \left[-\sqrt 2 +m,\sqrt 2 - m\right],\quad \hat r = 1-\mu+\kappa(\theta,\mu)\mu^\nu,\quad r=\mu^\nu,\] the following estimates
\begin{equation}\label{eqn: hat R, hat Theta inf}
\begin{aligned}
\hat R &= \hat R(\theta,\hat\Theta_0) = \hat R_\infty^u(1-\mu+\kappa(\theta,\mu)\mu^\nu,\hat\theta(\mu^\nu,\theta),\hat\Theta_0) \\
&= -\frac{\sqrt{2\left(1-\mu+\kappa(\theta,\mu)\mu^\nu\right)-\hat\Theta_0^2}}{1-\mu+\kappa(\theta,\mu)\mu^\nu} + \mathcal{O}_{C^1}(\mu^{1-2\nu}) =- \sqrt{2-\hat\Theta_0^2} + \mathcal{O}_{C^1}(\mu^\nu),\\
\hat \Theta &=\hat \Theta(\theta,\hat\Theta_0) = \hat \Theta_\infty^u\left(1-\mu+\kappa(\theta,\mu)\mu^\nu,\hat\theta(\mu^\nu,\theta),\hat\Theta_0\right) = \hat\Theta_0 + \mathcal{O}_{C^1}(\mu^{1-2\nu}),
\end{aligned}
\end{equation}
where we have used that $\mu^\nu \gg \mu^{1-2\nu}$ since $\nu \in (0,1/3)$. Computing both $\partial_{\theta}R(\mu^\nu,\theta,\hat R(\theta,\hat\Theta_0),\hat\Theta(\theta,\hat\Theta_0))$ and $\partial_{\theta}\Theta(\mu^\nu,\theta,\hat R(\theta,\hat\Theta_0),\hat\Theta(\theta,\hat\Theta_0))$ from \eqref{eqn: (R,Theta) - (hat R, hat Theta)} (relying on the estimates in \eqref{eqn: hat R, hat Theta inf}) yields \eqref{eqn: R_infu and Theta_infu}. Finally, we restrict the domain of $\hat\Theta_0$ to
\[ \hat\Theta_0 \in \left(\frac{1-\sqrt 3}{2},\frac{1+\sqrt 3}{2}\right)\subset \left[-\sqrt 2 + m, \sqrt 2 - m\right],\]
for $m = \frac{1}{200}$ so that $\cos\theta > \sin \theta$ and $\sqrt{2-\Theta_0^2 }> |\hat\Theta_0-1|$. Therefore the leading term for $R_\infty^u$ in \eqref{eqn: R_infu and Theta_infu} given by $f(\theta,\hat\Theta_0)= -\cos\theta\sqrt{2-\hat\Theta_0^2} + \sin\theta(\hat\Theta_0-1) < 0$, completing the proof.

\end{proof}

\section{The distance between the invariant manifolds}\label{sec: The distance between the invariant manifolds and EC orbits}
Once we have characterized the invariant manifolds of infinity and the ejection/collision orbits with $\mathcal J$, we analyze their intersections at a common section. To this end, we recall:
\begin{itemize}
    \item The curves $\Lambda_{\mathcal J}^\mp(\mu)$, provided by Proposition \ref{prop: Parameterization of the ECO orbits in (r,theta,R,Theta)}, are the intersection of the ejection and collision orbits $\mathcal J^\mp$ (see Definition \ref{def: EC-orbits}) with the section $\Sigma_\gamma$ in \eqref{eqn: Section Sigma (r,theta,R,Theta)} with $\gamma \in (\frac{3}{11},1)$. They admit a graph parameterization in polar coordinates $(\theta,\Theta)$, as shown in \eqref{eqn: Curves ECO-section}.
    \item The curves $\Lambda_\infty^{u,s}(\hat \Theta_0,\mu)$, provided by Proposition \ref{prop: Parameterization of the invariant manifolds of infinity in (r,theta,R,Theta)}, are the intersection of the invariant manifolds of infinity with the section $ \Sigma_\nu$ with $\nu \in \left(0,\frac 1 3\right)$. They admit a graph parameterization in polar coordinates $(\theta,\Theta)$, as shown in \eqref{eqn: curve unstable inf cap r=mudelta2} and \eqref{eqn: curve stable inf cap r=mudelta2}.
    \item Taking $\gamma = \nu\in \left(\frac{3}{11},\frac 1 3\right)$, both curves  $\Lambda_{\mathcal J}^\mp(\mu)$ and $\Lambda_\infty^{u,s}(\hat \Theta_0,\mu)$ belong to the same section $\Sigma_\nu$.
\end{itemize}
The following theorem provides the transverse intersection between the curves $\Lambda_\infty^s(\hat\Theta_0,\mu) \subset W_\mu^s(\Alpha_{\hat\Theta_0})$ and $ \Lambda_{\mathcal J}^{-}(\mu)\subset \mathcal J^-$ in $\Sigma_\nu^>$ (see Remark \ref{remark: Angle ECO-orbits} and Proposition \ref{prop: Parameterization of the invariant manifolds of infinity}) for adequate values of $\hat \Theta_0$. By the symmetry \eqref{eqn: Symmetry in synodical polar coordinates centered at J}, there will be also a transverse intersection between the curves ${\Lambda}_\infty^u(\hat\Theta_0,\mu)\subset W_\mu^u(\Alpha_{\hat\Theta_0})$ and $\Lambda_{\mathcal J}^{+}(\mu) \subset \mathcal J^+$ in $\Sigma_\nu^<$.
\begin{theorem}\label{thm: distance transversality}
    Fix $\hat \Theta_0 \in \left(\frac{1-\sqrt 3}{3},\frac{1+\sqrt{3}}{3}\right)$, $\nu \in \left(\frac{3}{11}, \frac 1 3\right)$ and consider the section $\Sigma_\nu \subset \{\mathcal H_\mu = h\}$ in \eqref{eqn: Section Sigma (r,theta,R,Theta)}. There exists $\mu_0> 0$ such that, for any $\mu \in (0,\mu_0)$ we have
    \begin{itemize}
        \item $\Lambda_\infty^s(\hat\Theta_0,\mu)$ and $\Lambda_{\mathcal J}^{-}(\mu)$ intersect transversally at a point
        \begin{equation}\label{eqn:p-Thm62}
             p_-(\hat{\Theta}_0,\mu) = \left(\theta_-(\hat{\Theta}_0,\mu),\Theta_\infty^s\left(\theta_-(\hat{\Theta}_0,\mu),\hat{\Theta}_0\right)\right) = \left(\theta_-(\hat{\Theta}_0,\mu), \Theta_{\mathcal J}^-\left(\theta_-(\hat{\Theta}_0,\mu)\right)\right),
        \end{equation}
        where $\theta_-(\hat{\Theta}_0,\mu) = \arctan\left(\frac{\hat{\Theta}_0-1}{\sqrt{2-\hat{\Theta}_0^2}}\right) +\mathcal{O}\left(\mu^{\frac{11\nu-3}{8}}\right) \in \left(-\frac \pi 4, \frac \pi 4\right)$.
        \item  ${\Lambda}_\infty^u(\hat\Theta_0,\mu)$ and $\Lambda_{\mathcal J}^{+}(\mu)$ intersect transversally at a point
        \begin{equation}\label{eqn:p+Thm62}
        p_+(\hat{\Theta}_0,\mu) = \left(\theta_+(\hat{\Theta}_0,\mu),\Theta_\infty^u\left(\theta_+(\hat{\Theta}_0,\mu),\hat{\Theta}_0\right)\right) = \left(\theta_+(\hat{\Theta}_0,\mu), \Theta_{\mathcal J}^+\left(\theta_+(\hat{\Theta}_0,\mu)\right)\right),
        \end{equation}
        where $\theta_+(\hat{\Theta}_0,\mu) = -\theta_-(\hat\Theta_0,\mu) = -\arctan\left(\frac{\hat{\Theta}_0 - 1}{\sqrt{2-\hat{\Theta}_0^2}}\right) +\mathcal{O}\left(\mu^{\frac{11\nu-3}{8}}\right)\in \left(-\frac \pi 4, \frac \pi 4\right).$
    \end{itemize}
\end{theorem}
\begin{proof}
We compute the distance between these curves using the $\Theta$-component in the transverse section $ \Sigma_\nu$ in \eqref{eqn: Section Sigma (r,theta,R,Theta)} for $h = -\hat\Theta_0$. Relying on \eqref{eqn: (R,Theta) in curves ECO-section}, \eqref{eqn: R_infu and Theta_infu} and \eqref{eqn: R_infs and Theta_infs}, the asymptotic formulas for the distances between $\Lambda_\infty^{s}(\hat{\Theta}_0,\mu), \Lambda_{\mathcal J}^{-}(\mu)\subset \Sigma_\nu^>$ (which we denote by $d_-$), and between $\Lambda_\infty^{u}(\hat{\Theta}_0,\mu), \Lambda_{\mathcal J}^{+}(\mu)\subset \Sigma_\nu^<$ (which we denote by $d_+$) are given by
\begin{equation}\label{eqn: approx distances}
    \begin{aligned}
         d_\mp(\theta,\hat{\Theta}_0) &= \mu^\nu\left(\cos\theta\left(\hat{\Theta}_0-1\right)\mp\sin\theta\sqrt{2-\hat{\Theta}_0^2}\right)+ \mathcal{O}\left(\mu^{\frac{19\nu-3}{8}}\right),\\
         d_\mp{}'(\theta,\hat{\Theta}_0) &= \mu^\nu\left(-\sin\theta\left(\hat{\Theta}_0-1\right)\mp\cos\theta\sqrt{2-\hat{\Theta}_0^2}\right) + \mathcal{O}\left(\mu^{\frac{19\nu-3}{8}}, \mu^{1-2\nu}\right),
    \end{aligned}
\end{equation}
for $\theta \in \left(-\frac{\pi}{4},\frac{\pi}{4}\right)$. Therefore, finding a transverse intersection between $ \Lambda_\infty^s(\hat\Theta_0,\mu)$ and $ \Lambda_{\mathcal J}^{-}(\mu)$ is equivalent to find a non-degenerate zero of $ d_-(\theta,\hat{\Theta}_0)$. We write
\[\mathcal{F}_-(\theta,\mu^\nu) = \mu^{-\nu}\cdot  d_-(\theta,\hat{\Theta}_0) = \cos\theta\left(\hat{\Theta}_0-1\right)-\sin\theta\sqrt{2-\hat{\Theta}_0^2} + \mathcal{O}\left(\mu^{\frac{11\nu-3}{8}}\right).\]
For $\hat\Theta_0\in  \left(\frac{1-\sqrt 3}{3},\frac{1+\sqrt{3}}{3}\right)$, the angle $\theta_* := \arctan\left(\frac{\hat\Theta_0-1}{\sqrt{2-\hat\Theta_0^2}}\right)\in\left(-\frac{\pi}{4},\frac{\pi}{4}\right)$ satisfies
\[
\mathcal F_-(\theta_*,0)=0,\qquad
\partial_\theta \mathcal F_-(\theta_*,0)
=  -\sqrt{(\hat\Theta_0-1)^2 + (2-\hat\Theta_0^2)} \;\neq\; 0,
\]
and the Implicit Function Theorem gives the result. 
\end{proof}

\section{Proof of Theorem \ref{thm: ejection-collision orbits}}\label{sec: Proof of ECO SJ}
We divide the proof of Theorem \ref{thm: ejection-collision orbits} as follows. In Section \ref{subsec: first part of theorem eco} we prove the existence of an unbounded sequence of orbits in $\mathcal J^- \cap \mathcal J^+$ (see Definition \ref{def: EC-orbits})
and we obtain the existence of large ejection-collision orbits $\mathcal J^\mp \cap \mathcal{S}^\pm$. Then, in Section \ref{subsec: second part of thm eco} we prove the existence of ballistic ejection-collision orbits between the primaries (see Remark \ref{def: large-ballistic} for the definition of large and ballistic ejection-collision orbits).

\subsection{Large ejection-collision orbits }\label{subsec: first part of theorem eco}

This section proves the first two statements of Theorem \ref{thm: ejection-collision orbits}. The proof 
requires a global analysis of the dynamics (including close passages to infinity), for which the polar coordinates $(r,\theta,R,\Theta)$, used for the local analysis close to Jupiter, are no longer adequate. Instead, we consider (rotating) polar coordinates centered at the center of mass $(\hat r, \hat \theta,\hat R, \hat \Theta)$, in which Proposition \ref{prop: Parameterization of the invariant manifolds of infinity} provides an explicit description of the invariant manifolds of infinity on the following ``outer'' sections. For fixed $\hat r_0 > 1$, we define
\begin{equation}\label{eqn: Sections hatr big}
    \begin{aligned}
        \hat \Sigma_{\hat r_0}^> &= \left\{(\hat r, \hat \theta,\hat R,\hat \Theta)\colon \hat r = \hat r_0, \hat{\mathcal H}_\mu(\hat r_0,\hat \theta,\hat R,\hat \Theta) = h, \hat R > 0\right\},\\
        \hat \Sigma_{\hat r_0}^< &= \left\{(\hat r, \hat \theta,\hat R,\hat \Theta)\colon \hat r = \hat r_0, \hat{\mathcal H}_\mu(\hat r_0,\hat \theta,\hat R,\hat \Theta) = h, \hat R < 0\right\},
    \end{aligned}
\end{equation}
where $\hat{\mathcal H}_\mu$ is the Hamiltonian defined in \eqref{eqn: Hamiltonian Polar Rotating Coordinates centered at CM J}. Then, the intersections $W_\mu^{s}(\Alpha_{\hat \Theta_0}) \cap \hat \Sigma_{\hat r_0}^>$ and $W_\mu^{u}(\Alpha_{\hat \Theta_0}) \cap \hat \Sigma_{\hat r_0}^<$ are curves parameterized as graphs in terms of $\hat \theta \in \mathds{T}$. 

We fix $h=-\hat\Theta_0\in \left(-\frac{1+\sqrt 3}{3},\frac{\sqrt 3-1}{3}\right)$, $\nu \in \left(\frac{3}{11},\frac 1 3\right)$ and $\mu \in (0,\mu_0)$ so that Theorem \ref{thm: distance transversality} holds. Denote by $\hat\Phi_\mu$ the flow of the equations of motion associated to the Hamiltonian $\hat{\mathcal H}_\mu$ in \eqref{eqn: Hamiltonian Polar Rotating Coordinates centered at CM J}. Since $W_\mu^{s}(\Alpha_{\hat\Theta_0})$ and the ejection orbits $\mathcal J^{-}$ are invariant (and analogously $W_\mu^{u}(\Alpha_{\hat\Theta_0})$ and the collision orbits $\mathcal J^{+}$), the transverse intersections $p_-(\hat\Theta_0,\mu)\in \Sigma_\nu^>$ and $p_+(\hat\Theta_0,\mu)\in \Sigma_\nu^<$ (defined in \eqref{eqn:p-Thm62} and \eqref{eqn:p+Thm62}) are sent by $\hat\Phi_\mu$ to two transverse intersections at $\hat \Sigma_{\hat r_0}^>$ and $\hat \Sigma_{\hat r_0}^<$ respectively. We denote by 
\begin{equation}\label{eqn: points Sigmahat}
\hat p_{\mathcal J}^- = (\hat \theta_{\mathcal J}^-, \hat \Theta_{\mathcal J}^-) \in \hat \Sigma_{\hat r_0}^>,\qquad \hat p_{\mathcal J}^+ = (\hat \theta_{\mathcal J}^+, \hat \Theta_{\mathcal J}^+) \in \hat \Sigma_{\hat r_0}^<,
\end{equation}
the corresponding transverse intersections of $W_\mu^{s}(\Alpha_{\hat\Theta_0})$ and $\mathcal J^-$ at $\hat \Sigma_{\hat r_0}^>$ and $W_\mu^{u}(\Alpha_{\hat\Theta_0})$ and $\mathcal J^+$ at $\hat \Sigma_{\hat r_0}^<$. 

We consider the following annulus
\begin{equation}\label{eqn: rectR}
\hat{\mathcal R} = \left\{(\hat\theta,\hat\Theta)\in \hat{\Sigma}_{\hat r_0}^>\colon \hat \Theta_\infty^s(\hat \theta)-\varepsilon<\hat \Theta<\hat \Theta_\infty^s(\hat\theta)\right\},
\end{equation}
for $\varepsilon > 0$ small enough, where  $\hat \Theta_\infty^s(\hat \theta) := \hat \Theta_\infty^s(\hat r_0,\hat \theta,\hat\Theta_0)$ is defined in \eqref{eqn: invariant manifolds of infty as graphs in terms of hat theta}. Then, we define the Poincar\'e map
\begin{equation}\label{eqn: Poincare map infty}
    \begin{aligned}
        \hat{\mathcal{P}}\colon \hat{\mathcal R} \subset  \hat\Sigma^>_{\hat r_0} &\to  \hat \Sigma^<_{\hat r_0}\\
         \hat p = (\hat \theta,\hat \Theta) &\mapsto \hat{\mathcal{P}}( \hat p) = \hat \Phi_\mu\left(t_\mu( \hat p),\left(\hat r_0, \hat \theta, \hat R(\hat \theta,\hat \Theta;h,\mu),\hat \Theta\right)\right),
    \end{aligned}
\end{equation}
where $\hat R(\hat \theta,\hat \Theta;h,\mu)$ is the radial velocity, which can also be computed from \eqref{eqn: Hamiltonian Polar Rotating Coordinates centered at CM J}, and $t_\mu(\hat p)$ is the time needed for the orbit with initial condition at $\hat p\in \hat{\mathcal R}$ to reach $\hat \Sigma^<_{\hat r_0}$. By construction, $t_\mu(\hat p)$ is well defined and finite for $ \hat p\in \hat{\mathcal R}$ (but becomes unbounded as $ \hat p$ gets closer to the invariant manifold $W_\mu^s(\Alpha_{\hat \Theta_0})\cap \hat \Sigma_{\hat r_0}^>$). 

\paragraph{Step 1. Proof of first statement of Theorem \ref{thm: ejection-collision orbits}:}
For $\varepsilon> 0$ small enough (independent of $\mu$), denote by
%
\[\hat {\mathcal R}_{\mathcal J} = \left\{(\hat \theta,\hat \Theta)\in \hat{\mathcal{R}}\colon |\hat \theta - \hat \theta^-_{\mathcal J}|<\varepsilon \right\}\subset \hat{\mathcal R}\]
and consider the $C^1$-curve
\[\hat \gamma^{-}_{\mathcal J} := \mathcal J^- \cap \hat{\mathcal R}_{\mathcal J}\]
%
which intersects transversally $ W_\mu^s(\Alpha_{\hat\Theta_0})\cap \hat \Sigma_{\hat r_0}^> $ at $\hat p^-_{\mathcal J}$, defined in \eqref{eqn: points Sigmahat} (see Figure \ref{fig:Transition map close to infinity}).
Since the points of the curve $\hat \gamma^{-}_{\mathcal J}$ are close to the point $ \hat p^-_{\mathcal J}$, they are close to $ W_\mu^s(\Alpha_{\hat\Theta_0})\cap \hat \Sigma_{\hat r_0}^> $. Hence, the study of the image $\hat{\mathcal{P}}(\hat \gamma^{-}_{\mathcal J})$, where $\hat{\mathcal P}$ is the Poincar\'e map defined in \eqref{eqn: Poincare map infty}, is given by the analysis of the dynamics ``close'' to $\Alpha_{\hat{\Theta}_0}$. Thus, one can easily adapt the approach done by Moser for the Sitnikov problem in \cite{moser2001stable} (see also \cite{MR3455155}) to this case, and prove that the image $\hat{\mathcal{P}}(\hat \gamma^{-}_{\mathcal J})$ ``spirals'' towards $ W_\mu^u(\Alpha_{\hat\Theta_0})\cap \hat \Sigma_{\hat r_0}^<$ (see Figure \ref{fig:Transition map close to infinity}).
In particular, we have
\begin{equation}\label{eqn: Spiral J}
W_\mu^u(\Alpha_{\hat\Theta_0})\cap \hat \Sigma_{\hat r_0}^< \subset \textrm{Cl}\left(\hat {\mathcal{P}}(\hat \gamma^{-}_{\mathcal J})\right).
\end{equation}
Consider the $C^1$-curve
\[\hat \gamma^{+}_{\mathcal J} := \hat{\mathcal V}_{\mathcal J} \cap \mathcal J^+,\]
where $\hat{\mathcal V}_{\mathcal J}$ is a sufficiently small neighborhood of $\hat p_{\mathcal J}^+$ in \eqref{eqn: points Sigmahat}.  Then, there exists a sequence $ \hat q_k \in \hat \gamma^{+}_{\mathcal J}\cap \hat{\mathcal{P}}(\hat \gamma^{-}_{\mathcal J})$ such that $\underset{k\to+\infty}{\lim}  \hat q_k =\hat p_{\mathcal J}^+$. In particular, the closer is $ \hat q_k$ to $W_\mu^u(\Alpha_{\hat\Theta_0})\cap \hat \Sigma_{\hat r_0}^< $, the larger is the time $t_\mu(\hat p_k)$, where $\hat p_k$ is such that $\hat{\mathcal{P}}( \hat p_k) = \hat q_k$. Since $\hat \gamma^{+}_{\mathcal J}  \subset \mathcal J^+$ and $\hat{\mathcal{P}}( \hat \gamma^{-}_{\mathcal J}) \subset \mathcal{J}^-$, the points $\hat q_k\in \mathcal{J}^- \cap \mathcal{J}^+$, completing the first statement of Theorem \ref{thm: ejection-collision orbits}.

\paragraph{Step 2. Proof of second statement of Theorem \ref{thm: ejection-collision orbits}:}

To prove the second statement, we use Theorem 1.7 of \cite{guardia2024oscillatorymotionsparabolicorbits} where is shown that, for sufficiently small values of $h=-\hat{\Theta}_0$ of order $\mu$ (and therefore satisfying the conditions of Theorem \ref{thm: distance transversality}), there exist transverse intersections  $W_\mu^s(\Alpha_{\hat{\Theta}_0})\transv W_\mu^u(S^+)$ and $W_\mu^u(\Alpha_{\hat{\Theta}_0}) \transv W_\mu^s(S^-)$ (where $W_\mu^{s,u}(S^\pm)$ are defined in \eqref{eqn: definition of the invariant manifolds of collision in rotating polar coordinates centered at P1 J}) at the sections $\overline \Sigma^>$ and $\overline \Sigma^<$ in \eqref{eqn: sectionSun} respectively. By the invariance of $W_\mu^{s,u}(\Alpha_{\hat{\Theta}_0})$ and $W_\mu^{s,u}(S^\pm)$, these transverse intersections are sent by $\hat \Phi_\mu$ to two transverse intersections at the sections $ \hat \Sigma_{\hat r_0}^>$ and $\hat \Sigma_{\hat r_0}^<$ in \eqref{eqn: Sections hatr big} respectively.

Denote by $ \hat p^-_{\mathcal S} = (\hat \theta^-_{\mathcal S},\hat \Theta^-_{\mathcal S}) \in  \hat \Sigma_{\hat r_0}^>$ the point of (transverse) intersection between $W_\mu^s(\Alpha_{\hat\Theta_0})$ and $W_{\mu}^u(S^+)$ with the section $ \hat \Sigma_{\hat r_0}^>$ (see Figure \ref{fig:Transition map close to infinity}). For $\varepsilon>0$ small enough (independent of $\mu$), denote by 
\[\hat{\mathcal R}_{\mathcal S}=\left\{(\hat \theta,\hat \Theta)\in \hat{\mathcal{R}}\colon |\hat \theta-\hat \theta^-_{\mathcal{S}}|< \varepsilon\right\}\subset \hat{\mathcal R}\]
and consider the $C^1$-curve
\[ \hat \gamma_{\mathcal S}^{-} :=  W_\mu^u(S^+)\cap \hat{\mathcal R}_{\mathcal S} \subset \mathcal S^-,\]
where $\mathcal S^-$ are the ejection orbits from the Sun (see Definition \ref{def: EC-orbits}). Following the same argument as in \eqref{eqn: Spiral J} we have
\[W_\mu^u(\Alpha_{\hat \Theta_0})\cap \hat \Sigma_{\hat r_0}^<\subset \mathrm{Cl}\left(\hat{\mathcal{P}}(\hat \gamma_{\mathcal S}^{-})\right),\]
where $\hat{\mathcal P}$ is the Poincar\'e map \eqref{eqn: Poincare map infty}.

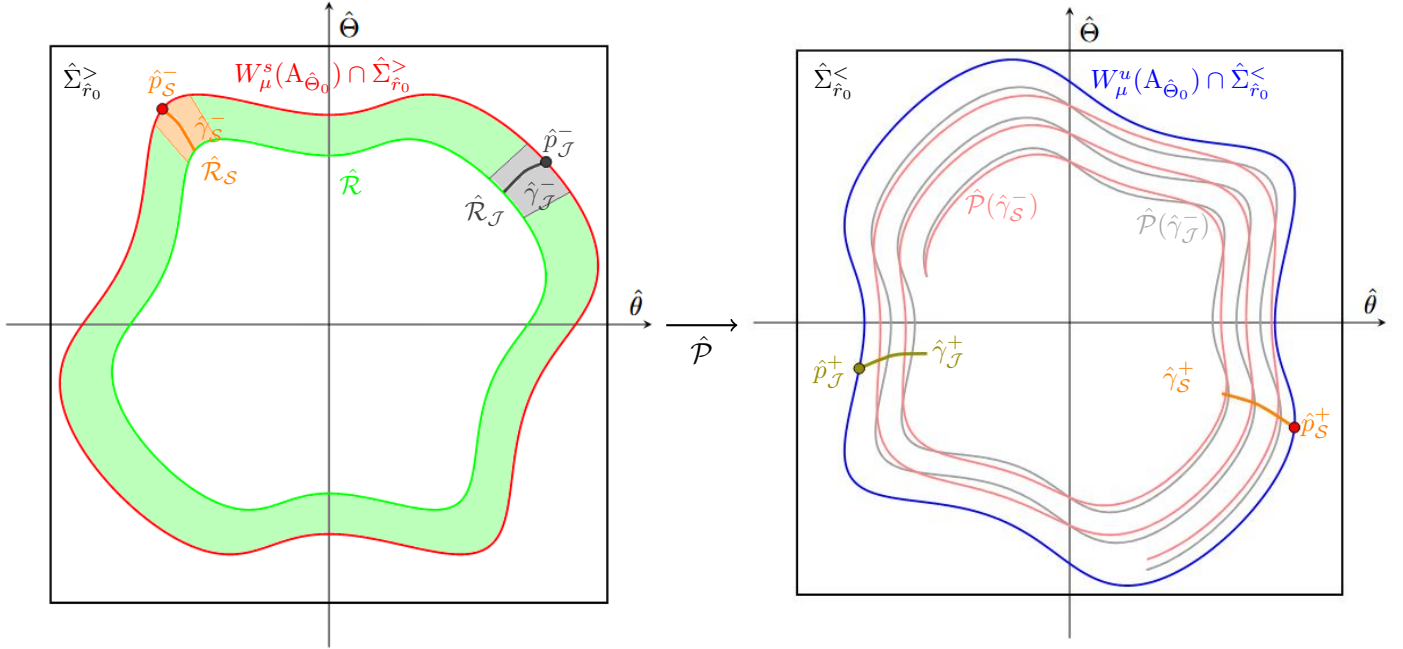
\begin{figure}[ht]
    \centering
    \begin{tikzpicture}
    \node (fig1) at (-3,0){
        \begin{overpic}[width=0.5\textwidth]{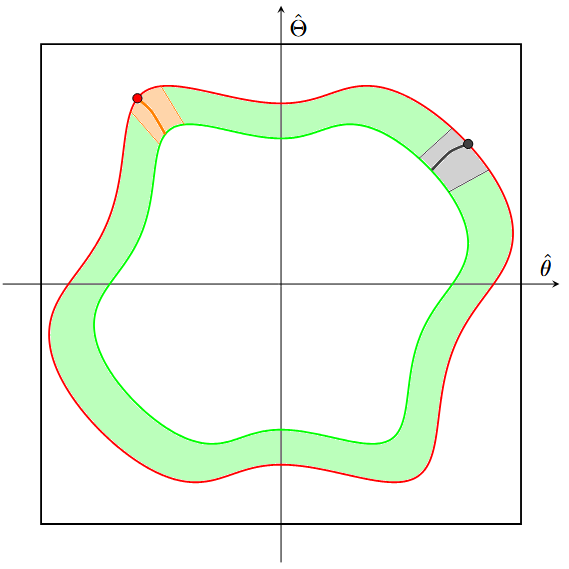}
            \put(9,86){\small $ \hat\Sigma_{\hat r_0}^>$}
            \put(22,86){\color{orange} \small $\hat p^-_{\mathcal S}$}
            \put(82,77){\color{darkgray} \small $\hat p_{\mathcal J}^-$}
            \put(30,72){\color{orange} \small $\hat{\mathcal R}_{\mathcal S}$}
            \put(51,70){\color{green} \small $\hat{\mathcal R}$}
            \put(70,66){\color{darkgray}\small $\hat{\mathcal R}_{\mathcal J}$}
            \put(35,87){\color{red}\small $W_\mu^s(\Alpha_{\hat\Theta_0})\cap  \hat\Sigma_{\hat r_0}^>$}
             \put(29,79){\color{orange}\small $ \hat\gamma_{\mathcal S}^-$}
            \put(79,69){\color{darkgray}\small $ \hat \gamma_{\mathcal J}^-$}
        \end{overpic}
    };
    \node (fig2) at (7,0){
        \begin{overpic}[width=0.5\textwidth]{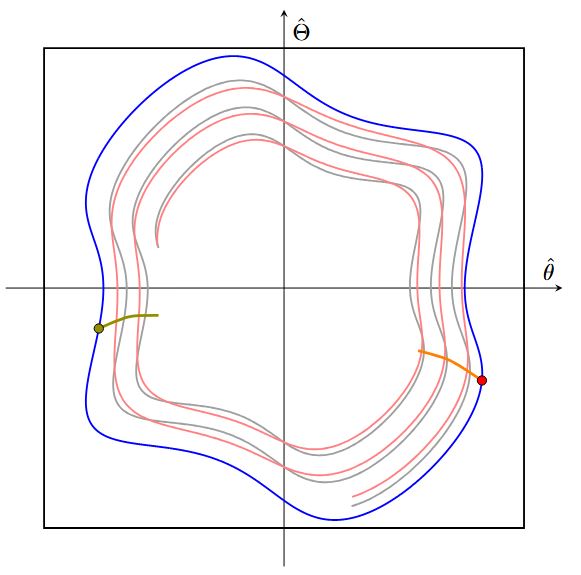}
            \put(10,85){\small $\hat{\Sigma}_{\hat r_0}^<$}
            \put(10,41){\color{olive} \small $\hat p_{\mathcal J}^+$}
            \put(28,44){\color{olive} \small $ \hat \gamma_{\mathcal J}^+$}
            \put(84,33){\color{orange} \small $ \hat p_{\mathcal S}^+$}
            \put(63,40){\color{orange}\small $\hat \gamma_{\mathcal S}^+$}
             \put (33,66){\color{red!50} \small $\hat{\mathcal P}(\hat \gamma_{\mathcal S}^-)$}
            \put (59,63) {\color{gray!70}\small $\hat{\mathcal P}(\hat \gamma_{\mathcal J}^-)$}
            \put(52,85){\color{blue} \small $W_\mu^u(\Alpha_{\hat\Theta_0})\cap  \hat{\Sigma}_{\hat r_0}^<$}
        \end{overpic}
        \label{fig:hierarchy-b}
    };
    \ \draw[->, thick] (fig1.east) -- (fig2.west) node[midway, below] {$\hat{\mathcal P}$};
    \end{tikzpicture}
    \caption{Representation of the ``spiraling effect'' of the transition map $\hat{\mathcal{P}}$ on the curves $\hat \gamma^{-}_{\mathcal S,\mathcal J}\subset  \hat{\Sigma}_{\hat r_0}^>$ and their transverse intersections with the curves $ \hat \gamma^{+}_{\mathcal S,\mathcal J}\subset \hat{\Sigma}_{\hat r_0}^<$.}
    \label{fig:Transition map close to infinity}
\end{figure}

Denote by $ \hat p^+_{\mathcal S} \in \hat\Sigma_{\hat r_0}^<$ the point of (transverse) intersection of $W_\mu^u(\Alpha_{\hat\Theta_0})$ and $W_\mu^s(S^-)$ with the section $\hat \Sigma_{\hat r_0}^<$ defined in \eqref{eqn: Sections hatr big} and consider the $C^1$-curve 
\[\hat  \gamma^{+}_{\mathcal S} := W_\mu^s(S^-) \cap \hat{\mathcal V}_{S}\subset \mathcal S^+,\]
where $\mathcal S^+$ are the collision orbits to the Sun (see Definition \ref{def: EC-orbits}) and $\hat{\mathcal V}_{S}$ is a sufficiently small neighborhood of $ \hat p^+_{\mathcal S}$. Therefore, there exists
\begin{itemize}
    \item A sequence $ \hat x_k \in  \hat \gamma^{+}_{\mathcal J}\cap \hat{\mathcal{P}}(\hat\gamma^{-}_{\mathcal S})$ such that $\underset{k\to+\infty}{\lim} \hat x_k = \hat p_{\mathcal J}^+$ in \eqref{eqn: points Sigmahat}. In particular, the closer is $ \hat x_k$ to $W_\mu^u(\Alpha_{\hat \Theta_0})\cap \hat \Sigma_{\hat r_0}^<$, the larger is the time $t_\mu(\hat p_k)$, where $ \hat p_k$ is such that $\hat{\mathcal{P}}(\hat p_k) = \hat x_k$. Since $\hat{\mathcal{P}}( \hat \gamma^{-}_{\mathcal S})\subset W_\mu^u(S^+)$, the points $\hat x_k \in W_\mu^u(S^+)\cap \mathcal{J}^+$ and, therefore, they belong to $\mathcal{S}^- \cap \mathcal{J}^+$.
    
    \item A sequence $ \hat y_k \in \hat \gamma^{+}_{\mathcal S}\cap \hat{\mathcal{P}}(\hat \gamma^{-}_{\mathcal J})$ such that $\underset{k\to+\infty}{\lim}  \hat y_k =  \hat p^+_{\mathcal S}$. In particular, the closer is $ \hat y_k$ to $ W_\mu^u(\Alpha_{\hat \Theta_0})\cap \hat \Sigma_{\hat r_0}^<$, the larger is the time $t_\mu(\hat p_k)$, where $ \hat p_k$ is such that $\hat{\mathcal{P}}(\hat p_k) = \hat y_k$. Since $\hat{\mathcal{P}}(\hat \gamma^{-}_{\mathcal J})\subset \mathcal{J}^-$, the points $ \hat y_k \in \mathcal{J}^-\cap W_\mu^u(S^-)$ and, therefore, they belong to $\mathcal{J}^-\cap \mathcal{S}^+$.
\end{itemize}
\subsection{
Ballistic ejection-collision orbits}\label{subsec: second part of thm eco}
This section is devoted to prove the last item of Theorem \ref{thm: ejection-collision orbits}, that is, the existence of a ballistic orbit (see Remark \ref{def: large-ballistic}) in $\mathcal S^+ \cap \mathcal J^-$ (by the symmetry \eqref{eqn: Symmetry of the RPC3BP in synodical polar coordinates centered at P1 J}, also a ballistic orbit in $\mathcal J^+\cap \mathcal S^-$) satisfying \eqref{eqn: SJ bounded}. To this end, we compare the ejection orbits $\mathcal J^-$  with the curve $\overline \Delta_{S^-}^s \subset \mathcal S^+\subset  \overline\Sigma^<$ in \eqref{eqn: definition of the intersection of the invariant manifolds of collision with delta2 as graphs J} (and between the collision orbits $\mathcal J^+$ and the curve $\overline \Delta_{S^+}^u \subset \mathcal S^-\subset \overline \Sigma^>$). For a fixed $h = \mathcal O(\mu)$, the following proposition provides a parameterization of the intersection between the ejection orbits $\mathcal J^-$ with the section $\overline \Sigma^<$ in \eqref{eqn: sectionSun} (and between the collision orbits $\mathcal J^+$ with the section $\overline \Sigma^>$).

\begin{proposition}\label{prop: parameterization invariant manifolds with section close to S}
Consider the rotating polar coordinates centered at the Sun $(\rS,  \thetaS,  \RS,  \ThetaS)$ defined in \eqref{eqn: Change from synodical cartesian to synodical polar centered at P1 J} and the parameter $\delta_0>0$  given in Proposition \ref{proposition: Perturbed invariant manifolds of collision in synodical polar coordinates J}. Fix $a,b$ with $a <b \in \mathds{R}$ and $ h_0 \in[a,b]$. Then, for $0<\delta<\delta_0$, there exists $\mu_0 > 0$ such that, for any $0<\mu<\mu_0$ and $h = \mu h_0$, the set of ejection orbits $\mathcal J^-$ intersects the section $\overline \Sigma^<$ in \eqref{eqn: sectionSun} in a curve written as a graph with respect to $ \thetaS$ as
\begin{equation}\label{eqn: Curve intersection Ejection-overline r=delta2}
     \overline \Delta_\mathcal{J}^{-}(\mu) = \left\{\left(\rS, \thetaS,  \RS_{\mathcal J}^-( \thetaS,\delta,\mu), \ThetaS_{\mathcal J}^-(\thetaS,\delta,\mu)\right)\colon  \rS = \delta^2, \thetaS \in \left(\thetaS_0-\delta^4,\thetaS_0+\delta^4\right)\right\}
\end{equation}
where $\thetaS_0 = -\frac{\sqrt 2}{3}(1-\delta^3)$ and $ \ThetaS_{\mathcal J}^-$ is a $C^1$-function satisfying
\begin{equation}\label{eqn: (R,Theta) Ejection-overline r=delta2}
    \begin{aligned}
         \ThetaS_{\mathcal J}^-(\thetaS,\delta,\mu) = \frac{\delta}{\sqrt 2}\left(\thetaS -\thetaS_0\right) + \mathcal{O}(\delta^5),\quad \partial_{\thetaS}\ThetaS_{\mathcal J}^-(\thetaS,\delta,\mu) = \frac{\delta}{\sqrt 2}\left(1+\mathcal O(\delta)\right).
    \end{aligned}
\end{equation}
Due to symmetry \eqref{eqn: Symmetry of the RPC3BP in synodical polar coordinates centered at P1 J}, the collision set of orbits $\mathcal J^+$ intersect the section $\overline \Sigma^>$ in \eqref{eqn: sectionSun} in a curve written as a graph with respect to $ \thetaS$ as
\begin{equation}\label{eqn: Curve intersection Collision-overline r=delta2}
     \overline \Delta_\mathcal{J}^{+}(\mu) = \left\{\left(\rS, \thetaS,  \RS_{\mathcal J}^+( \thetaS,\mu), \ThetaS_{\mathcal J}^+(\theta,\mu)\right)\colon  \rS = \delta^2, \theta \in\left(-\thetaS_0-\delta^4,-\thetaS_0+\delta^4\right)\right\}
\end{equation}
such that $\ThetaS_{\mathcal J}^+( \thetaS,\mu) = \ThetaS_{\mathcal J}^-( -\thetaS,\mu).$
\end{proposition}
The proof of this proposition is done in Appendix \ref{appendix: parameterization of invariant manifolds with section close to S}. 

Relying on this result 
and the Implicit Function Theorem, the following theorem provides the transverse intersection between the curves $\overline \Delta_{S^-}^{s}(\mu)$ and $\overline \Delta_{\mathcal{J}}^{-}(\mu)$ and the curves $\overline \Delta_{S^+}^{u}(\mu)$ and $\overline \Delta_{\mathcal{J}}^{+}(\mu)$ (see \eqref{eqn: definition of the intersection of the invariant manifolds of collision with delta2 as graphs J}, \eqref{eqn: Curve intersection Ejection-overline r=delta2} and \eqref{eqn: Curve intersection Collision-overline r=delta2} respectively). This proves the last statement of Theorem \ref{thm: ejection-collision orbits}.

\begin{theorem}\label{thm: Tranverse intersection S-J r=delta2}
Under the same assumptions as in Proposition \ref{prop: parameterization invariant manifolds with section close to S}, we have 
\begin{itemize}
    \item  $\overline \Delta_{S^-}^{s}(\mu)$ and $\overline \Delta_{\mathcal{J}}^{-}(\mu)$ intersect transversally at a point
    \[\pS_-= \left(\thetaS_-(\delta,\mu), \ThetaS_{S^-}^s(\thetaS_-(\delta,\mu))\right) = \left(\thetaS_-(\delta,\mu), \ThetaS_{\mathcal J}^-(\thetaS_-(\delta,\mu))\right),\]
    where $\thetaS_-(\delta,\mu) = \thetaS_0 + \mathcal{O}\left(\delta^4\right)$.
    \item $\overline \Delta_{S^+}^{u}(\mu)$ and $\overline \Delta_{\mathcal{J}}^{+}(\mu)$ intersect transversally at a point $\pS_+$ given by the symmetry \eqref{eqn: Symmetry of the RPC3BP in synodical polar coordinates centered at P1 J}.
    \item Both points $\pS_-$ and $\pS_+$ belong to the surface $\overline{\mathcal H}_\mu(\delta^2,\thetaS,\RS,\ThetaS) = h$, where $\overline{\mathcal H}_\mu$ is the Hamiltonian \eqref{eqn: Hamiltonian function in rotating polar coordinates centered at S J}.
\end{itemize}
\end{theorem}
\begin{proof}
Relying on \eqref{eqn: Theta S-} and \eqref{eqn: (R,Theta) Ejection-overline r=delta2}, we compute an asymptotic formula for the distance between the curves $\overline \Delta_{\mathcal{J}}^{-}(\mu)$ and $\overline \Delta_{S^-}^{s}(\mu)$ at the section $\overline \Sigma^<$ (defined in \eqref{eqn: sectionSun}) using the $\overline \Theta$-component, obtaining 
\begin{equation}\label{eqn: Estimates distance S-J}
    \begin{aligned}
        \overline d_-(\thetaS,\delta,\mu) = \frac{\delta}{\sqrt 2}\left(\thetaS -\thetaS_0\right) + \mathcal{O}\left(\delta^5\right),\quad \partial_{\thetaS}\overline d_-(\thetaS,\delta,\mu) = \frac{\delta}{\sqrt 2}\left(1 + \mathcal{O}(\delta)\right)
    \end{aligned}
\end{equation}
for $\thetaS \in \left(\thetaS_0-\delta^4,\thetaS_0+\delta^4\right)$, where $\thetaS_0$ is given in Proposition~\ref{prop: parameterization invariant manifolds with section close to S}.    

Finding a transverse intersection between $\overline \Delta_{\mathcal{J}}^{-}(\mu)$ and $\overline \Delta_{S^-}^{s}(\mu)$ in $\overline \Sigma^<$ is reduced to find a non-degenerate zero of $\overline d_-(\thetaS,\delta,\mu)$. Namely, we look for a non-degenerate zero of the function
\[\mathcal F(\thetaS,\delta,\mu) = \sqrt 2\delta^{-1}\overline d_-(\thetaS,\delta,\mu) = \thetaS+\frac{\sqrt 2}{3}(1-\delta^3) + \mathcal{O}(\delta^4).\]
Since $\mathcal F\left(\frac{\sqrt 2}{3},0,0\right) = 0$ and $\partial_{\thetaS} \mathcal F\left(\frac{\sqrt 2}{3},0,0\right)= 1 \neq 0$, the Implicit Function Theorem gives the result.
%
\end{proof}

\section{Proof of Theorems \ref{thm:chaosJ} and \ref{thm: Existence of parabolic, oscillatory and periodic orbits J}}\label{sec: Triple intersection}
We  fix $\nu = \gamma \in \left(\frac{3}{11}, \frac 1 3\right)$ so that the curves $ \Lambda_{\mathcal{J}}^{\mp}(\mu)$ and $ \Lambda_{\infty}^{u,s}(\hat \Theta_0,\mu)$, defined in \eqref{eqn: Curves ECO-section},  \eqref{eqn: curve unstable inf cap r=mudelta2} and \eqref{eqn: curve stable inf cap r=mudelta2} respectively, belong to the same section $ \Sigma_\nu$ in \eqref{eqn: Section Sigma (r,theta,R,Theta)}. 

The key step in the proof of Theorems \ref{thm:chaosJ} and \ref{thm: Existence of parabolic, oscillatory and periodic orbits J} is Proposition \ref{prop: triple intersection J} below. To state it, we define the \emph{triple intersection} of $W_\mu^u(\Alpha_{\hat\Theta_0}), W_\mu^s(\Alpha_{\hat\Theta_0})$ and $\mathcal J^-$.

\begin{definition}\label{def:tripleintersectionJ}
    We say that the invariant manifolds $W_\mu^u(\Alpha_{\hat\Theta_0}), W_\mu^s(\Alpha_{\hat\Theta_0})$ and the ejection orbits $\mathcal J^-$ have a triple intersection at $p_-^*=(\theta_-^*,\Theta_-^*)\in \Sigma_\nu^>$ in \eqref{eqn: Transverse sections} if
    \begin{itemize}
        \item $\Lambda_\infty^s(\hat\Theta_0,\mu)$ and $\Lambda_{\mathcal J}^-(\mu)$ intersect at $p_-^*=(\theta_-^*,\Theta_-^*) \in \Sigma_\nu^>$.
        \item $\Lambda_\infty^u(\hat\Theta_0,\mu)$ and $\Lambda_{\mathcal J}^+(\mu)$ intersect at $p_+^*=(\theta_+^*,\Theta_+^*)\in \Sigma_\nu^<$.
        \item The composition 
        \begin{equation}\label{eqn: composition triple intersection}
        f:=\Upsilon^{-1} \circ \psi\circ  \Xi^{-1}\circ \mathbf{f} \circ \Xi \circ \psi^{-1} \circ \Upsilon \colon \Sigma_\nu^< \to \Sigma_\nu^>
        \end{equation}
        maps the point $p_+^*$ to $p_-^*$. Here $\psi$, $\Xi$ and $\Upsilon$ are the diffeomorphisms defined in \eqref{eqn: Change of coordinates Levi-Civita}, \eqref{eqn: straightening (z,w)-(s,u)} and \eqref{eqn: (q,p) - (r,theta,R,Theta)} respectively, and the transition map $\mathbf f$ is given in Proposition \ref{prop: transition map close to collision}.

    \end{itemize}
    Moreover, if $f$ maps a $C^1$-curve in $ \Lambda_\infty^u(\hat\Theta_0,\mu)$ to a $C^1$-curve $\Gamma_\infty^{u,>} \subset \Sigma_\nu^>$ containing  $p_-^*$, then we say that the triple intersection is transverse if  the curves  $\Gamma_\infty^{u,>}$, $\Lambda_\infty^{s}(\hat\Theta_0,\mu)$ and $\Lambda_{\mathcal J}^-(\mu)$ intersect pairwise transversally at $p_-^*$.

\end{definition}

Next proposition proves the existence of a transverse triple intersection between the invariant manifolds $W_\mu^u(\Alpha_{\hat\Theta_0})$ and $W_\mu^s(\Alpha_{\hat\Theta_0})$ and the ejection orbits $\mathcal J^-$ at a suitable energy level.

\begin{proposition}\label{prop: triple intersection J}
    For any $\nu \in \left(\frac{3}{11},\frac 1 3\right),$ there exists $\mu_0 > 0$ such that, for any $\mu\in(0,\mu_0)$, there exists $\hat{\Theta}_0^* = \hat{\Theta}_0^* (\mu) =1+\mathcal{O}\left(\mu^{\frac{11\nu-3}{8}}\right)$ such that, for $h= h^*(\mu) = -\hat{\Theta}_0^*(\mu)$, the following statements hold.
    \begin{itemize}
    \item The invariant manifolds $W_\mu^u(\Alpha_{\hat\Theta_0}), W_\mu^s(\Alpha_{\hat\Theta_0})$ and the ejection orbits $\mathcal J^-$ have a transverse triple intersection at $p_-^*=(\theta_-^*,\Theta_-^*)\in \Sigma_\nu^>$ (in the sense of Definition \ref{def:tripleintersectionJ}) where
    \begin{equation*}
        \theta_-^* =\mathcal{O}\left(\mu^{\frac{11\nu-3}{8}}\right),\quad \Theta_-^* = \Theta_\infty^s(\theta_-^*,\hat{\Theta}_0^*(\mu)) = \Theta_{\mathcal J}^-(\theta_-^*),
    \end{equation*}
    where $\Theta_{\mathcal J}^-(\theta)$ and $\Theta_\infty^s(\theta,\hat\Theta_0)$ are defined in \eqref{eqn: (R,Theta) in curves ECO-section} and \eqref{eqn: R_infs and Theta_infs} respectively.
    
    \item For fixed $\iota > 0$ small enough independent of $\mu$, denote by $B_\iota(p_-^*) \subset  \Sigma_\nu^>$ a $\iota$-neighborhood of $p_-^*$. Let
        \begin{equation}\label{eqn: transversecurves}
        \begin{aligned}
            \Gamma_\infty^{s,>}:= \Lambda_\infty^s(\hat\Theta_0^*(\mu),\mu) \cap B_\iota(p_-^*),\quad \Gamma_{\mathcal J}^{-}:= \Lambda_{\mathcal J}^{-}(\mu) \cap B_\iota(p_-^*),
        \end{aligned}
        \end{equation}
        be the $C^1$-curves intersecting pairwise transversally at $p_-^*$ with $\Gamma_\infty^{u,>}$ (see Definition \ref{def:tripleintersectionJ}). Denote by 
        \begin{equation}\label{eqn: Curves as graphs of theta in coordinates (theta,Theta)}
        \begin{aligned}
        \gamma^{s,>}_\infty &:= \left\{\left(\theta,\Theta^{s}_\infty(\theta,\hat{\Theta}_0^*(\mu))\right)\colon \theta \in (\theta_-^* - \iota,\theta_-^*+\iota)\right\} \subset\Sigma_\nu^>,\\
        \gamma^{-}_{\mathcal{J}} &:= \left\{\left(\theta,\Theta^{-}_{\mathcal J}(\theta,\hat{\Theta}_0^*(\mu))\right)\colon \theta \in (\theta_-^* - \iota,\theta_-^*+\iota)\right\} \subset \Sigma_\nu^>,\\
        \gamma^{u,>}_\infty &:= \left\{\left(\theta, \Theta^{u,>}_\infty(\theta,\hat{\Theta}_0^*(\mu))\right)\colon \theta \in (\theta_-^* - \iota,\theta_-^*+\iota)\right\} \subset \Sigma_\nu^>,
        \end{aligned}
        \end{equation}
        their corresponding $C^1$ graph parameterizations in coordinates $(\theta,\Theta)$. Then, the angles
        \begin{equation}\label{eqn: angles trasnversality}
            A  = \measuredangle (\left(\gamma_{\infty}^{u,>}\right)'(\theta_-^*),  \left(\gamma_{\infty}^{s,>}\right)'(\theta_-^*)),\quad B  = \measuredangle ( \left(\gamma_{\infty}^{u,>}\right)'(\theta_-^*), \left(\gamma_{\mathcal J}^{-}\right)'(\theta_-^*))
        \end{equation}
        (taken in $[-\pi/2,\pi/2]$) satisfy
        \begin{equation}\label{eqn: Inequality angles}
         -\frac{\pi}{2}<  A <  B < 0,
        \end{equation}
        leading to the configuration depicted in Figure \ref{fig: trans intersection}.
    \end{itemize} 
\end{proposition}
The proof of this proposition is divided into two parts. In Section \ref{subsec: existence triple intersection} we prove that the triple intersection $p_-^*$ 
exists and in Section \ref{subsec: transversality and ordering triple intersection} we establish its transversality as well as the ordering given by \eqref{eqn: Inequality angles}. Now we prove Theorems \ref{thm:chaosJ} and \ref{thm: Existence of parabolic, oscillatory and periodic orbits J} relying on Proposition \ref{prop: triple intersection J}.

\begin{proof}[Proof of Theorems \ref{thm:chaosJ} and \ref{thm: Existence of parabolic, oscillatory and periodic orbits J}]
    At $h = h^*$, Proposition \ref{prop: triple intersection J} ensures that there exists a triple transverse intersection (in the sense of Definition \ref{def:tripleintersectionJ}) between the invariant manifolds of infinity and with the family of ejection orbits $\mathcal J^-$. In particular, this also means that there exist transverse intersections between the stable manifold of infinity and the family of ejection orbits $\mathcal J^-$ at $p_-^*\in \Sigma_\nu^>$, and between the unstable manifold of infinity and the family of collision orbits $\mathcal J^+$ at $p_+^*\in\Sigma_\nu^<$, that guarantees the existence of a parabolic ejection/collision orbit.
    \begin{figure}[h!]
            \centering
            \begin{overpic}[scale=0.6]{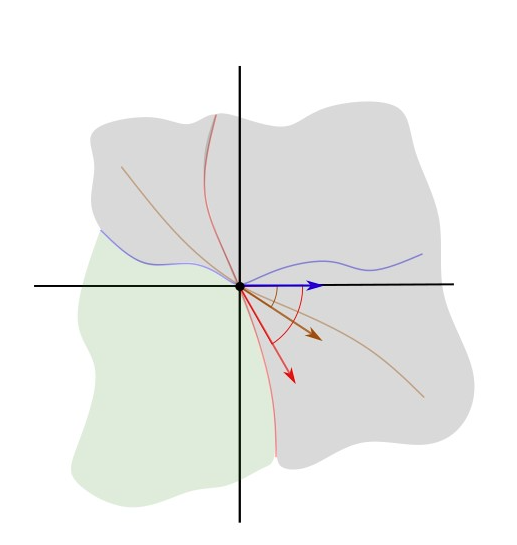}
                \put (10,85){$\Sigma_\nu^>$}
                \put (80,44) {$\theta$}
                \put (45,85) {$\Theta$}
                \put (39,44) {$p_-^*$}
                \put (58,44) {\color{blue} \footnotesize $ \gamma_\infty^{u,>}{}'(\theta_-^*)$}
                \put (52,26) {\color{red} \footnotesize $\gamma_{\infty}^{s,>}{}'(\theta_-^*)$}
                \put (55,34)  {\color{RawSienna} \footnotesize $ \gamma_{\mathcal J}^{-}{}'(\theta_-^*)$}
                \put (51,44)  {\color{brown} \footnotesize $B$}
                \put (52,37) {\color{red}\footnotesize $A$}
            \end{overpic}
            \caption{Representation of the ``ordering of the intersection'' given by the angles $A$ and $B$ in \eqref{eqn: angles trasnversality}. The region shaded in green represents the set $\mathcal D$ where the forward and backward return maps to the section $ \Sigma_\nu^>$ are well defined. That is,  the region where the construction made by Moser in \cite{moser2001stable} is performed.}
            \label{fig: trans intersection}
        \end{figure}
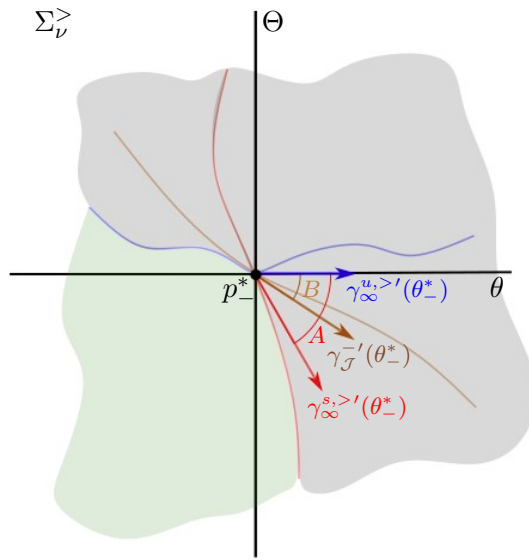
    
    Then, one can adapt the construction done by Moser in \cite{moser2001stable} to our setting (see also \cite{MR3455155,MR0573346}) to prove the existence of a hyperbolic set (at a suitable Poincar\'e section transverse to the flow) whose dynamics is conjugated to the shift of infinite symbols and which accumulates to the invariant manifolds of infinity. This leads to the existence of any combination of  hyperbolic, parabolic and oscillatory orbits in forward and backward time, and the existence of periodic orbits arbitrarily close to infinity (at the energy level $h^*$). The classical Moser approach relies on the existence of a transverse homoclinic point of the invariant manifolds at infinity, which is only present for the regularized flow associated to the equations of motion in \eqref{eqn: Eq motion in coordinates (s,u)} but not in original coordinates, since the homoclinic orbit goes through collision. Nevertheless, the Moser approach can be easily adapted to this setting. We follow the approach in \cite{guardia2024oscillatorymotionsparabolicorbits}, adapted to Levi-Civita coordinates. That is, we fix $\iota>0$ small and consider the points of intersection $p_-^*$,  $p_+^*$. Then, we define
\begin{itemize}
    \item $ D^>$ as the  points in $ \Sigma_\nu^>$, $\iota$-close to $p_-^*$, whose forward orbit hits $ \Sigma_\nu^<$.
    \item  $ D^<$ as the points in $ \Sigma_\nu^<$, $\iota$-close to $p_+^*$, whose backward orbit hits $ \Sigma_\nu^>$.
\end{itemize}
To characterize the domains $ D^>$ and $ D^<$, consider first two $\iota$-neighborhoods of the intersection points $p_-^*$ and $p_+^*$ in $ \Sigma_\nu$, which we denote by $B_\iota(p_-^*)\subset  \Sigma_\nu^>$ and $B_\iota(p_+^*)\subset  \Sigma_\nu^<$ respectively. We denote by 
\begin{equation*}
\begin{aligned}
\Lambda_\infty^{s,u}(\mu) &:= \Lambda_\infty^{s,u}(\hat{\Theta}_0^* (\mu),\mu), \quad R_\infty^{s,u}(\theta) := R_\infty^{s,u}(\theta,\hat{\Theta}_0^* (\mu)),\quad \Theta_\infty^{s,u}(\theta) := \Theta_\infty^{s,u}(\theta,\hat{\Theta}_0^* (\mu))
\end{aligned}
\end{equation*}
The curve $ \Lambda_\infty^s(\mu)$ intersects $B_\iota(p_-^*)$, dividing it into two connected open regions. By definition, $ \Lambda_\infty^s(\mu)$ corresponds to points whose forward orbits are parabolic. As a result, one of these regions contains points whose forward orbits escape to infinity with positive speed, that is hyperbolic orbits, while the other region contains points whose forward orbits return and intersect the section $ \Sigma_\nu^<$. The latter region, corresponding to the domain $ D^>$, is defined as points with radial momentum below the curve $ \Lambda_\infty^s(\mu)$. Namely
\begin{equation*}
     D^> = \left\{(\theta,\Theta) \in B_\iota(p_-^*) \colon R(\mu^\gamma,\theta,\Theta;h^*) < R_\infty^{s}(\theta)\right\} \subset  \Sigma_\nu^>,
\end{equation*}
where $R(\mu^\gamma,\theta,\Theta;h^*)$ is obtained from the Hamiltonian \eqref{eqn: Hamiltonian in polar coordinates centered at J} and $R_\infty^{s}(\theta) > 0$  is defined in \eqref{eqn: R_infs and Theta_infs}.

Similarly, the curve $ \Lambda_\infty^u(\mu)$  separates $B_\iota(p_+^*)$ into two connected components within $B_\iota(p_+^*)$. In this case, the domain $ D^<$ corresponds to the points with radial momentum smaller (in absolute value) than those on $ \Lambda_\infty^u(\mu)$. That is,
\begin{equation*}
     D^< = \left\{(\theta,\Theta) \in B_\iota(p_+^*) \colon R(\mu^\gamma,\theta,\Theta;h^*) > R_\infty^u(\theta)\right\}\subset  \Sigma_\nu^<,
\end{equation*}
where $R_\infty^{u}(\theta) < 0$ is defined in \eqref{eqn: R_infu and Theta_infu}.

As the points in $ D^<$ are close  to $p_+^*$, one can map this domain to $ \Sigma_\nu^>$ by means of the map $f$ in \eqref{eqn: composition triple intersection}. This result gives the domain $\tilde D^<$, defined as
\begin{equation*}
\tilde D^< = \left\{(\theta,R) \in B_\iota(p_-^*) \colon R < R_\infty^{u,>}(\theta)\right\},
\end{equation*}
where $R_\infty^{u,>}(\theta)$ corresponds to the image of $R_\infty^u(\theta)$ under $f$.

Alternatively, by the conservation of the Hamiltonian \eqref{eqn: Hamiltonian in polar coordinates centered at J}, the domains $ D^>$ and $\tilde D^<$ can be also defined as
\begin{equation*}
\begin{aligned}
 D^> = \left\{(\theta,\Theta) \in B_\iota(p_-^*) \colon \Theta < \Theta_\infty^{s}(\theta)\right\}\subset  \Sigma_\nu^>,\quad \tilde D^< = \left\{(\theta,\Theta) \in B_\iota(p_-^*) \colon \Theta < \Theta_\infty^{u,>}(\theta)\right\}\subset  \Sigma_\nu^>,
\end{aligned}
\end{equation*}
where $\Theta_\infty^{u,>}(\theta):=\Theta_\infty^{u,>}(\theta,\hat\Theta_0^*(\mu))$ is introduced in \eqref{eqn: Curves as graphs of theta in coordinates (theta,Theta)}.

Recall that Proposition \ref{prop: triple intersection J} gives the ``ordering'' of the invariant manifolds at the triple intersection (see \eqref{eqn: angles trasnversality} and \eqref{eqn: Inequality angles}). Hence the transition of the points from $ D^<$ to $\tilde D^<$ has not gone through collision with Jupiter. Therefore, the domain $\mathcal{ D}= D^> \cap \tilde D^<$ contains points in $ \Sigma_\nu^>$, $\iota$-close to $p_-^*$, whose backward and forward orbit hit $ \Sigma_\nu^>$.  Hence, proceeding as in \cite{guardia2024oscillatorymotionsparabolicorbits} ($\S$ 6, pp. 39--40), one can complete the proof. 
\end{proof}

\subsection{Existence of the triple intersection}\label{subsec: existence triple intersection}
Fix  $h = -\hat{\Theta}_0 \in \left(-\frac{1+\sqrt 3}{3},\frac{\sqrt 3-1}{3}\right)$, $\nu \in \left(\frac{3}{11},\frac 1 3\right)$ and $\mu > 0$ small enough so that Theorem \ref{thm: distance transversality} holds and consider the point
\begin{equation}\label{eqn: point p<}
    p_+ = p_+(\hat{\Theta}_0,\mu)= \left(\theta_+(\hat{\Theta}_0,\mu), \Theta_+(\hat{\Theta}_0,\mu)\right) \in  \Lambda_{\mathcal J}^{+}(\mu)\transv  \Lambda_\infty^u(\hat \Theta_0,\mu) \in \Sigma_\nu^<,
\end{equation}
with $\Theta_+(\hat{\Theta}_0,\mu) = \Theta_\infty^u(\theta_+(\hat{\Theta}_0,\mu), \hat{\Theta}_0) =  \Theta_{\mathcal J}^+(\theta_+(\hat{\Theta}_0,\mu))$, where $\Theta_\infty^u(\theta,\hat\Theta_0)$ and $\Theta_{\mathcal J}^+(\theta)$ are defined in \eqref{eqn: R_infu and Theta_infu} and \eqref{eqn: (R,Theta) in curves ECO-section} respectively.

Since the changes of coordinates $\psi$, $\Xi$ and $\Upsilon$ (defined in \eqref{eqn: Change of coordinates Levi-Civita}, \eqref{eqn: straightening (z,w)-(s,u)} and \eqref{eqn: (q,p) - (r,theta,R,Theta)}, respectively) are diffeomorphisms, we can apply Proposition \ref{prop: transition map close to collision} so the map $f$ in \eqref{eqn: composition triple intersection} is well defined and 
\[p_-^u = \left(\theta_-^u,\Theta_-^u\right) = f\circ p_+ \in  \Lambda_{\mathcal J}^{-}(\mu)\subset \Sigma_\nu^>.\] 
Therefore $ \Theta_-^u =  \Theta_{\mathcal{J}}^-(\theta_-^u)$ (see \eqref{eqn: (R,Theta) in curves ECO-section}) and  $p_-^u \in  \Lambda_{\mathcal J}^{-}(\mu)\transv W_\mu^u(\Alpha_{\hat \Theta_0})$. The $\theta$-component of both points can be related from \eqref{eqn: theta+-}  as
\[\theta_-^u = \theta_-^u(\hat\Theta_0,\mu) = \theta_+(\hat\Theta_0,\mu) + \mathcal{O}\left(\mu^{\frac{11\nu-3}{8}}\right)\]
since $\nu < \frac 1 3$. To guarantee the triple intersection we need that $p_-^u = p_-(\hat{\Theta}_0,\mu) \in  \Lambda_{\mathcal J}^{-}(\mu)\transv  \Lambda_\infty^s(\hat \Theta_0,\mu)$ defined in \eqref{eqn:p-Thm62}. Namely, we look for $\hat{\Theta}_0$ such that
\begin{equation}\label{eqn: triple intersection}
    \theta_-^u(\hat{\Theta}_0,\mu)= \theta_-(\hat{\Theta}_0,\mu)
\end{equation}
which is equivalent to solve
\begin{equation}\label{eqn: implicit equation intersection}
    2\arctan\left(\frac{\hat{\Theta}_0-1}{\sqrt{2-\hat{\Theta}_0^2}}\right) + \mathcal{O}\left(\mu^{\frac{11\nu-3}{8}}\right) = 0.
\end{equation}
%
The Implicit Function Theorem ensures that there exists $\mu_0 > 0$ such that, for all $\mu \in (0,\mu_0)$, there exists a unique $\hat{\Theta}_0^*(\mu)$ with $\hat{\Theta}_0^*(0) = 1$ satisfying \eqref{eqn: implicit equation intersection} (and therefore \eqref{eqn: triple intersection}). Since $\hat{\Theta}_0^*(\mu) \in \left(\frac{1-\sqrt 3}{3},\frac{1+\sqrt 3}{3}\right)$, the point
\begin{equation}\label{eqn: point triple intersection}
    p_-^* = \left(\theta_-(\hat\Theta_0^*(\mu),\mu), \Theta_-(\hat\Theta_0^*(\mu),\mu)\right),\quad \Theta_-(\hat\Theta_0^*(\mu),\mu) =\Theta_{\mathcal J}^-(\theta_-(\hat\Theta_0^*(\mu),\mu)) 
\end{equation}
belongs to $W_\mu^u(\Alpha_{\hat{\Theta}_0})\cap  \Lambda_{\infty}^s(\hat \Theta_0,\mu)\cap {\Lambda}_{\mathcal J}^{-}(\mu)$ at
\begin{equation}\label{eqn: energy triple intersection}
    h^* = h(\mu) = -\hat{\Theta}_0^*(\mu) = -1+\mathcal{O}\left(\mu^{\frac{11\nu-3}{8}}\right).
\end{equation}
%
\subsection{Transversality and ordering of the triple intersection}\label{subsec: transversality and ordering triple intersection}
From now on we consider $h^*$ as in \eqref{eqn: energy triple intersection} and use the expressions for $\theta_-$ and $\theta_+$ given in Theorem \ref{thm: distance transversality}. We denote by
\begin{equation}\label{eqn: notation transversality}
    \theta_\mp = \theta_\mp(\hat\Theta_0^*(\mu),\mu) = \mathcal{O}\left(\mu^{\frac{11\nu-3}{8}}\right),\quad d_\mp(\theta)=d_\mp(\theta, \hat{\Theta}_0^*(\mu)),\quad
    \Theta_\infty^{u,s}(\theta) = \Theta_\infty^{u,s}(\theta,\hat{\Theta}_0^*(\mu)),
\end{equation}
%
where $\theta_\mp(\hat\Theta_0^*(\mu),\mu)$ are defined in \eqref{eqn: point p<} and \eqref{eqn: point triple intersection}, the distances are introduced in \eqref{eqn: approx distances} and $\Theta_\infty^{u,s}(\theta,\hat\Theta_0^*(\mu))$ are defined in \eqref{eqn: R_infu and Theta_infu} and \eqref{eqn: R_infs and Theta_infs}.
%
%
The following lemma (whose proof is done in Appendix \ref{appendix: From Levi-Civita coordinates to synodical polar coordinates}) proves the existence of a $C^1$-curve $\Gamma_\infty^{u,>}\subset \Sigma_\nu^>$ related to $\Lambda_\infty^u(\hat{\Theta}_0^*(\mu),\mu)\subset \Sigma_\nu^<$ via the map $f$ in \eqref{eqn: composition triple intersection} and provides $C^1$ estimates for its graph parameterization in coordinates $(\theta,\Theta)$.

\begin{lemma}\label{lemma: tangent vector at Sigma delta,h as a graph}
Fix $\iota > 0$ small and denote by $B_\iota(p_+)$ a $\iota$-neighborhood of the point $p_+$ in \eqref{eqn: point p<}. Consider the $C^1$-curve given by 
\begin{equation}\label{eqn: curve gammainf<}
\Gamma_\infty^{u,<}:=\Lambda_\infty^u(\hat{\Theta}_0^*(\mu),\mu)\cap B_\iota(p_+) \subset \Sigma_\nu^<.
\end{equation}
Then the map $f$ in \eqref{eqn: composition triple intersection} maps the curve $\Gamma_\infty^{u,<}$ into a $C^1$-curve of the form
\begin{equation}
    \Gamma_\infty^{u,>} := f\circ \Gamma_\infty^{u,<} \subset \Sigma_\nu^>
\end{equation}
which admits the graph parameterization $\gamma_\infty^{u,>}(\theta) = (\theta,\Theta_\infty^{u,>}(\theta,\hat{\Theta}_0^*(\mu)))$ defined in \eqref{eqn: Curves as graphs of theta in coordinates (theta,Theta)}. For $\theta \in (\theta_- - \iota,\theta_-+\iota)$, the following estimates hold:
\begin{equation}\label{eqn: Expression of Thetau> infty through f}
    \begin{aligned}
      \Theta_\infty^{u,>}(\theta,\hat{\Theta}_0^*(\mu)) =\mu^{\nu}\left(\sin\theta + \mathcal{O}\left(\mu^{\frac{11\nu-3}{8}}\right)\right),\quad \Theta_\infty^{u,>}{}'(\theta_-,\hat{\Theta}_0^*(\mu)) =  \mu^{\nu}\left(1+ \mathcal{O}\left(\mu^{\frac{11\nu-3}{8}},\mu^{1-3\nu}\right)\right).
    \end{aligned}
\end{equation}
\end{lemma}
Denote by $\Theta_\infty^{u,>}{}'(\theta_-):=\Theta_\infty^{u,>}{}'(\theta_-,\hat{\Theta}_0^*(\mu))$. Relying on \eqref{eqn: Curves as graphs of theta in coordinates (theta,Theta)}, the transversality condition of the triple intersection (see Definition \ref{def:tripleintersectionJ}) is guaranteed once we prove that
\begin{equation}\label{eqn:ineq}
    \begin{aligned}
        \Theta_{\mathcal{J}}^-{}'(\theta_-) - \Theta_\infty^s{}'(\theta_-) \neq 0,\qquad \Theta_{\mathcal J}^-{}'(\theta_-) - \Theta_\infty^{u,>}{}'(\theta_-) \neq 0,\qquad \Theta_\infty^{s}{}'(\theta_-) - \Theta_\infty^{u,>}{}'(\theta_-) \neq 0.
    \end{aligned}
\end{equation}
Since $\nu\in\left(\frac{3}{11},\frac 1 3\right)$, by \eqref{eqn: (R,Theta) in curves ECO-section}, \eqref{eqn: R_infs and Theta_infs}, \eqref{eqn: energy triple intersection}, \eqref{eqn: notation transversality} and \eqref{eqn: Expression of Thetau> infty through f} we have
\begin{equation*}
\begin{aligned}
    \Theta_{\mathcal J}^-{}'(\theta_{-}) &= \mathcal{O}\left(\mu^{\frac{19\nu-3}{8}}\right),\\
    \Theta_{\infty}^s{}'(\theta_-) &= - \Theta_\infty^{u}{}'(-\theta_-) = -\mu^{\nu}\left(\sin\theta_-\left(\hat \Theta_0^* - 1\right) + \cos\theta_-\sqrt{2-\hat\Theta_0^*{}^2}\right) + \mathcal{O}\left(\mu^{1-2\nu}\right)\\
    &= - \mu^{\nu}\left(1+ \mathcal{O}\left(\mu^{\frac{11\nu-3}{8}},\mu^{1-3\nu}\right)\right).
\end{aligned}
\end{equation*}
Therefore, we rewrite \eqref{eqn:ineq} as
\begin{equation}\label{eqn: difference derivative transversality}
\begin{aligned}
    \Theta_{\mathcal{J}}^-{}'(\theta_-) - \Theta_\infty^s{}'(\theta_-)&=   \mu^{\nu}\left(1+ \mathcal{O}\left(\mu^{\frac{11\nu-3}{8}},\mu^{1-3\nu}\right)\right) \neq 0,\\
    \Theta_{\mathcal J}^-{}'(\theta_-) - \Theta_\infty^{u,>}{}'(\theta_-) &= -\mu^{\nu}\left(1+ \mathcal{O}\left(\mu^{\frac{11\nu-3}{8}},\mu^{1-3\nu}\right)\right) \neq 0,\\
     \Theta_\infty^{s}{}'(\theta_-) - \Theta_\infty^{u,>}{}'(\theta_-) &=  -2\mu^{\nu}\left(1+ \mathcal{O}\left(\mu^{\frac{11\nu-3}{8}},\mu^{1-3\nu}\right)\right)\neq 0.
\end{aligned}
\end{equation}
Following \eqref{eqn: Curves as graphs of theta in coordinates (theta,Theta)} and \eqref{eqn: difference derivative transversality}, we compute the angles in \eqref{eqn: angles trasnversality} as follows
\begin{equation*}
\begin{aligned}
    \sin A =& \frac{\Theta_\infty^s{}'(\theta_-) - \Theta_\infty^{u,>}{}'(\theta_-)}{\sqrt{\left(1+ \left(\Theta_\infty^s{}'(\theta_-)\right)^2\right)\left(1+ \left(\Theta_\infty^{u,>}{}'(\theta_-)\right)^2\right)}} = -2\mu^{\nu}\left(1+ \mathcal{O}\left(\mu^{\frac{11\nu-3}{8}},\mu^{1-3\nu}\right)\right),\\
    \sin B =& \frac{\Theta_{\mathcal J}^-{}'(\theta_-) - \Theta_\infty^{u,>}{}'(\theta_-)}{\sqrt{\left(1+ \left(\Theta_{\mathcal J}^-{}'(\theta_-)\right)^2\right)\left(1+ \left(\Theta_\infty^{u,>}{}'(\theta_-)\right)^2\right)}} = -\mu^{\nu}\left(1+ \mathcal{O}\left(\mu^{\frac{11\nu-3}{8}},\mu^{1-3\nu}\right)\right),
\end{aligned}
\end{equation*}
leading to \eqref{eqn: Inequality angles}. See Figure \ref{fig: trans intersection} to see the relative position of the curves.

\section*{Declarations}


\begin{itemize}
\item \textbf{Funding}: 
This work was partially supported by the grant PID-2021-122954NB-100 funded by MCIN/AEI/10.13039/501100011033 and ``ERDF A way of making Europe''. 

M.Guardia has been supported by the European Research Council (ERC) under the
European Union’s Horizon 2020 research and innovation programme (grant agreement
No. 757802). M.Guardia was also supported by the Catalan Institution for Research and
Advanced Studies via an ICREA Academia Prize 2018 and 2023.

J.Lamas has been supported by grant 2021 FI\_B 00117 under the European Social Fund.

Tere M.Seara has been supported by the Catalan Institution for Research and Advanced Studies via an ICREA Academia Prize 2018.

This work was also supported by the Spanish State Research Agency through the Severo Ochoa and Mar\'ia de Maeztu Program for Centers and Units of Excellence in R\&D(CEX2020-001084-M).

\item \textbf{Competing interests}: The authors have no competing interests to declare that are relevant to the content of this article. All authors certify that they have no affiliations with or involvement in any organization or entity with
any financial interest or non-financial interest in the subject matter or materials discussed in this manuscript.
The authors have no financial or proprietary interests in any material discussed in this article. Indirectly, the
authors hope to gain reputation in the scientific community via the publication of this article.
\item \textbf{Data availability} : The authors declare that all the data supporting the findings of this study are available within the paper.
\end{itemize}

\appendix

\section{Proof of Proposition \ref{prop: general ECO-orbits intersect section (z,w)}}\label{appendix: Proof of proposition 3.2}
We write the equations of motion \eqref{eqn: Eq motion in coordinates (s,u)} as
\begin{equation*}
    \begin{aligned}
        s' &= -s + Q^s(s,u),\quad u' &= u + Q^u(s,u).
    \end{aligned}
\end{equation*}
where $Q^s(s,u)$ and $Q^u(s,u)$ are analytic in $\Xi(\mathrm B)$, defined in \eqref{eqn: straightening (z,w)-(s,u)}. They satisfy 
\begin{alignat}{2}
    Q^s(s,u) =&\ s \mathcal O_2(s,u),\quad Q^u(s,u) &&= u \mathcal O_2(s,u),\\
    \partial_{s}Q^s(s,u) =&\ \mathcal{O}_2(s,u), \quad \partial_{s}Q^u(s,u) &&= u\mathcal{O}(s,u),\label{eqn: Estimate H}\\
    \partial_{u}Q^s(s,u) =&\ s\mathcal{O}(s,u), \quad \partial_{u}Q^u(s,u) &&= \mathcal{O}_2(s,u).
\end{alignat}
We provide the proof for the ejection orbit $(s_\mu(\tau,\beta), u_\mu(\tau,\beta))$ (with $\mathrm{Re}(\tau) >  0$) since the proof for the collision orbit is analogous. To simplify notation, throughout the proof we denote by $C >0$ any constant independent of $\mu$ and by
\begin{equation}\label{eqn: notationA}
\begin{aligned}
\tau &:= \tau(\beta),&\quad (s(\tau),u(\tau)) &:= (s_\mu(\tau,\beta),u_\mu(\tau,\beta)),\\
\tau_{\mathrm{lin}} &:=\tau_{\mathrm{lin}}(\beta),&\quad (s_{\mathrm{lin}}(\tau),u_{\mathrm{lin}}(\tau)) &:= ( s_{\mathrm{lin}}(\tau,\beta), u_{\mathrm{lin}}(\tau,\beta)),
\end{aligned}
\end{equation}
where $\tau_{\mathrm{lin}}(\beta)$ and $( s_{\mathrm{lin}}(\tau,\beta), u_{\mathrm{lin}}(\tau,\beta))$ are defined in \eqref{eqn: Time lineal Levi-Civita} and \eqref{eqn: slin ulin} respectively for $\beta \in \mathds{T}_{\sigma_0}$ in \eqref{eqn: complexdomainbeta}.

The initial condition $ (s_0,u_0):= (s_\mu(0,\beta),u_\mu(0,\beta))$ in \eqref{eqn: Collision in coord (s,u)} is of the form
\begin{equation}\label{eqn: Init condition (s0,u0)}
    \begin{aligned}
        s_{01} &= -\frac{\xi\mu^{\frac 1 2}}{\sqrt 2} \cos\beta + \mathcal{O}_{\frac 5 2}(\mu),&\quad s_{02} &= -\frac{\xi\mu^{\frac 1 2}}{\sqrt 2} \sin\beta + \mathcal{O}_{\frac 5 2}(\mu),\\
        u_{01} &= \frac{\xi \mu^{\frac 1 2}}{\sqrt 2}\cos\beta + \mathcal{O}_{\frac 5 2}(\mu),&\quad u_{02} &= \frac{\xi\mu^{\frac 1 2}}{\sqrt 2}\sin\beta +\mathcal{O}_{\frac 5 2}(\mu),
    \end{aligned}
\end{equation}
with $\beta \in \mathds{T}_{\sigma_0}$ and therefore satisfies the uniform estimate
\begin{equation}\label{eqn: Estimate (s0,u0)}
    |s_{0i}| \leq C\mu^{\frac 1 2},\quad |u_{0i}| \leq C\mu^{\frac 1 2},\quad i=1,2.
\end{equation}
Fix $\eta \in (0,3/2]$ and $a_0 \in (0,1)$. Since $|\operatorname{Im}(\tau_{\mathrm{lin}})| = \mathcal O(\mu^{1-\gamma})$ for $\mu > 0$ small enough (see \eqref{eqn: Time lineal Levi-Civita}), we can always assume that $|\operatorname{Im}(\tau_{\mathrm{lin}})|< \frac{a_0}{2}$. Let $\mathrm B$ be the ball from Proposition \ref{prop: Straightening invariant manifolds collision LC} (where the estimates in \eqref{eqn: Estimate H} are still valid). Denote by 
\[T^* = \sup \left\{ T \in (0,+\infty)\colon (s(\tau),u(\tau)) \in \mathrm B \text{ for all  } \tau \in \mathbf S_T\right\},\]
where $\mathbf S_T$ is the complex strip
\[\mathbf{S}_T = \left\{\tau \in \mathds{C}\colon \mathrm{Re}(\tau) \in [0,T],\; |\mathrm{Im}(\tau)| \le a_0\right\}.\]
We study the evolution of the orbit $(s(\tau),u(\tau))$ for $\tau \in \mathbf D$ with
\begin{equation}\label{eqn: Deta}
\mathbf D =  \left\{\tau \in \mathds{C}\colon \mathrm{Re}(\tau)\in[0,\eta \min \{ \mathrm{Re}(\tau_{\mathrm{lin}}),T^*\}],\; |\mathrm{Im}(\tau)|\leq a_0\right\},
\end{equation}
Applying the variation of constants formula we have for $\tau \in \mathbf D$
\begin{equation}\label{eqn: sol su}
    s(\tau) = e^{-\tau}(s_{0}+\tilde s(\tau)), \quad u(\tau) = e^{\tau}(u_{0}+\tilde u(\tau)),
\end{equation}
with $(\tilde s(\tau), \tilde u(\tau))$ satisfying
\begin{equation}\label{eqn: Variation of constants}
    \begin{aligned}
        \tilde s(\tau) &= (\tilde s_1(\tau), \tilde s_2(\tau))^\top = \int_0^{\tau} e^z Q^s\left(e^{-z}(s_0+\tilde s(z)), e^z(u_0+\tilde u(z))\right)dz,\\
        \tilde u(\tau) &= (\tilde u_1(\tau),\tilde u_2(\tau))^\top = \int_0^{\tau} e^{-z} Q^u\left(e^{-z}(s_0+\tilde s(z)), e^z(u_0+\tilde u(z))\right)dz.
    \end{aligned}
\end{equation}
We analyze the evolution of $\tilde s(\tau)$, $\tilde u(\tau)$ for $\tau\in \mathbf D$ by a fixed point argument. For fixed $(s_0,u_0)$ in \eqref{eqn: Init condition (s0,u0)}, we denote by $\chi$ the Banach space of analytic functions in $\accentset{\circ}{\mathbf D}$ which extends continuously to $\partial \mathbf D$. Hence, the space $\chi^4$ is a Banach space under the norm 
\[ \|f\|_{\chi^4}= \|(f_1^s, f_2^s,f_1^u,f_2^u)\|_{\chi^4} = \max\left(\|f_1^s\|,\|f_2^s\|,\|f_1^u\|,\|f_2^u\|\right),\]
where $\|\cdot\|$ denotes the supremum norm in $\mathbf D$. We define the following operator acting on $\chi^4$
\begin{equation}\label{eqn: Operator F fixed point}
    \begin{aligned}
        \mathcal F\colon \chi^4 &\to \chi^4\\
        f=(f^s,f^u) &\mapsto \mathcal F(f) = \left(\mathcal F^s(f),\mathcal F^u(f)\right),
    \end{aligned}
\end{equation}
such that
\begin{equation*}
    \begin{aligned}
        \mathcal F^s(f)(z)&=\int_0^{\tau} e^z Q^s\left(e^{-z}(s_0+f^s(z)), e^z(u_0+f^u(z))\right)dz,\\
        \mathcal F^u(f)(z) &= \int_0^{\tau} e^{-z} Q^u\left(e^{-z}(s_0+f^s(z)), e^z(u_0+f^u(z))\right)dz.
    \end{aligned}
\end{equation*}
The solutions $(\tilde s(\tau), \tilde u (\tau))$ defined in \eqref{eqn: Variation of constants} are fixed points of $\mathcal F$.

A fixed point theorem argument allows us to provide estimates for $\tilde s(\tau)$ and $\tilde u(\tau)$. Relying on \eqref{eqn: Estimate H} and \eqref{eqn: Estimate (s0,u0)} we obtain the following estimates
\begin{equation*}
    \begin{aligned}
        \|\mathcal F^s(0)\| &= \underset{\tau \in \mathbf D}{\sup} \left(\left|\int_0^{\tau} e^z Q^s(e^{-z}s_0,e^{z}u_0) dz\right|\right)\\
        &\leq\underset{\tau \in \mathbf D}{\sup} \left(\int_0^\tau\left(C\mu^{\frac 3 2} \left|e^{-z}+e^z\right|^2\right)dz\right) \leq C \mu^{\frac 3 2} e^{2\eta\mathrm{Re}(\tau_{\mathrm{lin}})} \leq C \mu^{\frac{3}{2}-\eta+\eta\gamma}, \\
        \|\mathcal F^u(0)\|&= \underset{\tau \in \mathbf D}{\sup} \left(\left|\int_0^{\tau} e^{-z} Q^u(e^{-z}s_0,e^{z}u_0) dz\right|\right)\\
        &\leq \underset{\tau \in \mathbf D}{\sup} \left(\int_0^{\tau} \left(C \mu^{\frac 3 2} \left|e^{-z}+e^z\right|^2\right)dz\right)\leq C \mu^{\frac 3 2} e^{2\eta\mathrm{Re}(\tau_{\mathrm{lin}})} \leq C \mu^{\frac{3}{2}-\eta+\eta\gamma}.\\
    \end{aligned}
\end{equation*}
Therefore
\begin{equation}\label{eqn: epsilon(n)}
    \|\mathcal F(0)\|_{\chi^4} \leq C \mu^{\frac{3}{2}-\eta+\eta\gamma} = \varepsilon(\eta).
\end{equation}
We impose $\eta\in \left(1,\frac 3 2\right)$ and $\varepsilon(\eta) \ll |s_0| = |u_0| = \mathcal{O}(\mu^{\frac 1 2})$ satisfying \eqref{eqn: Estimate (s0,u0)}. That is, we look for $\eta,\gamma$ such that
\begin{equation}\label{eqn: cond eps}
1-\eta+\eta\gamma > 0.
\end{equation}
This is achieved for $\gamma \in \left(1-\frac{1}{\eta},1\right)$. 

We define the ball $B_{2\varepsilon} = \left\{f \in \chi^4\colon \|f\|_{\chi^4} < 2\varepsilon(\eta)\right\}$ and we prove that the operator $\mathcal F$ is Lipschitz in $B_{2\varepsilon}$. Consider $f,g\in B_{2\varepsilon}$. For $t\in [0,1]$ the function 
\[h^t(\tau) = \left(h_s^t(\tau),h_u^t(\tau)\right) = t\cdot \left(f_s^t(\tau),f_u^t(\tau)\right) + (1-t)\cdot \left(g_s^t(\tau),g_u^t(\tau)\right)\]
satisfies, for $\tau \in \mathbf D$,
\[|s_{0}+h_{s}^t| \leq C\mu^{\frac 1 2}, \quad |u_{0}+h_{u}^t| \leq C\mu^{\frac 1 2}.\]
Therefore, relying on \eqref{eqn: Estimate H} we obtain 
\begin{equation*}
    \begin{aligned}
        \big|Q&^s\left(e^{-\tau}(s_0+f^s), e^{\tau}(u_0+f^u)\right)-Q^s\left(e^{-\tau}(s_0+g^s), e^{\tau}(u_0+g^u)\right)\big|\\
        \leq& \int_0^1 \left(\left|\partial_{s}Q^s\left(e^{-\tau}(s_0+h_s^t), e^{\tau}(u_0+h_u^t)\right)\right|\cdot e^{-\tau}|f^s-g^s|\right.\\
        &\left.+ \left|\partial_{u}Q^s\left(e^{-\tau}(s_0+h_s^t), e^{\tau}(u_0+h_u^t)\right)\right|\cdot e^{\tau}|f^u-g^u|\right) dt \leq C e^{\tau}\mu\|f-g\|_{\chi^4},\\
        &\\
        \big|Q&^u\left(e^{-\tau}(s_0+f^s), e^{\tau}(u_0+f^u)\right)-Q^u\left(e^{-\tau}(s_0+g^s), e^{\tau}(u_0+g^u)\right)\big|\\
        \leq& \int_0^1 \left(\left|\partial_{s}Q^u\left(e^{-\tau}(s_0+h_s^t), e^{\tau}(u_0+h_u^t)\right)\right|\cdot e^{-\tau}|f^s-g^s|\right.\\
        &+\left. \left|\partial_{u}Q^u\left(e^{-\tau}(s_0+h_s^t), e^{\tau}(u_0+h_u^t)\right)\right|\cdot e^{\tau}|f^u-g^u|\right) dt \leq C e^{3\tau}\mu\|f-g\|_{\chi^4}.
    \end{aligned}
\end{equation*}
Hence we obtain the following estimate for the norms $\|\mathcal F^s(f) - \mathcal F^s(g)\|$ and $\|\mathcal F^u(f) - \mathcal F^u(g)\|$
\begin{equation*}
    \begin{aligned}
        \|\mathcal F^{s,u}(f) - \mathcal F^{s,u}(g)\| &\leq\underset{\tau \in \mathbf D}{\sup} \left(\int_0^{\tau} e^{\pm z} |Q^{s,u}(f)-Q^{s,u}(g)|dz\right)\leq C \mu e^{2\eta\mathrm{Re}(\tau_{\mathrm{lin}})}\|f-g\|_{\chi^4} \\
        &\leq C \mu^{1-\eta+\eta\gamma}\|f-g\|_{\chi^4},\\
    \end{aligned}
\end{equation*}
so that
\[\|\mathcal F(f) - \mathcal F(g)\|_{\chi^4} \leq C \mu^{1-\eta+\eta\gamma}\|f-g\|_{\chi^4}\]
corresponds to the Lipschitz constant. Hence, for $\eta \in \left(1,\frac 3 2\right)$, $\gamma \in \left(1-\frac{1}{\eta},1\right)$ and $\mu > 0$ small enough, the operator $\mathcal F$ is contractive and satisfies $\mathcal F(B_{2\varepsilon}) \subset B_{2\varepsilon}$. Then the Banach fixed point theorem ensures that there exists a unique fixed point in $B_{2\varepsilon}$ for the operator $\mathcal F$ in \eqref{eqn: Operator F fixed point}, which we denote by $(\tilde s(\tau), \tilde u(\tau))$. It satisfies, for $\tau \in \mathbf D$ in \eqref{eqn: Deta}, the following estimates 
\begin{equation}\label{eqn: Estimates tilde s and tilde u}
    \begin{aligned}
        |\tilde s(\tau)| \leq  \|(\tilde s,\tilde u)\|_{\chi^4} < 2\varepsilon(\eta),\quad |\tilde u(\tau)|  \leq \|(\tilde s,\tilde u)\|_{\chi^4}< 2\varepsilon(\eta).
    \end{aligned}
\end{equation}
Hence, for $\eta \in \left(1,\frac 3 2\right)$, $\gamma \in \left(1-\frac{1}{\eta},1\right)$ and $\tau \in \mathbf D$,
\begin{equation}\label{eqn: estimate (s,u) - lineal}
    \begin{aligned}
        s(\tau) = e^{-\tau} \left(s_0 + \mathcal{O}(\varepsilon(\eta))\right),\qquad u(\tau) = e^{\tau}\left( u_0 + \mathcal{O}(\varepsilon(\eta))\right).
    \end{aligned}
\end{equation}
Fix $\eta = \frac{11}{8} \in \left(1, \frac 3 2\right)$ and $\gamma \in \left(1-\frac{1}{\eta},1\right) = \left(\frac{3}{11},1\right)$. Then $\varepsilon(\eta)$ in \eqref{eqn: epsilon(n)} becomes
\begin{equation}\label{eqn: epsfixedpoint}
\varepsilon := \varepsilon\left(\frac{11}{8}\right) = \mathcal O\left(\mu^{\frac{1+11\gamma}{8}}\right).
\end{equation}
Recall that we have simplified the notation and dropped the dependencies on $\beta$ (see \eqref{eqn: notationA}). The next step is to look for $\tau \in \mathbf D$ satisfying $(s(\tau), u(\tau))\in \partial \mathbf B_\gamma$ defined in \eqref{eqn: neighborhood of collision delta mu (s,u)}. Namely, we look for a zero of a function of the form
\begin{equation}\label{eqn: F}
F(\tau) = (s_1(\tau)+u_1(\tau))^2 + (s_2(\tau)+u_2(\tau))^2+ \mathcal{O}_6(s(\tau),u(\tau)) - \mu^{\gamma}.
\end{equation}
We use a fixed point argument. We write $s(\tau)$ and $u(\tau)$ in \eqref{eqn: sol su} as
\[s(\tau) = s_{\mathrm{lin}}(\tau) + e^{-\tau}\tilde s(\tau),\quad u(\tau) = u_{\mathrm{lin}}(\tau) + e^{\tau}\tilde u(\tau),\]
where $s_{\mathrm{lin}}(\tau)$ and $u_{\mathrm{lin}}(\tau)$ are defined in \eqref{eqn: slin ulin}. Let $\tau(z):=\tau_{\mathrm{lin}}+z$, where $\tau_{\mathrm{lin}}$ is defined in \eqref{eqn: Time lineal Levi-Civita} and set
\[
\mathfrak{L}(z):=s_{\mathrm{lin}}(\tau(z))+u_{\mathrm{lin}}(\tau(z)),
\qquad
\mathfrak{R}(z):=e^{-\tau(z)}\,\tilde s(\tau(z))+e^{\tau(z)}\,\tilde u(\tau(z)).
\]
We denote by $(\cdot\mid\cdot)$ the complex bilinear form on $\mathds C^2$
defined for $u=(u_1,u_2)$ and $v=(v_1,v_2)$ by $(u\mid v):=u_1 v_1 + u_2 v_2$. Then by the triangle inequality and Cauchy-Schwarz we have
\begin{equation}\label{eqn: Bound (u|v)}
|(u\mid v)|=\big|u_1 v_1+u_2 v_2\big|
\le |u_1||v_1|+|u_2||v_2|
\le |u|\,|v|.
\end{equation}
We write $F(\tau(z)) = F(\tau_{\mathrm{lin}}+ z) = h_0(z) + h_1(z)$ in \eqref{eqn: F} with
\begin{equation}\label{eqn: h0 and h1}
    \begin{aligned}
        h_0(z) =&\ \bil{\mathfrak{L}(z)}{\mathfrak{L}(z)} - \mu^\gamma,\\
        h_1(z) =&\ 2\bil{\mathfrak{L}(z)}{\mathfrak{R}(z)} + \bil{\mathfrak{R}(z)}{\mathfrak{R}(z)} +\mathcal{O}_6(s(\tau(z)),u(\tau(z))).
    \end{aligned}
\end{equation}
Relying on \eqref{eqn: slin ulin} we obtain
\begin{equation}\label{eqn: h0z'}
\begin{aligned}
    h_0'(z) =&\ 2 \bil{\mathfrak{L}'(z)}{\mathfrak{L}(z)} = 2 \left[\bil{u_\mathrm{lin}(\tau(z))}{u_\mathrm{lin}(\tau(z))} - \bil{s_\mathrm{lin}(\tau(z))}{s_\mathrm{lin}(\tau(z))}\right],\\
    h_1'(z) =&\ 2\bil{\mathfrak L'(z)}{\mathfrak R(z)} + 2\bil{\mathfrak L(z)}{\mathfrak R'(z)} + 2\bil{\mathfrak R(z)}{\mathfrak R'(z)} +\mathcal{O}_6(s(\tau(z)),u(\tau(z))).
\end{aligned}
\end{equation}
We define the operator 
\begin{equation*}
    \mathcal W(z) = z- \frac{h_0(z) + h_1(z)}{h_0'(0)},
\end{equation*}
where $h_0'(0)$ can be computed from Lemma \ref{lemma: Linear ECO-orbits intersect section (z,w)} as
\begin{equation}\label{eqn: h0'(0)}
    \begin{aligned}
        h_0'(0) &= 2 \left[\bil{u_\mathrm{lin}(\tau_{\mathrm{lin}})}{u_\mathrm{lin}(\tau_{\mathrm{lin}})} - \bil{s_\mathrm{lin}(\tau_{\mathrm{lin}})}{s_\mathrm{lin}(\tau_{\mathrm{lin}})}\right]= 2\mu^\gamma+ \mathcal{O}(\mu) \neq 0
    \end{aligned}
\end{equation}
for $\beta \in \mathds{T}_{\sigma_0}$ in \eqref{eqn: complexdomainbeta}. Fixed points of this operator correspond to zeroes of $F$ in \eqref{eqn: F}. 

Relying on Lemma \ref{lemma: Linear ECO-orbits intersect section (z,w)} and equations \eqref{eqn: Estimates tilde s and tilde u} and \eqref{eqn: Bound (u|v)} we obtain the following estimates for the terms of $h_1(0)$ in \eqref{eqn: h0 and h1}
\begin{equation*}
    \begin{aligned}
        |2\bil{\mathfrak{L}(0)}{\mathfrak{R}(0)} | \leq 2|\mathfrak{L}(0)|\cdot |\mathfrak{R}(0)| &\leq C \mu^{\frac \gamma 2}\cdot \mu^{\frac{\gamma-1}{2}}\varepsilon = C \mu^{\gamma -\frac 1 2}\varepsilon,\\
        |\bil{\mathfrak{R}(0)}{\mathfrak R(0)}| + \mathcal{O}_6(s(\tau_{\mathrm{lin}}),u(\tau_{\mathrm{lin}})) &\le C\mu^{\gamma-1}\varepsilon^2,
    \end{aligned}
\end{equation*}
where $\varepsilon$ is defined in \eqref{eqn: epsfixedpoint}. Therefore we have
\begin{equation}\label{eqn: D}
\begin{aligned}
    |\mathcal W(0)| =& \frac{|h_1(0)|}{|h_0'(0)|}\leq C\mu^{-\gamma}\left(\mu^{\gamma -\frac 1 2}\varepsilon + \mu^{\gamma-1}\varepsilon^2\right) \leq C \mu^{-\frac 1 2}\varepsilon \leq C\mu^{\frac{11\gamma-3}{8}} = D_\mu,
\end{aligned}
\end{equation}
which is small since $\gamma \in \left(\frac{3}{11},1\right).$

We define the ball $B_{2D} = \left\{z\in \mathds{C}\colon |z| \leq  2D_\mu\right\}$ (and therefore $\tau(z)=\tau_{\mathrm{lin}} + z \in \mathbf D$ in \eqref{eqn: Deta}). We prove that the operator $\mathcal W$ is Lipschitz in $B_{2D}$. We consider $z_1,z_2\in B_{2D}$ so that 
%
\[|\mathcal W(z_2)-\mathcal W(z_1)| \leq \underset{z\in B_{2D}}{\sup}|\mathcal W'(z)|\cdot |z_2-z_1|\]
with
\begin{equation}\label{eqn: W'}
    \mathcal W'(z) = 1 - \frac{h_0'(z) + h_1'(z)}{h_0'(0)},
\end{equation}
where $h_0'(z)$ and $h_1'(z)$ are defined in \eqref{eqn: h0z'}. Relying on \eqref{eqn: slin ulin}, for $z\in B_{2D}$ we can expand $h_0'(z)$ as
\begin{equation}
\begin{aligned}
    h_0'(z) &= 2\left(\bil{u_\mathrm{lin}(\tau(z))}{u_\mathrm{lin}(\tau(z))} - \bil{s_\mathrm{lin}(\tau(z))}{s_\mathrm{lin}(\tau(z))}\right) \\
    &= 2\left(\bil{u_\mathrm{lin}(\tau_{\mathrm{lin}}}{u_{\mathrm{lin}}(\tau_\mathrm{lin})}\cdot e^{2z}-\bil{s_{\mathrm{lin}}(\tau_{\mathrm{lin}})}{s_{\mathrm{lin}}(\tau_{\mathrm{lin}})}\cdot e^{-2z}\right) \\
    &= h_0'(0) + 4z\left(\bil{u_\mathrm{lin}(\tau_{\mathrm{lin}}}{u_{\mathrm{lin}}(\tau_\mathrm{lin})}+\bil{s_{\mathrm{lin}}(\tau_{\mathrm{lin}})}{s_{\mathrm{lin}}(\tau_{\mathrm{lin}})}\right) + \mathcal{O}_2(z).
\end{aligned}
\end{equation}
Hence from Lemma \ref{lemma: Linear ECO-orbits intersect section (z,w)} we obtain
\begin{equation}\label{eqn: Estimate h_0'(s) - h_0'(0)}
    |h_0'(z)-h_0'(0)| \leq C\mu^\gamma D_\mu,
\end{equation}
where $D_\mu$ is defined in \eqref{eqn: D}.

The estimate of the derivative $h_1'(z)$ in \eqref{eqn: h0z'} requires some estimates for the derivatives $\tilde s_i'(z)$ and $\tilde u_i'(z)$, which can be computed from \eqref{eqn: Variation of constants} as
\begin{equation*}
\begin{aligned}
    \tilde s'(z) = e^{\tau_{\mathrm{lin}}+z}  Q^s(\tilde s(z), \tilde u(z)),\quad \tilde u'(z) =e^{-(\tau_{\mathrm{lin}}+z)}Q^u(\tilde s(z), \tilde u(z)),
\end{aligned}
\end{equation*} 
where $Q^s,Q^u$ are defined in \eqref{eqn: Estimate H}. Relying on \eqref{eqn: Time lineal Levi-Civita} and \eqref{eqn: Estimate (s0,u0)} we obtain for $z \in B_{2D}$
\begin{equation}\label{eqn: Estimates x',y'}
    \begin{aligned}
        |e^{-(\tau_{\mathrm{lin}}+z)}\tilde s'(z)|\leq \left|s(z) + u(z)\right|^2 \leq C\mu^\gamma,\quad |e^{\tau_{\mathrm{lin}}+z}\tilde u'(z)|\leq \left|s(z) + u(z)\right|^2 \leq C\mu^\gamma.
    \end{aligned}
\end{equation}
From \eqref{eqn: Time lineal Levi-Civita}, \eqref{eqn: Estimates tilde s and tilde u}, \eqref{eqn: Bound (u|v)} and \eqref{eqn: Estimates x',y'} we obtain, for $z\in B_{2D}$, the following estimates for the terms of $h_1'(z)$ in \eqref{eqn: h0z'}
\begin{equation*}
\begin{aligned}
\Big|2\left[\bil{\mathfrak L'(z)}{\mathfrak R(z)} + \bil{\mathfrak L(z)}{\mathfrak R'(z)}\right] \Big| &\leq 2\left(|\mathfrak{L}'(z)|\cdot |\mathfrak{R}(z)| + |\mathfrak{L}(z)|\cdot|\mathfrak{R}'(z)|\right)\\
&\leq\ C\mu^{\frac{\gamma}{2}} \cdot \mu^{\frac{\gamma-1}{2}}\varepsilon + C\mu^{\frac{\gamma}{2}}\cdot \left(\mu^{\frac{\gamma-1}{2}}\varepsilon+ 2K\mu^\gamma\right) \\
&\leq C \mu^{\gamma -\frac{1}{2}}\varepsilon \leq C \mu^\gamma \mu^{\frac{11\gamma-3}{8}},\\
\Big|2 \bil{\mathfrak R(z)}{\mathfrak R'(z)} +\mathcal{O}_6(s(\tau_{\mathrm{lin}} + z),u(\tau_{\mathrm{lin}} + z))\Big|
&\leq\ C\left(\mu^{\frac{\gamma-1}{2}}\varepsilon + 2K\mu^\gamma\right) \cdot \mu^{\frac{\gamma-1}{2}}\varepsilon \\
&\leq C \mu^{\gamma-1}\varepsilon^2 \leq C \mu^\gamma \mu^{\frac{11\gamma-3}{8}},
\end{aligned}
\end{equation*}
where $\varepsilon$ is defined in \eqref{eqn: epsfixedpoint}. Therefore we have
\begin{equation*}
    |h_1'(z)| \leq C\mu^\gamma\mu^{\frac{11\gamma-3}{8}},
\end{equation*}
which is small since $\gamma \in \left(\frac{3}{11},1\right)$. Combining this estimate with the ones in \eqref{eqn: h0'(0)} and \eqref{eqn: Estimate h_0'(s) - h_0'(0)} we obtain
\begin{equation*}
\begin{aligned}
    |\mathcal W'(z)| &= \left|1 - \frac{h_0'(0) + h_0'(z) - h_0'(0) + h_1'(z)}{h_0'(0)}\right| \leq \frac{1}{h_0'(0)}\left(\left|h_0'(z)-h_0'(0)\right| + |h_1'(z)|\right) \\
    &\leq C\mu^{-\gamma}\left( C\mu^\gamma D_\mu+ C\mu^\gamma \mu^{\frac{11\gamma-3}{8}}\right) \leq C\mu^{\frac{11\gamma-3}{8}},
\end{aligned}
\end{equation*}
where $D_\mu$ is defined in \eqref{eqn: D}. Hence, the operator $\mathcal W$ is contractive and satisfies $\mathcal W(B_{2D}) \subset B_{2D}$. Then, the Banach fixed point theorem ensures that there exists a unique fixed point in $B_{2D}$ for the operator $\mathcal W$, which we denote by $ z_*$, so the time $\tau_* = \tau_{\mathrm{lin}} +  z_*$ corresponds to a zero of $F$ in \eqref{eqn: F}. Using Cauchy's estimates we have for $\beta \in \mathds{T}_{\sigma_0/2}$
\[|z_*'(\beta)|\leq 2\frac{|z_*(\beta)|}{\sigma_0} \leq \frac{4}{\sigma_0} D_\mu \leq C \mu^{\frac{11\gamma-3}{8}},\]
which leads to \eqref{eqn: Estimate tilde T}.
Evaluating \eqref{eqn: estimate (s,u) - lineal} at $\tau = \tau_*$ leads to  \eqref{eqn: Intersect eco (s,u)-Sigma} and \eqref{eqn: sb ub}, completing the proof.

\section{Proof of Proposition \ref{prop: Parameterization of the invariant manifolds of infinity}}\label{appendix: proposition 5.2}
The proof of Proposition \ref{prop: Parameterization of the invariant manifolds of infinity} is based on the approach proposed in \cite{lochak2003splitting, sauzin2001new}. That is, we take advantage of the fact that the invariant manifolds of infinity are Lagrangian and therefore they can be locally parameterized as graphs of generating functions, which are solutions of the so-called Hamilton-Jacobi equation. 

Hence, we consider the Hamilton-Jacobi equation associated to Hamiltonian \eqref{eqn: Hamiltonian Polar Rotating Coordinates centered at CM J} and we look for functions $\hat{S}^{s,u}(\hat{r},\hat{\theta};\mu,\hat{\Theta}_0)$ such that
\[(\hat{R},\hat{\Theta})= \left(\partial_{\hat{r}}\hat{S}(\hat{r},\hat{\theta};\mu,\hat{\Theta}_0), \partial_{\hat{\theta}}\hat{S}(\hat{r},\hat{\theta};\mu,\hat{\Theta}_0)\right)\]
parameterize the invariant manifolds as a graph. Then, the Hamilton-Jacobi equation reads
\begin{equation}\label{eqn: General HJ}
\hat{\mathcal{H}}_\mu(\hat{r},\hat{\theta},\partial_{\hat{r}}\hat{S},\partial_{\hat{\theta}}\hat{S}) = - \hat{\Theta}_0.
\end{equation}
Recall that we put $-\hat{\Theta}_0$ in the right hand side since it corresponds to the  the level of energy of the periodic orbit $\Alpha_{\hat \Theta_0}$.

For the unperturbed Hamiltonian, that is, considering $\mu = 0$ (see \eqref{eqn: potential in polar rotating coordinates centered at CM J}), this equation simply reads
\[\frac 1 2 \left(\partial_{\hat{r}}\hat{S}^2 + \frac{\partial_{\hat{\theta}}\hat{S}^2}{\hat{r}^2}\right) - \partial_{\hat{\theta}}\hat{S} - \frac{1}{\hat{r}} = - \hat{\Theta}_0.\]
It has two solutions of the form
\begin{equation}\label{eqn: S0}
\hat{S}_0^{s,u}(\hat{r},\hat{\theta}; \hat{\Theta}_0) = \hat{f}^{s,u}(\hat{r};\hat{\Theta}_0) + \hat{\Theta}_0\hat{\theta}
\end{equation}
where $\hat{f}^{s,u}(\hat{r};\hat{\Theta}_0)$ satisfy $\partial_{\hat{r}}\hat{f}^{s,u}(\hat{r};\hat{\Theta}_0) = \hat R_0^{s,u}(\hat r, \hat \Theta_0)$ and $\hat R_0^{s,u}(\hat r, \hat \Theta_0)$ are defined in \eqref{eqn: (R,Theta) infty mu=0}.

We look for solutions of \eqref{eqn: General HJ} close to \eqref{eqn: S0}. We write $\hat{S}^{s,u} = \hat{S}_0^{s,u} + \hat{S}_1^{s,u}$. Then the equation for $\hat{S}_1^{s,u}$ becomes
\[\partial_{\hat{r}}\hat{f} \partial_{\hat{r}}\hat{S}_1 + \frac 1 2 \partial_{\hat{r}}\hat{S}_1^2 +\left(\frac{\hat{\Theta}_0}{\hat{r}^2} - 1\right)\partial_{\hat{\theta}}\hat{S}_1 + \frac{1}{2\hat{r}^2}\partial_{\hat{\theta}}\hat{S}_1^2 - \hat{V}(\hat{r},\hat{\theta};\mu) = 0.\]
To look for solutions of this equation we reparameterize the variables $(\hat r,\hat \theta)$ through the unperturbed separatrix in \eqref{eqn: orbits non-rotating mu=0}. Namely, we consider the changes
\begin{equation}\label{eqn: Reparameterization (u,v)}
(\hat{r},\hat{\theta}) = \left(\check{r}_h^\pm(\hat{u},\hat{\Theta}_0), \check{\theta}_h^\pm(\hat{u},\hat v, \hat{\Theta}_0)\right).
\end{equation}
%
We omit the symbols $\pm$ to simplify notation. We define the generating functions
\begin{equation}\label{eqn: Generating function hat(T)}
\hat{T}^{s,u}(\hat{u},\hat{v};\mu,\hat{\Theta}_0) = \hat{S}^{s,u}(\check{r}_h(\hat{u},\hat{\Theta}_0), \check{\theta}_h(\hat{u},\hat v, \hat{\Theta}_0);\mu,\hat{\Theta}_0),
\end{equation}
which can be written as $\hat{T}^{s,u}=\hat{T}^{s,u}_0 + \hat{T}^{s,u}_1$ where
\begin{equation*}
\begin{aligned}
\hat{T}_0(\hat{u},\hat{v};\hat{\Theta}_0) &=\hat{S}_0(\check{r}_h(\hat{u},\hat{\Theta}_0), \check{\theta}_h(\hat{u},\hat v, \hat{\Theta}_0);\hat{\Theta}_0),\\
\hat{T}_1(\hat{u},\hat{v};\mu,\hat{\Theta}_0) &= \hat{S}_1(\check{r}_h(\hat{u},\hat{\Theta}_0), \check{\theta}_h(\hat{u},\hat v, \hat{\Theta}_0);\mu,\hat{\Theta}_0).
\end{aligned}
\end{equation*}
Equation \eqref{eqn: General HJ} now reads
\begin{equation}\label{eqn: HJ T1}
\partial_{\hat{u}}\hat{T}_1 - \partial_{\hat{v}}\hat{T}_1 + \frac{1}{\check R_h}\left(\partial_{\hat{u}}\hat{T}_1 - \frac{\hat{\Theta}_0}{\check{r}_h^2}\partial_{\hat{v}}\hat{T}_1\right)^2 + \frac{1}{2\check{r}_h^2}\partial_{\hat{v}}\hat{T}_1^2 - \hat{V}(\check{r}_h(\hat u,\hat \Theta_0), \check{\theta}_h(\hat{u},\hat v, \hat{\Theta}_0);\mu)=0,
\end{equation}
where $\check r_h,\check \theta_h, \check R_h$ are defined in \eqref{eqn: orbits non-rotating mu=0} and $\hat V$ is defined in \eqref{eqn: potential in polar rotating coordinates centered at CM J}.

%
%
The change of variables \eqref{eqn: Reparameterization (u,v)} implies that we are looking for parameterizations of the stable and unstable invariant manifolds of the form
\begin{equation}
\begin{aligned}
\hat{r} &= \check{r}_h^\pm(\hat{u},\hat{\Theta}_0),\\
\hat{\theta} &= \check{\theta}_h^\pm(\hat{u},\hat v, \hat{\Theta}_0),\\
\hat{R} &= \check{R}_h^\pm(\hat u,\hat \Theta_0) + \frac{1}{\check R_h^\pm(\hat u,\hat\Theta_0)}\left(\partial_{\hat{u}}\hat{T}_1^{s,u} - \frac{\hat{\Theta}_0}{\check{r}_h^\pm(\hat{u},\hat{\Theta}_0)^2}\partial_{\hat{v}}\hat{T}_1^{s,u}\right),\\
\hat{\Theta} &= \hat{\Theta}_0 + \partial_{\hat{v}}\hat{T}_1^{s,u},
\end{aligned}
\end{equation}
where $\hat{T}_1^{s,u}$ are solutions of equation \eqref{eqn: HJ T1} with asymptotic boundary conditions for the unstable manifold
\begin{equation}\label{eqn: asympt}
\begin{aligned}
\underset{\hat{u}\to -\infty}{\lim}  \frac{\partial_{\hat{u}}\hat{T}_1^u(\hat{u},\hat{v};\mu,\hat{\Theta}_0)}{\check R_h^-(\hat u,\hat \Theta_0)} = 0,\quad \underset{\hat{u}\to -\infty}{\lim} \partial_{\hat{v}}\hat{T}_1^u(\hat{u},\hat{v};\mu,\hat{\Theta}_0) = 0,
\end{aligned}
\end{equation}
and analogous ones for the stable manifolds taking $\hat{u} \to +\infty$. The symmetry
\begin{equation}\label{eqn: symmetry potential hat(V)}
\hat{V}(\hat{r},-\hat{\theta}) = \hat{V}(\hat{r},\hat{\theta})
\end{equation}
of $\hat V$ in \eqref{eqn: potential in polar rotating coordinates centered at CM J} implies that if $\hat T_1(\hat u,\hat v)$ solves \eqref{eqn: HJ T1}, then $-\hat T_1(-\hat u,-\hat v)$ also solves it. Imposing the asymptotic condition \eqref{eqn: asympt} for $\hat T_1^u$ (and the reverse for $\hat T_1^s$) yields
\begin{equation}\label{eqn: symmetry T1}
\hat T_1^s(\hat u,\hat v)=-\hat T_1^u(-\hat u,-\hat v).
\end{equation}
Therefore, proving the existence of the unstable manifold implies existence of the stable one.


\begin{proposition}\label{prop: param u v inf}
Fix $m > 0$, $\hat \Theta_0 \in [-\sqrt 2 +m, \sqrt 2 -m]$ and $\kappa > 0$. Then, for $\nu\in \left(0,\frac 1 3\right)$, there exists $\mu_0>0$ such that, for $0 < \mu < \mu_0$, the unstable manifold $W_\mu^u(\Alpha_{\hat{\Theta}_0})$ of the system associated to Hamiltonian $\hat{\mathcal{H}}$ in \eqref{eqn: Hamiltonian Polar Rotating Coordinates centered at CM J} can be written as
\[W_\mu^u(\Alpha_{\hat{\Theta}_0}) = \left\{(\hat{r},\hat{\theta},\hat{R},\hat{\Theta}) = \hat {\mathcal Y}^u(\hat{u},\hat{v},\hat{\Theta}_0)\colon (\hat{u},\hat{v}) \in \hat{D}\times \mathds{T}_{\sigma}\right\}\]
where $\hat D$ and $\mathds{T}_{\sigma}$ are defined as
\begin{equation}\label{eqn: domain parameters (hat u,hat v)}
    \begin{aligned}
        \hat{D} &=  \left\{\hat u \in \mathds{C}\colon \mathrm{Re}(\hat u) \leq - \frac 1 3\sqrt{2(1-\mu+\kappa\mu^\nu) - \hat \Theta_0^2}\cdot\left(1-\mu+\kappa\mu^\nu + \hat\Theta_0^2 \right)\right\},\\
        \mathds{T}_{\sigma} &= \left\{\hat{v}\in \mathds{C}/ (2\pi\mathds{Z})\colon \left|\operatorname{Im}(\hat v)\right|<  \sigma - \frac 1 2 \log\left(1 - \frac{2\hat\Theta_0}{\hat \Theta_0 + \sqrt{2(1-\mu+\kappa\mu^\nu)}}\right)\right\},
    \end{aligned}
\end{equation}
with 
\begin{equation}\label{eqn: sigmatori}
\sigma = \frac 1 4 \log\left(\frac{1-\mu+\kappa\mu^\nu}{1-\mu}\right).
\end{equation}
Here $ \hat {\mathcal Y}^u$ is a real-analytic function of the form
\[ \hat {\mathcal Y}^u(\hat{u},\hat{v},\hat{\Theta}_0) = \left(\hat{r}(\hat{u},\hat{\Theta}_0),\hat{\theta}(\hat{u},\hat{v},\hat{\Theta}_0),\hat{R}(\hat{u},\hat{v},\hat{\Theta}_0), \hat{\Theta}(\hat{u},\hat{v},\hat{\Theta}_0)\right)\]
such that
\begin{equation}\label{eqn: Parameterization of Wu(mu)}
    \begin{aligned}
        \hat{r}(\hat{u},\hat{\Theta}_0) &= \check{r}_h^-(\hat{u},\hat{\Theta}_0),\\
        \hat{\theta}(\hat{u},\hat{v},\hat{\Theta}_0) &= \check\theta_h^-(\hat u,\hat v,\hat\Theta_0),\\
        \hat{R}(\hat{u},\hat{v},\hat{\Theta}_0) &= \check R_h^-(\hat u, \hat\Theta_0) + \mathcal{O}(\mu^{1-2\nu}),\\
        \hat{\Theta}(\hat{u},\hat{v},\hat{\Theta}_0) &= \hat{\Theta}_0 + \mathcal{O}(\mu^{1-2\nu}),
    \end{aligned}
\end{equation}
where $\check{r}_h^-(\hat{u},\hat{\Theta}_0)$, $\check{\theta}_h^-(\hat u,\hat v,\hat{\Theta}_0)$ and $\check R_h^-(\hat u, \hat \Theta_0)$  are defined in \eqref{eqn: orbits non-rotating mu=0}. Moreover
\begin{equation}\label{eqn: Estimates of the derivatives of the parameterization of infty}
    \begin{aligned}
        \partial_{\hat v} \hat R(\hat u, \hat v, \hat{\Theta}_0) =  \partial_{\hat v} \hat \Theta(\hat u, \hat v, \hat{\Theta}_0) = \mathcal{O}(\mu^{1-3\nu}).
    \end{aligned}
\end{equation}
\end{proposition}
To prove this proposition, we rewrite \eqref{eqn: HJ T1} as 
\begin{equation}\label{eqn: L = F}
\mathcal{L}(\hat{T}_1) = \mathcal{F}(\hat{T}_1)
\end{equation}
where
\begin{equation}\label{eqn: Operator F}
\begin{aligned}
\mathcal{L}(f) &= \partial_{\hat{u}} f - \partial_{\hat{v}} f,\\
\mathcal{F}(\hat{T}_1) &= -\frac{1}{\check R_h}\left(\partial_{\hat{u}}\hat{T}_1 - \frac{\hat{\Theta}_0}{\check{r}_h^2}\partial_{\hat{v}}\hat{T}_1\right)^2 - \frac{1}{2\check{r}_h^2}\partial_{\hat{v}}\hat{T}_1^2 + \hat{V}(\check{r}_h(\hat u,\hat \Theta_0), \check{\theta}_h(\hat u, \hat v,\hat\Theta_0);\mu).
\end{aligned}
\end{equation}
The rest of the proof is devoted to finding a solution of \eqref{eqn: L = F} with asymptotic conditions \eqref{eqn: asympt}. To this end, we consider the domain $\hat D \times \mathds{T}_{\sigma_1}$ in \eqref{eqn: domain parameters (hat u,hat v)} for $\sigma_1 = 3\sigma = \frac 3 4\log\left(\frac{1-\mu+\kappa\mu^\nu}{1-\mu}\right)$.
\begin{remark}\label{remark: lower bound u}
    From \eqref{eqn: wt}, \eqref{eqn: orbits non-rotating mu=0} and for $\hat\Theta_0 \in \left[-\sqrt 2 + m,\sqrt 2 -m\right]$ we have that
    \[0 \notin \hat D\quad \text{and}\quad |\check{r}_h^{-1}(\hat u)| \leq C\]
    for $\hat u \in \hat D$ in \eqref{eqn: domain parameters (hat u,hat v)}, where $C > 0$ is an adequate constant (that depends on $m$).
\end{remark}
The next lemma  gives the estimates for the potential $\hat{V}$ in \eqref{eqn: potential in polar rotating coordinates centered at CM J} in the domain $\hat{D} \times \mathds{T}_{\sigma_1}$.
\begin{lemma}\label{lemma: Bound for the potential}
There is $C > 0$ such that for any $(\hat u, \hat v) \in \hat D \times \mathds{T}_{\sigma_1}$, the following bound is satisfied 
\begin{equation}\label{eqn: Bound for V}
    |\hat{V}(\check{r}_h(\hat u,\hat \Theta_0), \check{\theta}_h(\hat u,\hat v,\hat\Theta_0);\mu)| \leq C\mu^{1-\nu}.
\end{equation}
In particular, for $|\hat{u}|\to \infty$,
\begin{equation}\label{eqn: Bound V u to inf}
    |\hat{V}(\check{r}_h(\hat u,\hat \Theta_0),  \check{\theta}_h(\hat u,\hat v,\hat\Theta_0);\mu)|\leq C \frac{\mu}{|\hat{u}|^{\frac 4 3}}.
\end{equation}
\end{lemma}
\begin{proof}[Proof of Lemma \ref{lemma: Bound for the potential}]
We denote by $\hat r := \check r_h(\hat u, \hat \Theta_0)$ and $\hat \theta := \hat v + \check \theta_h(\hat u,\hat \Theta_0)$. The estimate of the potential $\hat V$ in \eqref{eqn: potential in polar rotating coordinates centered at CM J} can be split as follows
\[\left|\hat V(\hat r,\hat \theta;\mu)\right| \leq \left|\hat V_{\mathcal S}(\hat r, \hat \theta;\mu)- \hat V_{\mathrm{CM}}(\hat r)\right|+\left|\hat V_{\mathcal J}(\hat r, \hat \theta;\mu)\right|\]
where
\begin{equation*}
    \begin{aligned}
        \hat V_{\mathcal S}(\hat r, \hat \theta;\mu) =& \frac{1-\mu}{\left(\hat r^2 + 2\hat r\mu\cos\hat \theta + \mu^2\right)^{\frac 1 2}},\quad \hat V_{\mathrm{CM}}(\hat r) = \frac{1}{\hat r},\\
        \hat V_{\mathcal J}(\hat r,\hat \theta;\mu) =& \frac{\mu}{\left(\hat r^2 - 2\hat r(1-\mu)\cos\hat \theta + (1-\mu)^2\right)^{\frac 12}}.   
    \end{aligned}
\end{equation*}
The domain $\hat D$ in \eqref{eqn: domain parameters (hat u,hat v)} is considered so that $|\hat r| \geq 1-\mu + \kappa\mu^\nu$, meaning that we are far from both the primary $\mathcal S$ and the center of mass. Moreover, for $\mu = 0$ we have that $\hat V_{\mathcal S}(\hat r,\hat \theta;0) - \hat V_{\mathrm{CM}}(\hat r) = 0$. Therefore we have 
\[\left|\hat V_{\mathcal S}(\hat r, \hat \theta;\mu)- \hat V_{\mathrm{CM}}(\hat r)\right|  \lesssim \mu.\]
We estimate $ \left|\hat V_\mathcal J(\hat r, \hat \theta;\mu)\right|$ by looking for a lower bound for $|\hat r^2 - 2\hat r(1-\mu)\cos\hat \theta + (1-\mu)^2|$ on $\hat D \times \mathds T_{\sigma_1}$. We denote by $\tilde r := (1-\mu)^{-1}\hat r$ so that
\[\left|\hat r^2 - 2\hat r(1-\mu)\cos\hat \theta + (1-\mu)^2\right| = (1-\mu)\left|\tilde r - e^{i\hat \theta}\right|\cdot\left|\tilde r - e^{-i\hat\theta}\right|.\]
Since
\[|\tilde r| \geq 1 + \frac{\kappa\mu^\nu}{1-\mu},\quad |e^{\pm i\hat\theta}| \leq \exp\left(\frac{3}{4}\log\left(1+\frac{\kappa\mu^\nu}{1-\mu}\right)\right) = 1+\frac{3}{4}\frac{\kappa\mu^\nu}{1-\mu} + \mathcal{O}(\mu^{2\nu}),\quad \forall (\hat u,\hat v)\in \hat D\times \mathds{T}_{\sigma_1},\]
we obtain
\[\left|\tilde r - e^{\pm i\hat \theta}\right| \geq |\tilde r| - |e^{\pm i\hat\theta}| \geq 1 + \frac{\kappa\mu^\nu}{1-\mu} - \left(1+\frac 3 4\frac{\kappa\mu^\nu}{1-\mu} + \mathcal{O}(\mu^{2\nu})\right) \geq \frac{\kappa\mu^\nu}{4(1-\mu)} + \mathcal{O}(\mu^{2\nu}),\]
which leads to \eqref{eqn: Bound for V}.

Finally, the estimate \eqref{eqn: Bound V u to inf} is a direct  consequence of Corollary \ref{corollary: Behaviour at infinity of the unperturbed separatrix} and the expansion of $\hat V$ in $\hat r^{-1}$ for large $\hat r$.
\end{proof}
As a consequence of this Lemma, the operator $\mathcal F$ in \eqref{eqn: Operator F} is well-defined. The next step is to set up a fixed point argument to solve equation \eqref{eqn: L = F} in a suitable Banach space for functions defined on $\hat{D} \times \mathds{T}_{\sigma_1}$. For a big enough constant $K > 0$, we define the following norm for $f\colon \hat{D}\to \mathds{C}$,
\[\|f\|_a = \underset{\substack{\hat{u}\in \hat{D}\\|\hat u| > K}}{\sup} |\hat{u}^a f(\hat{u})| + \underset{\substack{\hat u \in \hat D\\ |\hat{u}|\leq K}}{\sup}|f(\hat u)|.\]
Then, for $\sigma > 0$ and functions $f\colon \hat{D} \times \mathds{T}_\sigma \to \mathds{C}$, we define
\[\|f\|_{a,\sigma} = \underset{k\in \mathds{Z}}{\sum} \|f^{[k]}\|_a e^{|k|\sigma},\]
and the function space
\[\mathcal{Z}_{a,\sigma} = \left\{ f\colon \hat{D}\times \mathds{T}_\sigma \to \mathds{C},\text{ real-analytic}, \|f\|_{a,\sigma} < \infty\right\}.\]
It can be checked that it is a Banach space for any fixed $a\geq 0$. Moreover, since equation \eqref{eqn: L = F} involves the derivatives $\partial_{\hat{u}}\hat{T}_1$ and $\partial_{\hat{v}}\hat{T}_1$, we need a Banach space that controls at the same time the norms of a function and its derivatives. To this end we define the norm
\[\llfloor f\rrfloor_{a,\sigma}= \|f\|_{a,\sigma} + \|\partial_{\hat{u}}f\|_{a+1,\sigma} + \|\partial_{\hat{v}}f\|_{a+1,\sigma}\]
and the corresponding Banach space
\begin{equation}\label{eqn: Space TildeZ}
\Tilde{\mathcal{Z}}_{a,\sigma} = \left\{f\colon \hat{D} \times \mathds{T}_\sigma \to \mathds{C}, \text{ real-analytic }, \llfloor f\rrfloor_{a,\sigma}< \infty\right\}.
\end{equation}
Next lemma, whose proof is analogous to the one of Lemma 5.6 of \cite{MR3455155}, gives properties about the following inverse of $\mathcal L$
\begin{equation}\label{eqn: Operator G}
\mathcal{G}(f)(\hat{u},\hat{v}) = \int_{-\infty}^0 f(\hat{u}+s,\hat{v}-s)\;ds
\end{equation}
acting on the Banach space $\mathcal{Z}_{a,\sigma}$.
\begin{lemma}\label{lemma: Bounds of G}
    The operator $\mathcal{G}$ defined in \eqref{eqn: Operator G}, when considered acting on the space $\mathcal{Z}_{a,\sigma}$, satisfies the following properties.
    \begin{itemize}
        \item For any $a > 1$, $\mathcal{G}\colon \mathcal{Z}_{a,\sigma}\to \mathcal{Z}_{a-1,\sigma}$ is well-defined and linear continuous. Moreover, $\mathcal{L} \circ \mathcal{G} = \mathrm{Id}$.
        \item If $f \in \mathcal{Z}_{a,\sigma}$ for some $a > 1$, then
        \[\|\mathcal{G}(f)\|_{a-1,\sigma}  \lesssim  K \|f\|_{a,\sigma}.\]
        \item If $f\in \mathcal{Z}_{a,\sigma}$ for some $a \geq 1$, then $\partial_{\hat{u}}\mathcal{G}(f), \partial_{\hat{v}}\mathcal{G}(f) \in \mathcal{Z}_{a,\sigma}$ and
        \[\|\partial_{\hat{u}}\mathcal{G}(f)\|_{a,\sigma}  \lesssim \|f\|_{a,\sigma}, \quad \|\partial_{\hat{v}}\mathcal{G}(f)\|_{a,\sigma} \lesssim  \|f\|_{a,\sigma}.\]
        \item From the previous statements, one can conclude that if $f\in \mathcal{Z}_{a,\sigma}$ for some $a > 1$, then $\mathcal{G}(f)\in \Tilde{\mathcal{Z}}_{a-1,\sigma}$ and
        \[\llfloor\mathcal{G}(f)\rrfloor_{a-1,\sigma} \lesssim   \|f\|_{a,\sigma}.\]
    \end{itemize}
\end{lemma}
We look for a fixed point in the space $\Tilde{\mathcal{Z}}_{1/3,\sigma}$  of the operator
\begin{equation}\label{eqn: Operator Ftilde}
\Tilde{\mathcal{F}} = \mathcal{G}\circ \mathcal{F}
\end{equation}
where $\mathcal{F}$ and $\mathcal{G}$ are the operators defined in \eqref{eqn: Operator F} and \eqref{eqn: Operator G}, respectively. Proposition \ref{prop: param u v inf} is a straightforward consequence of the following proposition.
\begin{proposition}\label{prop: Fixed point}
    Fix $\kappa > 0$ and $\sigma_0= \frac{2}{3}\sigma_1 = 2\sigma = \frac 1 2 \log\left(\frac{1-\mu+\kappa\mu^\nu}{1-\mu}\right)$ in \eqref{eqn: sigmatori}. There exists a constant $b_0 > 0$ such that, for $\mu > 0$ small enough, the operator $\Tilde{\mathcal{F}}$ in \eqref{eqn: Operator Ftilde} has a fixed point $\hat{T}_1 \in B(b_0\mu^{1-2\nu}) \subset \Tilde{\mathcal{Z}}_{1/3,\sigma_0}$. 
\end{proposition}
\begin{proof}[Proof of Proposition \ref{prop: Fixed point}]
We bound $\Tilde{\mathcal{F}}(0) = \mathcal{G}\circ \mathcal{F}(0)$. By the definition of $\mathcal{F}$ in \eqref{eqn: Operator F} we have that
\begin{equation}\label{eqn: F(0)}
    \mathcal{F}(0) = \hat{V}(\check{r}_h(\hat u,\hat \Theta_0), \check{\theta}_h(\hat u,\hat v,\hat\Theta_0);\mu),
\end{equation}
with $\hat V$ is defined in \eqref{eqn: potential in polar rotating coordinates centered at CM J}. It satisfies
\begin{equation}\label{eqn: Bound Fourier norm V}
\|\hat V\|_{4/3,\sigma_0} = \underset{k\in\mathds{Z}}{\sum} \|\hat V^{[k]}\|_{4/3}\cdot e^{|k|\sigma_0}\leq  \|\hat V\|_{4/3,\sigma_1} \underset{k\in\mathds{Z}}{\sum}e^{-|k|(\sigma_1-\sigma_0)} \lesssim \mu^{1-2\nu},
\end{equation}
where we have used that $\hat{V}$ is an analytic function on $\hat D\times \mathds{T}_{\sigma_1}$ in \eqref{eqn: domain parameters (hat u,hat v)} so that $\|\hat V^{[k]}\|_{4/3} \leq \|\hat V\|_{4/3,\sigma_1}\cdot e^{-|k|\sigma_1}$. Therefore  $\mathcal{F}(0) \in \mathcal{Z}_{4/3,\sigma_0}$ and $\|\mathcal{F}(0)\|_{4/3,\sigma_0} \leq K\mu^{1-2\nu}$ for some adequate constant $K > 0$. Then, applying the last statement of Lemma \ref{lemma: Bounds of G}, there exists a constant $b_0 > 0$ such that
\[\llfloor\mathcal{F}(0)\rrfloor_{1/3,\sigma_0}\leq \frac{b_0}{2}\mu^{1-2\nu}.\]
We show that the operator $\Tilde{\mathcal{F}}$ is contractive in the ball $B(b_0\mu^{1-2\nu})\subset \Tilde{\mathcal{Z}}_{1/3,\sigma_0}$. Let $f_1,f_2\in B(b_0\mu^{1-2\nu})$. Using the last statement of Lemma \ref{lemma: Bounds of G} we have that
\[\llfloor\Tilde{\mathcal{F}}(f_2) - \Tilde{\mathcal{F}}(f_1)\rrfloor_{1/3,\sigma_0} \lesssim  \|\mathcal{F}(f_2)-\mathcal{F}(f_1)\|_{4/3,\sigma_0}.\]
We bound the right hand side of this formula. We write it as
\begin{equation*}
    \begin{aligned}
        \mathcal{F}(f_2)-\mathcal{F}(f_1) 
        =& - \frac{1}{\check R_h}\left(\partial_{\hat u}f_2+\partial_{\hat u}f_1 - \frac{2\hat{\Theta}_0}{\check r_h^2}\partial_{\hat v}f_2\right)\left(\partial_{\hat u}f_2 - \partial_{\hat u}f_1\right) \\
        &+\left(\frac{2\hat{\Theta}_0}{\check R_h \check r_h^2}\partial_{\hat u}f_1 - \left(\frac{\hat{\Theta}_0^2}{\check R_h\check r_h^4} + \frac{1}{2\check r_h^2}\right)(\partial_{\hat v}f_2+\partial_{\hat v}f_1)\right)\left(\partial_{\hat v}f_2 - \partial_{\hat v}f_1\right).
    \end{aligned}
\end{equation*}
Using Remark \ref{remark: lower bound u}, the behavior of $\check{r}_h$ from Corollary \ref{corollary: Behaviour at infinity of the unperturbed separatrix} and the fact that $f_1,f_2\in B(b_0\mu^{1-2\nu})$ we obtain
\[\|\mathcal{F}(f_2)-\mathcal{F}(f_1)\|_{4/3,\sigma_0}  \lesssim  \mu^{1-2\nu}\llfloor f_2-f_1\rrfloor_{1/3,\sigma_0}.\]
Therefore
\[\llfloor\Tilde{\mathcal{F}}(f_2)-\Tilde{\mathcal{F}}(f_1)\rrfloor_{1/3,\sigma_0}  \lesssim  \|\mathcal{F}(f_2)-\mathcal{F}(f_1)\|_{4/3,\sigma_0}  \lesssim \mu^{1-2\nu}\llfloor f_2-f_1\rrfloor_{1/3,\sigma_0},\]
and thus, the operator $\mathcal{\tilde F}$ is well-defined and it is contractive. Hence, it has a unique fixed point $\hat T_1 \in B(b_0\mu^{1-2\nu})\subset \Tilde{\mathcal{Z}}_{1/3,\sigma_0}$. This leads to the parameterization \eqref{eqn: Parameterization of Wu(mu)} of $W_\mu^u(\Alpha_{\hat \Theta_0})$. The estimates \eqref{eqn: Estimates of the derivatives of the parameterization of infty} are obtained from the Cauchy's estimates on $\mathds{T}_\sigma$ defined in \eqref{eqn: domain parameters (hat u,hat v)}, whose minimum distance to $\partial \mathds{T}_{\sigma}$ is $\frac 1 4 \log \left(\frac{1-\mu+\kappa\mu^\nu}{1-\mu}\right)> \frac {1}{8\kappa} \mu^{\nu}$. 
\end{proof}

\section{Proof of Proposition \ref{prop: parameterization invariant manifolds with section close to S}}\label{appendix: parameterization of invariant manifolds with section close to S}
We provide the proof for the ejection orbits. One can deduce an analogous result for the collision ones using that the system is reversible.

Fix $\gamma \in \left(\frac{3}{11},\frac 1 3\right)$. The proof relies on extending of the ejection curve $ \Lambda_{\mathcal J}^-(\mu)$ in \eqref{eqn: Curves ECO-section} (which belongs to the section $ \Sigma_\gamma$ in \eqref{eqn: Section Sigma (r,theta,R,Theta)}) to the section $\overline \Sigma$ in \eqref{eqn: SigmahSun}. The argument is carried out using the polar coordinates centered at the center of mass $(\hat r, \hat \theta, \hat R, \hat \Theta)$, and is structured as follows:
\begin{enumerate}
    \item We express the curve $ \Lambda_{\mathcal J}^-(\mu) \subset\Sigma_\gamma$ in coordinates $(\hat r, \hat \theta,\hat R, \hat \Theta)$. We denote this curve as $\hat \Lambda_{\mathcal J}^-(\mu)$.
    \item We express the target section $\overline \Sigma$ in coordinates $(\hat r, \hat \theta, \hat R,\hat \Theta)$, and denote it as $\hat \Sigma$.
    \item We extend the curve $\hat \Lambda_{\mathcal J}^-(\mu)$ by the flow associated to the Hamiltonian \eqref{eqn: Hamiltonian Polar Rotating Coordinates centered at CM J} to reach a curve in the section $\hat \Sigma$. This is done in two stages: first we consider the flow for the unperturbed Hamiltonian ($\mu = 0$) and then we study the perturbed flow for $\mu > 0$ small enough.
    \item We express the image curve in polar coordinates $(\rS,\thetaS,\RS,\ThetaS)$ (see \eqref{eqn: Change from synodical cartesian to synodical polar centered at P1 J}) and we parameterize it as graph in terms of $\thetaS$, yielding $\eqref{eqn: Curve intersection Ejection-overline r=delta2}$ and completing the proof.
\end{enumerate}
\paragraph{Step 1. Coordinate transformation of the curve $ \Lambda_{\mathcal J}^-(\mu)$:}
In coordinates $(\hat r, \hat \theta, \hat R,\hat \Theta)$, the ejection curve $ \Lambda_{\mathcal J}^{-}(\mu)$ becomes
 \begin{equation}\label{eqn: Ejection Curve in coordinates CM}
         \hat \Lambda_{\mathcal J}^{-} (\mu) =\left\{\left(\hat r_{\mathcal J}^-( \theta;\mu), \hat \theta_{\mathcal J}^-(\theta;\mu) , \hat R_{\mathcal{J}}^-(\theta;\mu), \hat \Theta_{\mathcal J}^-(\theta;\mu)\right)\colon \theta \in \mathds{T}\right\},
 \end{equation}
where
\begin{equation}\label{eqn: Parameterization of the ejection curve in coordinates centered at CM}
    \begin{aligned}
        \hat r_{\mathcal J}^-(\theta;\mu) =& \sqrt{\mu^{2\gamma}+ 2\mu^\gamma\cos\theta(1-\mu) + (1-\mu)^2},\\
        \hat \theta_{\mathcal J}^-(\theta;\mu) =& \arctan\left(\frac{\mu^\gamma\sin\theta}{\mu^\gamma\cos\theta + (1-\mu)}\right),\\
        \hat R_{\mathcal J}^-(\theta;\mu) =&\;  R_{\mathcal J}^-(\theta;\xi,\mu)\cos(\theta-\hat\theta_{\mathcal J}^-(\theta;\mu)) - \frac{\Theta_{\mathcal J}^-(\theta;\xi,\mu)}{\mu^\gamma}\sin(\theta-\hat\theta_{\mathcal J}^-(\theta;\mu)) + \sin\hat\theta_{\mathcal J}^-(\theta;\mu)(1-\mu),\\
        \hat \Theta_{\mathcal J}^-(\theta;\mu) =&\; \hat r_{\mathcal J}^-(\theta;\mu)  R_{\mathcal J}^-(\theta;\xi,\mu)\sin(\theta-\hat\theta_{\mathcal J}^-(\theta;\mu)) + \frac{\hat r_{\mathcal J}^-(\theta;\mu)}{\mu^\gamma}\Theta_{\mathcal J}^-(\theta;\xi,\mu)\cos(\theta-\hat \theta_{\mathcal J}^-(\theta;\mu)) \\
        &+ \hat r_{\mathcal J}^-(\theta;\mu)\cos\hat \theta_{\mathcal J}^-(\theta;\mu)(1-\mu).
    \end{aligned}
\end{equation}
Since $h = \mathcal{O}(\mu)$, we obtain the following estimate for the value of $\xi$ in \eqref{eqn: Energy xi} 
\[\xi = (2h+3)^{-\frac 1 2} + \mathcal{O}(\mu) = \frac{1}{\sqrt 3} + \mathcal{O}(\mu).\]
We substitute this estimate in both $( R_{\mathcal J}^-, \Theta_{\mathcal J}^-)$ defined in \eqref{eqn: (R,Theta) in curves ECO-section} to obtain
\begin{equation}\label{eqn: Estimates on the parameterization of the ejection curve in CM}
    \begin{aligned}
        \hat r_{\mathcal J}^-(\theta, \upsilon) =&\; 1+ \upsilon \cos\theta + \mathcal{O}_2(\upsilon),\qquad \hat \theta_{\mathcal J}^-(\theta,\upsilon) =\;  \upsilon\sin\theta + \mathcal{O}_2(\upsilon),\\
        \hat R_{\mathcal J}^-(\theta,\upsilon,\sigma) =&\; \sqrt 3\cos\theta+ \mathcal{O}(\upsilon,\sigma),\qquad \hat \Theta_{\mathcal J}^-(\theta,\upsilon,\sigma) =\; 1+\sqrt{3}\sin\theta +  \mathcal{O}(\upsilon,\sigma),
    \end{aligned}
\end{equation}
where
\begin{equation}\label{eqn: notation upsilon, sigma}
    \sigma := \sigma(\mu) = \mu^{\frac{11\gamma-3}{8}},\quad \upsilon := \upsilon(\mu) = \mu^\gamma
\end{equation}
are defined to simplify notation. Note that $\sigma \gg \upsilon$ since $\gamma \in \left(\frac{3}{11},\frac{1}{3}\right) < 1$.

\paragraph{Step 2. Coordinate transformation of the section $\overline \Sigma$:}
We consider the transformation from coordinates $(\rS,\thetaS,\RS,\ThetaS)$ to coordinates $(\hat r, \hat \theta,\hat R,\hat \Theta)$
\begin{equation}\label{eqn: Change S CM polar}
    \begin{aligned}
         \rS(\hat r, \hat \theta,\mu)&= \sqrt{\hat r^2 + 2\hat r \mu\cos\hat\theta + \mu^2},\qquad \thetaS (\hat r, \hat \theta,\mu) = \arctan\left(\frac{\hat r\sin\hat \theta}{\hat r\cos\hat \theta + \mu}\right),\\
         \RS(\hat r, \hat \theta,\hat R,\hat \Theta,\mu) &= \hat R \frac{ \rS(\hat r,\hat \theta)-\mu\cos\thetaS(\hat r,\hat \theta)}{\hat r} - \mu \hat\Theta\frac{ \rS(\hat r,\hat \theta)\sin\thetaS(\hat r, \hat \theta)}{ \hat r},\\
        \ThetaS(\hat r, \hat \theta,\hat R,\hat \Theta,\mu) &= \mu\hat R\frac{ \rS(\hat r,\hat \theta)\sin\thetaS(\hat r, \hat \theta)}{\hat r}+\hat\Theta\frac{ \rS(\hat r,\hat \theta)\left(\rS(\hat r,\hat \theta)-\mu\cos\thetaS(\hat r,\hat \theta)\right)}{\hat r^2}.
    \end{aligned}
\end{equation}
This transformation satisfies
\begin{equation}\label{eqn: Estimation polar coordinates S-CM}
    \begin{aligned}
        \rS(\hat r, \hat \theta,\mu)  &= \hat r + \mu\cos\hat \theta + \mathcal{O}\left(\frac{\mu^2}{\hat r}\right),\quad \thetaS(\hat r, \hat \theta,\mu)  = \hat \theta - \mu \frac{\sin \hat \theta}{\hat r} + \mathcal{O}_2 \left(\frac{\mu}{\hat r}\right),\\
        \RS(\hat r, \hat \theta,\hat R,\hat \Theta,\mu) &= \hat R -\mu\hat\Theta\sin\hat\theta + \mathcal{O}\left(\frac{\mu^2}{\hat r^2}\right),\\
        \ThetaS(\hat r, \hat \theta,\hat R,\hat \Theta,\mu) &= \hat \Theta\left(1 + \frac{\mu\cos\hat \theta}{\hat r}\right) + \mu \hat R \sin\hat \theta + \mathcal{O}_2(\mu).
    \end{aligned}
\end{equation}
We define the function $\mathcal F(\hat r, \hat \theta,\mu) =\rS(\hat r, \hat \theta,\mu) - \delta^2$ and look for its zeroes. It satisfies
\[\mathcal F(\delta^2,\hat\theta,0) = 0,\quad \partial_{\hat r}\mathcal F(\delta^2,\hat\theta,0) = 1 \neq 0.\]
Then, the Implicit Function Theorem ensures that there exists $\mu_0 > 0$ such that, for any $0<\mu<\mu_0$ and $\hat \theta \in \mathds{T}$, there is $\hat r = \hat P(\hat \theta,\mu)$ with $\hat P(\hat \theta,0) = \delta^2$ such that $\rS(\hat P(\hat \theta,\mu),\hat\theta;\mu) = \delta^2$. Moreover, we have
\begin{equation}\label{eqn:hat r=delta2}
    \hat P(\hat\theta,\mu) = \delta^2 + \mathcal{O}(\mu).
\end{equation}
Hence, in coordinates $(\hat r, \hat \theta, \hat R,\hat \Theta)$, the section $\overline \Sigma$ in \eqref{eqn: SigmahSun} becomes
\begin{equation}\label{eqn: Section hatr = delta2}
    \hat \Sigma = \left\{(\hat r,\hat \theta,\hat R,\hat \Theta)\colon \hat r = \hat P(\hat \theta,\mu), \mathcal{\hat H}_\mu(\hat P(\hat \theta,\mu), \hat \theta, \hat R,\hat \Theta) = h\right\}.
\end{equation}
where $\hat{\mathcal H}_\mu$ is the Hamiltonian \eqref{eqn: Hamiltonian Polar Rotating Coordinates centered at CM J}.

\paragraph{Step 3. Extension by the flow:}
We extend the ejection curve $\hat \Lambda_{\mathcal J}^-(\mu)$ in \eqref{eqn: Parameterization of the ejection curve in coordinates centered at CM} by the flow associated to the Hamiltonian \eqref{eqn: Hamiltonian Polar Rotating Coordinates centered at CM J} to the section $\hat \Sigma$ defined in \eqref{eqn: Section hatr = delta2}. To carry out this extension, denote by $\hat F$ the vector field associated to the Hamiltonian \eqref{eqn: Hamiltonian Polar Rotating Coordinates centered at CM J}, which can be written as
\[\hat F(\hat r, \hat \theta, \hat R,\hat \Theta;\mu) = \hat F_0(\hat r, \hat R, \hat \Theta) + \hat F_1(\hat r, \hat \theta;\mu)\]
where
\begin{equation}\label{eqn: hatF0 and hatF1}
    \begin{aligned}
        \hat F_0(\hat r, \hat R,\hat \Theta) = \left(\hat R, \frac{\hat \Theta}{\hat r^2}-1,\frac{\hat \Theta^2}{\hat r^3}- \frac{1}{\hat r^2},0\right)^T,\qquad \hat F_1(\hat r, \hat \theta;\mu) = \left(0,0,\partial_{\hat r}\hat V(\hat r, \hat \theta,\mu), \partial_{\hat \theta}\hat V(\hat r, \hat \theta,\mu)\right)^T,
    \end{aligned}
\end{equation}
and $\hat V(\hat r, \hat \theta,\mu)$ is the potential \eqref{eqn: potential in polar rotating coordinates centered at CM J}, which satisfies $\hat V(\hat r,\hat \theta,0) = 0$.

The vector field $\hat{F}$ has multiple singularities located at $\hat{r}=0$, $(\hat{r},\hat{\theta}) = (\mu,\pi)$ and $(\hat{r},\hat{\theta}) = (1-\mu,0)$ which correspond to the positions of the center of mass and the primaries $\mathcal S$ and $\mathcal J$ respectively. However, the extension we consider avoids these singularities, and thus the vector field $\hat F$ remains regular throughout the associated flow.

We first perform the extension for $\hat F_0$. To this end we recall Remark \ref{remark: parabolic EC orbits with J} (the case $\check \Theta_0 = 0$) and consider, in non-rotating coordinates $(\check r, \check \theta, \check R,\check \Theta)$, the ballistic trajectory $\check \gamma_c^-(t,0) \in \mathcal S^+ \cap \mathcal J^-$ in \eqref{eqn: parabolic ejection and collision mu=0}, defined for $t\in[-t_c,0)$ where $t_c = \frac{\sqrt 2}{3}$ is given in \eqref{eqn: tc thetac hatTheta} (see Remark \ref{def: large-ballistic} for the definition of ballistic trajectory). Since for $\mu = 0$ the Hamiltonian $\check{\mathcal H}_0$ in \eqref{eqn:Hamiltonian rotating polar coordinates for mu = 0 J} is autonomous, without loss of generality we consider the initial condition $\check \gamma_c^-(0,0) \in \mathcal J$ so that $\check \gamma_c^-(t,0) \underset{t\to t_c^-}{ \longrightarrow} \mathcal S$.

We denote by $\hat \Phi_0(t,x_0) := \left(\hat \Phi_0^{\hat r}, \hat \Phi_0^{\hat \theta}, \hat \Phi_0^{\hat R}, \hat \Phi_0^{\hat \Theta}\right)$ the flow associated to the Hamiltonian $\hat{\mathcal{H}}_0$ with initial condition $x_0 = (\hat r_0, \hat \theta_0, \hat R_0, 0)$ (where $\hat R_0 = \hat R_0(\hat r_0,\hat \Theta_0)$ is obtained through the conservation of the Hamiltonian $\hat{\mathcal H}_0$). Then we have
\begin{equation}\label{eqn: Flow 2BP Theta= 0}
    \hat\Phi_0^{\hat r}(t,x_0) = \left(\hat r_0^{\frac 3 2}-\frac{3t}{\sqrt 2}\right)^{\frac 2 3},\quad \hat\Phi_0^{\hat \theta}(t,x_0) = \hat \theta_0-t,\quad \hat\Phi_0^{\hat R}(t,x_0) = -\sqrt{\frac{2}{\hat\Phi_0^{\hat r}(t,x_0)}}, \quad \hat\Phi_0^{\hat \Theta}(t,x_0) = 0.
\end{equation}
We compute the time $t_0$ such that $\hat \Phi_0^{\hat r}(t_0,x_0)\in \hat \Sigma$ in \eqref{eqn: Section hatr = delta2}, yielding
\begin{equation}\label{eqn: Time Flow 2BP Theta = 0}
t_0 = t_0(\hat r_0) = \frac{\sqrt 2}{3}\left(\hat r_0^{\frac 3 2}-\delta^3\right).
\end{equation}
Building on these results, the following lemma provides estimates for both the flow $\hat \Phi_0$ and the time required for initial conditions in $\hat \Lambda_{\mathcal J}^-(\mu)$, defined in \eqref{eqn: Ejection Curve in coordinates CM}, to reach the section $\hat \Sigma$ in \eqref{eqn: Section hatr = delta2}.
\begin{lemma}[Extension by the flow $\hat \Phi_0$]\label{lemma: Extension phi0}
Let $\delta_0 > 0$ be the parameter given in Proposition \ref{proposition: Perturbed invariant manifolds of collision in synodical polar coordinates J} and fix $\gamma \in \left(\frac{3}{11},\frac 1 3\right)$. For  $0<\delta<\delta_0$, consider the time interval
\begin{equation}\label{eqn: interval t Phi0}
\mathcal T= \left[0,\frac{\sqrt 2}{3}\left(1-\frac{\delta^3}
{2}\right)\right].
\end{equation}
Then, there exist $\mu_0,\omega_0 > 0$ such that, for $(\mu,\omega) \in (0,\mu_0)\times (0,\omega_0)$ and any initial condition $\hat x \in \hat \Gamma_{\mathcal J}^-(\upsilon,\sigma)$ (with $\upsilon,\sigma$ given in \eqref{eqn: notation upsilon, sigma}) defined as
\begin{equation}\label{eqn: reduced curve init cond}
    \hat \Gamma_{\mathcal J}^-(\upsilon,\sigma) = \left\{\left(\hat r_{\mathcal J}^-(\theta;\upsilon,\sigma), \hat \theta_{\mathcal J}^-(\theta;\upsilon,\sigma) , \hat R_{\mathcal{J}}^-(\theta;\upsilon,\sigma), \hat \Theta_{\mathcal J}^-(\theta;\upsilon,\sigma)\right)\colon \theta \in \left(\theta_*-\omega,\theta_*+\omega\right)\right\}\subset \hat \Lambda_{\mathcal J}^-(\mu),
\end{equation}
where $\theta_* = \arcsin\left(\frac{1}{\sqrt 3}\right) + \pi$ , the flow $\hat \Phi_0$ satisfies the following estimates
\begin{equation}\label{eqn: Estimate flow Phi_0(t,x)}
\begin{aligned}
    \hat \Phi_0(t,\hat x) = \hat \Phi_0 (t,\hat x_0(\theta_*)) + \mathcal{O}(\upsilon,\sigma,\omega),\quad \partial_{\hat x}\hat \Phi_0(t,\hat x) =   \partial_{\hat x}\hat \Phi_0 (t,\hat x_0(\theta_*)) + \mathcal{O}(\upsilon,\sigma,\omega),
\end{aligned}
\end{equation}
for $t\in \mathcal T$.

Moreover, the time needed for the flow $\hat \Phi_0$ starting at $\hat x \in \hat \Gamma_{\mathcal J}^-(\upsilon,\sigma)$ to reach the section $\hat \Sigma$ in \eqref{eqn: Section hatr = delta2} satisfies
\begin{equation}\label{eqn: Estimate time for Phi_0(t,x)}
\begin{aligned}
      t(\theta,\upsilon,\sigma)= \frac{\sqrt 2}{3}\left(1-\delta^3\right) + \mathcal{O}(\upsilon,\sigma,\omega) \in \mathcal T,\qquad 
      \partial_{\theta}t(\theta,\upsilon,\sigma) = \frac{3}{5\sqrt 2}+ \mathcal{O}(\upsilon,\sigma,\omega).
\end{aligned}
\end{equation}
\end{lemma}
\begin{proof}[Proof of Lemma \ref{lemma: Extension phi0}]
We compute the time needed for an arbitrary point $\hat x_0\in\hat \Lambda_{\mathcal J}^{-}(0)$ defined in \eqref{eqn: Ejection Curve in coordinates CM} to reach the section $\hat \Sigma$ in \eqref{eqn: Section hatr = delta2} under the flow $\hat \Phi_0$ in \eqref{eqn: Flow 2BP Theta= 0}. Relying on \eqref{eqn: Estimates on the parameterization of the ejection curve in CM}, such point $\hat x_0$ has the expression
\begin{equation*}
    \begin{aligned}
        \hat x_0(\theta_0) &= \left(\hat r_{\mathcal J}^-(\theta_0,0), \hat \theta_{\mathcal J}^-(\theta_0,0),  \hat R_{\mathcal J}^-(\theta_0,0,0), \hat\Theta_{\mathcal J}^-(\theta_0,0,0)\right)\\
        &= \left(1,0,\sqrt 3\cos\theta_0,1+\sqrt 3\sin\theta_0\right).
    \end{aligned}
\end{equation*}
To apply the expression in \eqref{eqn: Flow 2BP Theta= 0}, the $\hat \Theta$-component of $\hat x_0(\theta_0)$ must be zero. That is, we look for $\theta_0 \in \mathds{T}$ satisfying $\hat R_{\mathcal J}^-(\theta_0,0,0) < 0$ and $\hat \Theta_{\mathcal J}^-(\theta_0,0,0)= 0$ in \eqref{eqn: Estimates on the parameterization of the ejection curve in CM}, or equivalently 
\[\cos\theta_0 <0,\quad 1+\sqrt 3\sin\theta_0 = 0,\]
which yields
\begin{equation}\label{eqn: theta_*}
\theta_* = \arcsin\left(\frac{1}{\sqrt 3}\right) + \pi
\end{equation}
so that 
\begin{equation}\label{eqn: Init condition ejection curve Theta=0}
\hat x_0(\theta_*) = \left(1,0, \sqrt 3\cos\theta_*,0\right) = (1,0,-\sqrt 2,0)\in \mathcal J.
\end{equation}
In this case, the time needed to reach the section $\hat \Sigma$ in \eqref{eqn: Section hatr = delta2} is computed from \eqref{eqn: Time Flow 2BP Theta = 0} as 
\begin{equation}\label{eqn: time t01}
    t_0^* = \frac{\sqrt 2}{3}\left(1-\delta^3\right).
\end{equation}
We extend the previous analysis to points in $\hat \Gamma_{\mathcal{J}}^{-}(\upsilon,\sigma)$, defined in \eqref{eqn: reduced curve init cond}. 
Recall that $\hat \Gamma_{\mathcal{J}}^{-}(0,0)$ corresponds to the point $x_0(\theta_*)$ defined in \eqref{eqn: Init condition ejection curve Theta=0}. For any initial condition 
\[\hat x(\theta,\upsilon,\sigma)=\left(\hat r_{\mathcal J}^-(\theta,\upsilon), \hat \theta_{\mathcal J}^-(\theta,\upsilon) , \hat R_{\mathcal{J}}^-(\theta,\upsilon,\sigma), \hat \Theta_{\mathcal J}^-(\theta,\upsilon,\sigma)\right)\in \hat \Gamma_{\mathcal J}^{-}(\upsilon,\sigma)\]
we consider the flow $\hat \Phi_0(t,\hat x(\theta,\upsilon,\sigma))$ for $t \in \mathcal T$ (defined in \eqref{eqn: interval t Phi0}) and we expand it with respect to the initial condition $\hat x(\theta_*,0,0) = \hat x_0(\theta_*)$ in \eqref{eqn: Init condition ejection curve Theta=0}, yielding
\begin{equation}\label{eqn: Expansion Phi_0 in init cond}
\begin{aligned}
\hat \Phi_0(t,\hat x(\theta,\upsilon,\sigma)) =&\ \hat \Phi_0(t,\hat x_0(\theta_*)) + \partial_{\hat x}\hat \Phi_0(t,\hat x_0(\theta_*))\cdot (\hat x(\theta,\upsilon,\sigma)-\hat x_0(\theta_*))\\
&+\mathcal{O}_2\left(\|\hat x(\theta,\upsilon,\sigma)-\hat x_0(\theta_*)\|\right)
\end{aligned}
\end{equation}
where, relying on \eqref{eqn: Estimates on the parameterization of the ejection curve in CM} and \eqref{eqn: Init condition ejection curve Theta=0}, the difference $\hat x(\theta,\upsilon,\sigma)-\hat x_0(\theta_*)$ is given by
\begin{equation}\label{eqn: Difference init condition}
\begin{aligned}
\hat x(\theta,\upsilon,\sigma)-\hat x_0(\theta_*) =  \begin{pmatrix}\hat r_{\mathcal{J}}^-(\theta,\upsilon)\\ \hat \theta_{\mathcal{J}}^-(\theta,\upsilon) \\ \hat R_{\mathcal{J}}^-(\theta,\upsilon,\sigma)\\ \hat \Theta_{\mathcal J}^-(\theta,\upsilon,\sigma)\end{pmatrix} - \begin{pmatrix}1\\ 0 \\ - \sqrt 2 \\0\end{pmatrix} = \begin{pmatrix}\\ \upsilon\cos\theta + \mathcal{O}_2(\upsilon)\\ \upsilon\sin\theta + \mathcal{O}_2(\upsilon) \\ \sqrt 2+\sqrt{3}\cos\theta  + \mathcal{O}(\upsilon,\sigma)\\ 1+\sqrt 3 \sin\theta +\mathcal{O}(\upsilon,\sigma)\end{pmatrix} = \mathcal{O}(\upsilon,\sigma,\omega)
\end{aligned}
\end{equation}
for $\theta\in (\theta_*-\omega,\theta_*+\omega)$, and $\partial_x \hat \Phi_0(t,\hat x_0(\theta_*))$ is the solution of the variational equations whose motion is expressed as
\begin{equation}\label{eqn:vareq}
    \begin{cases}
    \partial_{\hat x}\hat\Phi_0(t,\hat x_0(\theta_*))' = D\hat F_0\left(\hat \Phi_0\left(t,\hat x(\theta_*,0,0)\right)\right)\cdot \partial_{\hat x}\hat\Phi_0(t,\hat x_0(\theta_*))\\
    \partial_{\hat x}\hat\Phi_0(0,\hat x_0(\theta_*)) = \mathrm{Id},
    \end{cases}
\end{equation}
where $\hat F_0$ is defined in \eqref{eqn: hatF0 and hatF1} so that
\begin{equation}\label{eqn: DhatF0}
D\hat F_0(\hat\Phi_0(t,\hat x_0(\theta_*)) = \begin{pmatrix}0&0&1&0\\0&0&0&\frac{1}{\left(\hat \Phi_0^{\hat r}(t,\hat x_0(\theta_*))\right)^2}\\\frac{2}{\left(\hat \Phi_0^{\hat r}(t,\hat x_0(\theta_*))\right)^3} & 0 & 0& 0\\ 0&0&0&0\end{pmatrix}=\begin{pmatrix}0&0&1&0\\0&0&0&\frac{1}{\left(1-\frac{3t}{\sqrt 2}\right)^{\frac 4 3}}\\\frac{2}{\left(1-\frac{3t}{\sqrt 2}\right)^2} & 0 & 0& 0\\ 0&0&0&0\end{pmatrix}.
\end{equation}
Equation \eqref{eqn:vareq} is integrable, and the solutions are bounded for $t \in \mathcal T$ defined in \eqref{eqn: interval t Phi0}. They are given by
\begin{equation}\label{eqn: partialxPhi0}
    \partial_{\hat x}\hat \Phi_0(t,\hat x_0(\theta_*)) = \begin{pmatrix}\frac{4}{5\left(1-\frac{3t}{\sqrt 2}\right)^{\frac 1 3}} + \frac 1 5\left(1-\frac{3t}{\sqrt 2}\right)^{\frac 4 3} & 0 & \frac{3}{5\left(1-\frac{3t}{\sqrt 2}\right)^{\frac 1 3}} - \frac{3}{5}\left(1-\frac{3t}{\sqrt 2}\right)^{\frac 4 3} & 0\\ 0 & 1 & 0 & \frac{\sqrt 2}{\left(1-\frac{3t}{\sqrt 2}\right)^{\frac 1 3}}\\  \frac{4}{15\left(1-\frac{3t}{\sqrt 2}\right)^{\frac 1 3}} - \frac{4}{15}\left(1-\frac{3t}{\sqrt 2}\right)^{\frac 4 3} & 0 & \frac{1}{5\left(1-\frac{3t}{\sqrt 2}\right)^{\frac 4 3}} + \frac{12}{15}\left(1-\frac{3t}{\sqrt 2}\right)^{\frac 1 3} & 0\\
    0 & 0 &0 & 1\end{pmatrix},
\end{equation}
yielding \eqref{eqn: Estimate flow Phi_0(t,x)}.

We compute the time $t$ satisfying $\hat \Phi_0(t,\hat x(\theta,\upsilon,\sigma)) \in \hat \Sigma$, defined in \eqref{eqn: Section hatr = delta2}.  Namely, we find a zero of a function of the form
\[\mathcal G(t,\theta,\upsilon,\sigma) = \hat \Phi_0^{\hat r}(t,\hat x(\theta,\upsilon,\sigma)) - \delta^2.\]
Note that we have taken $\mathcal G$ as a function of $\upsilon$ and $\sigma$, which depend on the parameter $\mu \in (0,\mu_0)$ as described in \eqref{eqn: notation upsilon, sigma}.

Recall that $\hat \Phi_0(t,\hat x(\theta,\upsilon,\sigma))$ no longer corresponds to the expression provided in \eqref{eqn: Flow 2BP Theta= 0} since $\hat\Theta_{\mathcal{J}}^-(\theta,\upsilon,\sigma)\neq 0$ for $\upsilon,\sigma \neq 0$. Nevertheless, the function $\mathcal G$ satisfies
\[\mathcal{G}(t_0^*,\theta_*,0,0) =0,\quad \partial_t \mathcal G(t_0^*,\theta_*,0,0) = \hat\Phi_0^{\hat R}(t_0^*,\hat x_0(\theta_*)) = -\sqrt{\frac{2}{\delta^2}}\neq 0,\]
where $\hat x_0(\theta_*)$ and $t_0^*$ are defined in \eqref{eqn: Init condition ejection curve Theta=0} and \eqref{eqn: time t01} respectively, and where we use the expression for the unperturbed vector field $\hat F_0$ in \eqref{eqn: hatF0 and hatF1}.
As a result, the Implicit Function Theorem ensures that there exist 
\begin{equation}\label{eqn:sigma,ups,omega}
\sigma_0,\upsilon_0,\omega_0 = \sigma_0(\delta), \upsilon(\delta), \omega_0(\delta) > 0
\end{equation}
such that there exists $t = t(\theta,\upsilon,\sigma)$ smooth, defined for  $\theta \in (\theta_*-\omega,\theta_*+\omega)$ (where $\theta_*$ is defined in \eqref{eqn: theta_*}) and $(\upsilon,\sigma)  \in \left(0,\upsilon_0\right)\times (0,\sigma_0)$ with $t(\theta_*,0,0) = t_0^*$ and
\begin{equation}\label{eqn: Time condition IFT}
\hat \Phi_0^{\hat r}(t(\theta,\upsilon,\sigma),\hat x(\theta,\upsilon,\sigma)) = \delta^2.
\end{equation}
Since $t(\theta_*,0,0) = t_0^* \in \mathcal T$ in \eqref{eqn: interval t Phi0}, the expressions in \eqref{eqn: partialxPhi0} are well-defined and
\[\partial_{\theta}t(\theta_*,0,0) = - \frac{\partial_{\hat x}\hat \Phi_0^{\hat r}(t_0^*,\hat x_0(\theta_*))\cdot \partial_{\theta
}\hat x(\theta_*,0,0)}{\partial_t \hat \Phi_0^{\hat r}(t_0^*,\hat x_0(\theta_*))}.\]
Relying on \eqref{eqn: Difference init condition} and \eqref{eqn: partialxPhi0},  we obtain
\begin{equation*}
\begin{aligned}
\partial_{\theta}\hat x(\theta_*,0,0) &= \begin{pmatrix}0 & 0 & -\sqrt 3 \sin\theta_* &\sqrt 3 \cos\theta_*\end{pmatrix}^T = \begin{pmatrix}0& 0 & 1& -\sqrt 2 \end{pmatrix}^T,\\
\partial_{\hat x}\hat \Phi_0^{\hat r}(t_0^*,\hat x_0(\theta_*)) &=  \begin{pmatrix}\frac{4}{5\left(1-\frac{3t_0^*}{\sqrt 2}\right)^{\frac 1 3}} + \frac 1 5\left(1-\frac{3t_0^*}{\sqrt 2}\right)^{\frac 4 3} & 0 & \frac{3}{5\left(1-\frac{3t_0^*}{\sqrt 2}\right)^{\frac 1 3}} - \frac{3}{5}\left(1-\frac{3t_0^*}{\sqrt 2}\right)^{\frac 4 3}& 0\end{pmatrix} \\
&= \begin{pmatrix} \frac{4}{5\delta} + \frac{\delta^4}{5} & 0 & \frac{3}{5\delta} -\frac{3\delta^4}{5} & 0\end{pmatrix}.
\end{aligned}
\end{equation*}
Therefore $\partial_{\theta}t(\theta_*,0,0) = \frac{3}{5\sqrt 2}\left(1+\delta^5\right)$, leading to \eqref{eqn: Estimate time for Phi_0(t,x)} and completing the proof.
\end{proof}

Now we perform the extension under the flow induced by the vector field $\hat F$ defined in \eqref{eqn: hatF0 and hatF1}, which we denote as $\hat \Phi_\mu$, result of the following lemma.

\begin{lemma}[Extension by the flow $\hat \Phi_\mu$]\label{lemma: Extension phimu}
Let $\delta_0 > 0$ be the parameter given in Proposition \ref{proposition: Perturbed invariant manifolds of collision in synodical polar coordinates J} and fix $\gamma \in\left(\frac{3}{11},\frac 1 3\right)$. Then, for $0<\delta<\delta_0$, there exist $\upsilon_0,\sigma_0,\omega_0> 0$ such that, for $(\upsilon,\sigma,\omega)\in(0,\upsilon_0)\times (0,\sigma_0)\times (0,\omega_0)$, the time needed for the flow $\hat \Phi_\mu$, starting at an initial condition $\hat x \in \hat \Gamma_{\mathcal J}^-(\upsilon,\sigma)$ in \eqref{eqn: reduced curve init cond} (where $\upsilon,\sigma$ are defined in \eqref{eqn: notation upsilon, sigma}), to reach the section $\hat \Sigma$ defined in \eqref{eqn: Section hatr = delta2} satisfies the following estimates
\begin{equation}\label{eqn: t_mu}
\begin{aligned}
    t_\mu(\theta, \upsilon, \sigma) &=t(\theta,\upsilon,\sigma) + \mathcal O(\mu) =  \frac{\sqrt 2}{3}(1-\delta^3) +  \mathcal{O}(\upsilon,\sigma,\omega,\mu) \in \mathcal T,\\
    \partial_{\theta}t_\mu(\theta,\upsilon,\sigma) &= \partial_{\theta}t(\theta,\upsilon,\sigma) + \mathcal{O}(\mu) = \frac{3}{5\sqrt 2}+\mathcal{O}(\upsilon,\sigma,\omega,\mu),
\end{aligned}
\end{equation}
where $\mathcal T$ is defined in \eqref{eqn: interval t Phi0} and $t(\theta,\upsilon,\sigma),\partial_{\theta}t(\theta,\upsilon,\sigma)$ are given in \eqref{eqn: Estimate time for Phi_0(t,x)}. 

Moreover, for $t\in [0,t_\mu(\theta,\upsilon,\sigma)]$, the flow $\hat \Phi_\mu$ satisfies
\begin{equation}\label{eqn: Estimate perturbed flow}
    \begin{aligned}
        \hat \Phi_\mu(t, \hat x) &= \hat \Phi_0(t, \hat x) + \mathcal{O}\left(\mu^{1-2\gamma}\right) = \hat \Phi_0(t, \hat x_0(\theta_*)) + \mathcal{O}(\upsilon,\sigma,\omega,\mu^{1-2\gamma}),\\
        \partial_{\hat x}\hat \Phi_\mu(t, \hat x) &= \partial_{\hat x}\hat \Phi_0(t, \hat x) + \mathcal{O}\left(\mu^{1-3\gamma}\right) = \partial_{\hat x}\hat \Phi_0(t, \hat x_0(\theta_*) + \mathcal{O}(\upsilon,\sigma,\omega,\mu^{1-3\gamma}),\\
    \end{aligned}
\end{equation}
where $\hat x_0(\theta_*)$ is defined in \eqref{eqn: Init condition ejection curve Theta=0}.
\end{lemma}
\begin{proof}[Proof of Lemma \ref{lemma: Extension phimu}]
First, we compute the time needed for an arbitrary point $\hat x(\theta,\upsilon,\sigma) \in \hat \Gamma_{\mathcal J}^-(\sigma,\upsilon)$ defined in \eqref{eqn: reduced curve init cond} to reach the section $\hat \Sigma$ in \eqref{eqn: Section hatr = delta2} under the flow $\hat \Phi_\mu$. Namely, we look for a zero of a  function of the form
\[\mathcal F(t,\theta,\upsilon,\sigma,\mu) = \hat \Phi_\mu^{\hat r}(t,\hat x(\theta,\upsilon,\sigma)) - \hat P\left(\Phi_\mu^{\hat \theta}(t,\hat x(\theta,\upsilon,\sigma)),\mu\right),\]
where $\hat P(\hat \theta,\mu)$ is defined in \eqref{eqn:hat r=delta2} for $\hat \theta \in \mathds{T}$. This function satisfies
\begin{equation*}
\begin{aligned}
\mathcal F(t(\theta,\upsilon,\sigma),\theta,\upsilon,\sigma,0)  &=\hat\Phi_0^{\hat r}(t(\theta,\upsilon,\sigma),x(\theta,\upsilon,\sigma)) - \delta^2 = 0,\\
\end{aligned}
\end{equation*}
where $t(\theta,\upsilon,\sigma)$ is defined in \eqref{eqn: Estimate time for Phi_0(t,x)}. Relying on \eqref{eqn: Flow 2BP Theta= 0} and \eqref{eqn: Estimate flow Phi_0(t,x)} we also obtain
\begin{equation*}
\begin{aligned}
\partial_t \mathcal F(t(\theta,\upsilon,\sigma),\theta,\upsilon,\sigma,0) &= \hat \Phi_0^{\hat R}(t(\theta,\upsilon,\sigma),x(\theta,\upsilon,\sigma)) =  -\sqrt{\frac{2}{\delta^2}} + \mathcal{O}(\upsilon,\sigma,\omega)\neq 0.
\end{aligned}
\end{equation*}
Hence, the Implicit Function Theorem ensures that there exists $\mu_0,\omega_0 > 0$ and a smooth function $t_\mu(\theta,\upsilon, \sigma)$ defined for $\mu \in (0,\mu_0)$, $\omega\in (0,\omega_0)$ and $\theta\in (\theta_*-\omega,\theta_*+\omega)$ (where $ \upsilon\in(0,\upsilon_0)$, $\sigma\in(0,\sigma_0)$ are defined in \eqref{eqn: notation upsilon, sigma} and \eqref{eqn:sigma,ups,omega} respectively) with $ t_0(\theta, \upsilon,\sigma) = t(\theta,\upsilon,\sigma)$ such that 
\[\hat \Phi_\mu(t_\mu(\theta,\upsilon, \sigma),x(\theta,\upsilon,\sigma)) \in \hat \Sigma,\]
and satisfying \eqref{eqn: t_mu}.

Denote by 
\begin{equation}\label{eqn: rmin rmax}
    \hat r_{\mathrm{min}} := \underset{\hat\theta\in\mathds{T}}{\min}\;\hat P(\hat\theta,\mu),\quad \hat r_{\mathrm{max}} := \underset{\theta\in (\theta_*-\omega_0,\theta_*+\omega_0)}{\max}\hat r_{\mathcal{J}}^-(\theta,\upsilon).
\end{equation}
where $\hat P(\hat\theta,\mu)$ and $\hat r_{\mathcal J}^-(\theta,\upsilon)$ are given in \eqref{eqn:hat r=delta2} and \eqref{eqn: Estimates on the parameterization of the ejection curve in CM} respectively. Then, for any point $\hat x = \hat x(\theta,\upsilon,\sigma) \in \hat \Gamma_{\mathcal J}^-(\upsilon,\sigma)$ in \eqref{eqn: reduced curve init cond}, we consider $T \in [0,t_\mu(\theta,\upsilon,\sigma)]$ to be the first time that $\hat \Phi_\mu^{\hat r}(T,\hat x)\notin \left[\hat r_{\mathrm{min}},\hat r_{\mathrm{max}}\right)$ and we study the evolution of such orbit for $t \in [0,T]$. If it does not exist we consider $t\in [0,t_\mu(\theta,\upsilon,\sigma)]$.

Using the mean value theorem we have the following estimate for the flow $\hat\Phi_\mu$ and its derivative $\partial_{\hat x}\hat\Phi_\mu$,
\begin{equation*}
\begin{aligned}
    |\hat \Phi_\mu(t,\hat x) - \hat \Phi_0(t,\hat x)| &\leq |D\hat F_0|\int_0^t |\hat \Phi_\mu(s,x)-\hat\Phi_0(s,x)|\;ds + \int_0^t |\hat F_1|\;ds, \\
    |\partial_{\hat x}\hat \Phi_\mu(t,\hat x) - \partial_{\hat x}\hat \Phi_0(t,\hat x)| &\leq |D^2\hat F_0|\int_0^t |\partial_{\hat x}\hat \Phi_\mu(s,x)-\partial_{\hat x}\hat\Phi_0(s,x)|\;ds + \int_0^t |D\hat F_1|\;ds,
\end{aligned}
\end{equation*}
where $\hat F_0$, $\hat F_1$ are defined in \eqref{eqn: hatF0 and hatF1}. Then, by Gronwall's lemma
\begin{equation}\label{eqn: estimate phimu-phi0}
\begin{aligned}
    |\hat \Phi_\mu(t,\hat x) - \hat \Phi_0(t,\hat x)| \leq t |\hat F_1| e^{|D\hat F_0| t}, \qquad |\partial_{\hat x}\hat \Phi_\mu(t,\hat x) - \partial_{\hat x}\hat \Phi_0(t,\hat x)| \leq t |D\hat F_1| e^{|D^2\hat F_0| t}.\\
\end{aligned}
\end{equation}
We bound both $\hat{F}_1$ and $D\hat F_1$ (which are related to the potential $\hat V$ in \eqref{eqn: potential in polar rotating coordinates centered at CM J}) for $(\hat r, \hat \theta) \in \left[\hat r_{\mathrm{min}}, \hat r_{\mathrm{max}}\right] \times \mathds{T}$, where $\hat r_{\mathrm{min}}, \hat r_{\mathrm{max}}$ are defined in \eqref{eqn: rmin rmax}. To this end, we consider instead the complex domain $(\hat r, \hat \theta) \in \hat D\times \mathds{T}_\sigma$ defined as
\begin{equation}\label{eqn: domain potential F1}
\begin{aligned}
    \hat D = \left\{\hat r \in \mathds{C}\colon  \mathrm{Re}(\hat r) \in \left[\hat r_{\mathrm{min}},\hat r_{\mathrm{max}} \right], |\mathrm{Im}(\hat r)| \leq \sigma\right\},\quad \mathds T_\sigma =\left\{\hat\theta \in \mathds{C}\setminus (2\pi\mathds{Z}) \colon |\mathrm{Im}(\hat \theta)| <\sigma\right\},
\end{aligned}
\end{equation}
where $\sigma = \frac 3 4 \log\left(\frac{1-\mu-\mu^\gamma}{1-\mu}\right)$. The potential $\hat V$ in \eqref{eqn: potential in polar rotating coordinates centered at CM J} is analytic in this domain and its estimate can be split as follows
\[\left|\hat V(\hat r,\hat \theta;\mu)\right| \leq \left|\hat V_{\mathcal S}(\hat r, \hat \theta;\mu)- \hat V_{\mathrm{CM}}(\hat r)\right|+\left|\hat V_{\mathcal J}(\hat r, \hat \theta;\mu)\right|\]
where
\begin{equation*}
    \begin{aligned}
        \hat V_{\mathcal S}(\hat r, \hat \theta;\mu) =& \frac{1-\mu}{\left(\hat r^2 + 2\hat r\mu\cos\hat \theta + \mu^2\right)^{\frac 1 2}},\quad \hat V_{\mathrm{CM}}(\hat r) = \frac{1}{\hat r},\\
        \hat V_{\mathcal J}(\hat r,\hat \theta;\mu) =& \frac{\mu}{\left(\hat r^2 - 2\hat r(1-\mu)\cos\hat \theta + (1-\mu)^2\right)^{\frac 12}}.   
    \end{aligned}
\end{equation*}
Relying on \eqref{eqn: Estimates on the parameterization of the ejection curve in CM}, \eqref{eqn: Estimation polar coordinates S-CM} and \eqref{eqn: rmin rmax}, the domain $\hat D$ in \eqref{eqn: domain potential F1} is considered so that $|\mathrm{Re}(\hat r)| \in (\delta^2-M\mu, 1-M\mu^\gamma)$ for some adequate constant $M$ independent of $\mu$ and $\delta$, meaning we are ``far'' from both the primary $\mathcal S$ and the center of mass. Moreover for $\mu = 0$ we have that $\hat V_{\mathcal S}(\hat r,\hat \theta;0) - \hat V_{\mathrm{CM}}(\hat r) = 0$, and therefore
\[\left|\hat V_{\mathcal S}(\hat r, \hat \theta;\mu)- \hat V_{\mathrm{CM}}(\hat r)\right| \lesssim \mu.\]
Finally, the estimate of $\hat V_{\mathcal J}$ follows the same argument as the one in the proof of Lemma \ref{lemma: Bound for the potential}, yielding
\begin{equation*}
|\hat V(\hat r, \hat \theta)| \lesssim \mu^{1-\gamma}.
\end{equation*}
Hence, from the Cauchy's estimates we obtain, for $(\hat r, \hat \theta )\in [\hat r_{\mathrm{min}}, \hat r_{\mathrm{max}}]\times \mathds{T}$
\begin{equation}\label{eqn:leadingtermpotential}
    |\hat F_1| \lesssim \mu^{1-2\gamma},\qquad |D\hat F_1| \lesssim \mu^{1-3\gamma}.
\end{equation}
To compute the estimates for \eqref{eqn: estimate phimu-phi0}  we also bound both $|D\hat F_0|$ and $|D^2\hat F_0|$, where $\hat F_0$ is defined in \eqref{eqn: hatF0 and hatF1}. Both functions depend on $(\hat r, \hat \theta, \hat \Theta)$, where $\hat r \in [\hat r_{\mathrm{min}}, \hat r_{\mathrm{max}}], \hat \theta \in \mathds{T}$ and $\hat \Theta$ satisfies
\[\hat \Phi_\mu^{\hat\Theta}(t,\hat x) = \hat \Theta_{\mathcal J}^-(\theta,\upsilon,\sigma) + \int_0^t \hat F_1^{\hat \Theta}(\hat \Phi_\mu(s,\hat x))\;ds,\]
where $\hat \Theta_{\mathcal J}^-(\theta,\upsilon,\sigma)$ is defined in \eqref{eqn: Estimates on the parameterization of the ejection curve in CM} for $\theta\in (\theta_*-\omega, \theta_*+\omega)$ (with $\theta_*$ defined in \eqref{eqn: theta_*}), $t\in[0,t_\mu(\theta,\upsilon,\sigma)]\subset \mathcal T$ defined in \eqref{eqn: interval t Phi0} and $\hat F_1$ satisfies \eqref{eqn:leadingtermpotential}. Hence, we can bound both $|D\hat F_0|$ and $|D^2\hat F_0|$ by an uniform constant $C_0$, leading to \eqref{eqn: Estimate perturbed flow} and completing the proof.

\end{proof}

\paragraph{Step 4. Parameterization of the image curve as a graph:} We express the image curve in coordinates $(\rS,\thetaS,\RS,\ThetaS)$ (defined in \eqref{eqn: Change S CM polar}) and parameterize it as a graph in terms of $\thetaS$. We consider $t = t_\mu =  t_\mu(\theta, \upsilon, \sigma).$ Then, we denote by  
\begin{equation*}
\begin{aligned}
    \hat \theta(\theta,\upsilon,\sigma) := \hat \Phi_\mu^{\hat \theta}(t_\mu,\hat x(\theta,\upsilon,\sigma)),\quad \hat R(\theta,\upsilon,\sigma) := \hat \Phi_\mu^{\hat R}(t_\mu,\hat x(\theta,\upsilon,\sigma)),\quad \hat \Theta(\theta,\upsilon,\sigma) :=\Phi_\mu^{\hat \Theta}(t_\mu,\hat x(\theta,\upsilon,\sigma)).
\end{aligned}
\end{equation*}
We compute $\hat R(\theta,\upsilon,\sigma)$ using the estimate in \eqref{eqn: Estimate perturbed flow}, obtaining
\begin{equation}\label{eqn: hat R in overline r=delta2}
     \hat R(\theta,\upsilon,\sigma)  = -\sqrt{\frac{2}{\delta^2}} +\mathcal{O}(\upsilon,\sigma,\omega,\mu^{1-2\gamma}). 
\end{equation}    
The expressions for $\hat \theta(\theta,\upsilon,\sigma), \hat \Theta(\theta,\upsilon,\sigma)$ and its derivatives $\partial_{\theta}\hat \theta(\theta,\upsilon,\sigma)$, $\partial_{\theta}\hat \Theta(\theta,\upsilon,\sigma)$ are obtained from both \eqref{eqn: Estimate perturbed flow} and the expansion in \eqref{eqn: Expansion Phi_0 in init cond}. Both functions satisfy
\begin{equation}\label{eqn: hatthetaSun}
\begin{aligned}
z(\theta,\upsilon,\sigma) =&\  \hat \Phi_0^{z}(t_\mu,x_0(\theta_*)) + \partial_{\hat x} \hat\Phi_0^{z}(t_\mu,\hat x_0(\theta_*))\cdot\left(\hat x(\theta,\upsilon,\sigma)-\hat x(\theta_*,0,0)\right) \\
&+ \mathcal{O}_2\left(\|\hat x(\theta,\upsilon,\sigma)-\hat x_0(\theta_*)\|\right) + \mathcal{O}(\mu^{1-2\gamma})\\
\partial_{\theta}z(\theta,\upsilon,\sigma) =&\ \partial_t \hat \Phi_\mu^{\hat \theta}(t_\mu,\hat x(\theta,\upsilon,\sigma))\cdot \partial_{\theta}t_\mu + \partial_{\hat x}\hat \Phi_\mu^{z}(t_\mu,\hat x(\theta,\upsilon,\sigma))\cdot \partial_{\theta}\hat x(\theta,\upsilon,\sigma)\\
=&\ \partial_t \hat \Phi_\mu^{z}(t_\mu,\hat x(\theta,\upsilon,\sigma))\cdot \partial_{\theta}t_\mu \\
&+\left(\partial_{\hat x}\hat \Phi_0(t, \hat x( \theta_*,0,0)) + \mathcal{O}(\upsilon,\sigma,\theta-\theta_*,\mu^{1-3\gamma})\right)\cdot \partial_{\theta}\hat x(\theta,\upsilon,\sigma)
\end{aligned}
\end{equation}
where $z$ denotes either the function $\hat\theta$ or $\hat \Theta$. To obtain estimates of \eqref{eqn: hatthetaSun}, first we recall the estimates in \eqref{eqn: Difference init condition}, \eqref{eqn: partialxPhi0}, \eqref{eqn: t_mu} and \eqref{eqn: Estimate perturbed flow}.
\begin{alignat}{2}
    t_\mu &= -\frac{\sqrt 2}{3}(1-\delta^3) + \mathcal{O}(\upsilon,\sigma,\omega,\mu),\qquad&& \partial_{\theta}t_\mu = \frac{3}{5\sqrt 2}+ \mathcal{O}(\upsilon,\sigma,\omega,\mu),\\
    \hat \Phi_0^{\hat \theta}(t_\mu,x_0(\theta_*)) &= -t_\mu, \quad&&
    \hat \Phi_0^{\hat \Theta}(t_\mu,x_0(\theta_*)) = \mathcal{O}(\upsilon,\sigma,\omega,\mu^{1-2\gamma}),\\ 
    \partial_{\hat x} \hat\Phi_0^{\hat \theta}(t_\mu,\hat x_0(\theta_*)) &= \begin{pmatrix}0 & 1 &0& \sqrt 2 \left(1-\frac{3t_\mu}{\sqrt 2}\right)^{-\frac 1 3}\end{pmatrix},\quad&& \partial_{\hat x} \hat \Phi_0^{\hat \Theta}(t_\mu,\hat x_0(\theta_*)) = \begin{pmatrix}0&0&0&1\end{pmatrix},\label{eqn:midestimates}\\
    \partial_t \hat \Phi_\mu^{\hat \theta}(t_\mu,\hat x(\theta,\upsilon,\sigma)) &= \frac{\hat\Theta(\theta,\upsilon,\sigma)}{\delta^4}-1,\quad&& \partial_t \hat \Phi_\mu^{\hat \Theta}(t_\mu,\hat x(\theta,\upsilon,\sigma)) = \partial_{\hat\theta}\hat V(\delta^2,\hat\theta(\theta,\upsilon,\sigma),\mu),\\  
    \hat x(\theta,\upsilon,\sigma)-\hat x(\theta_*,0,0) &=\begin{pmatrix}\upsilon \cos\theta + \mathcal{O}_2(\upsilon)\\\upsilon \sin\theta + \mathcal{O}_2(\upsilon)\\ \sqrt 2 + \sqrt 3\cos\theta +\mathcal O(\upsilon,\sigma)\\ 1+ \sqrt 3 \sin\theta + \mathcal{O}(\upsilon,\sigma)\end{pmatrix},\quad&& \partial_{\theta}\hat x(\theta,\upsilon,\sigma) =\begin{pmatrix}-\upsilon \sin\theta + \mathcal{O}_2(\upsilon)\\\upsilon \cos\theta + \mathcal{O}_2(\upsilon)\\ -\sqrt 3\sin\theta +\mathcal O(\upsilon,\sigma)\\  \sqrt 3 \cos\theta + \mathcal{O}(\upsilon,\sigma)\end{pmatrix}.
\end{alignat}
Since $\theta\in (\theta_*-\omega,\theta_*+\omega)$ for $\theta_*$ defined in \eqref{eqn: theta_*} we have that
\begin{equation}\label{eqn:estimatethetaw0}
1+\sqrt 3\sin\theta = -\sqrt 2(\theta-\theta_*) + \mathcal{O}_2(\theta-\theta_*).
\end{equation}
Substituting both \eqref{eqn:midestimates} and \eqref{eqn:estimatethetaw0} in \eqref{eqn: hatthetaSun} we obtain 
%
\begin{equation}\label{eqn:Estimatehattheta}
\begin{aligned}
    \hat \theta(\theta,\upsilon,\sigma) =& -t_\mu + \begin{pmatrix}0 & 1 &0& \sqrt 2 \left(1-\frac{3t_\mu}{\sqrt 2}\right)^{-\frac 1 3}\end{pmatrix}\cdot \begin{pmatrix}\upsilon \cos\theta + \mathcal{O}_2(\upsilon)\\\upsilon \sin\theta + \mathcal{O}_2(\upsilon)\\ \sqrt 2 + \sqrt 3\cos\theta +\mathcal O_1(\upsilon,\sigma)\\ 1+ \sqrt 3 \sin\theta + \mathcal{O}(\upsilon,\sigma)\end{pmatrix} \\
    &+ \mathcal{O}_2(\upsilon,\sigma, \omega) + \mathcal O(\mu^{1-2\gamma}) \\
    =& -\frac{\sqrt 2}{3}(1-\delta^3) + \mathcal{O}(\upsilon,\sigma,\omega,\mu^{1-2\gamma}) \\
    &+ \frac{\sqrt 2}{\delta}\left(1+\delta^{-3}\mathcal{O}(\upsilon,\sigma,\omega,\mu)\right)\left(-\sqrt 2(\theta-\theta_*)  + \mathcal{O}\left(\upsilon,\sigma,\omega^2\right)\right) .
\end{aligned}
\end{equation}
and
\begin{equation}\label{eqn: hat Theta in overline r=delta2}
\begin{aligned}
    \hat \Theta(\theta,\upsilon,\sigma) =&\ \begin{pmatrix}0 & 0 &0& 1\end{pmatrix}\cdot\begin{pmatrix}\upsilon \cos\theta + \mathcal{O}_2(\upsilon)\\\upsilon \sin\theta + \mathcal{O}_2(\upsilon)\\ \sqrt 2 + \sqrt 3\cos\theta +\mathcal O_1(\upsilon,\sigma)\\ 1+ \sqrt 3 \sin\theta + \mathcal{O}(\upsilon,
    \sigma)\end{pmatrix} +\mathcal{O}_2(\upsilon,\sigma, \omega) + \mathcal O(\mu^{1-2\gamma}) \\
    =&\; -\sqrt 2(\theta-\theta_*) + \mathcal{O}(\upsilon,\sigma,\omega^2,\mu^{1-2\gamma}).
\end{aligned}
\end{equation}
The derivative $\partial_{\theta}\hat\theta(\theta,\upsilon,\sigma)$ is given by
\begin{equation}\label{eqn:derthetaJhattheta}
\begin{aligned}
    \partial_{\theta}\hat \theta(\theta,\upsilon,\sigma) =& \left(\frac{\hat \Theta(\theta,\upsilon,\sigma)}{\delta^4}-1\right)\cdot\left(\frac{3}{5\sqrt 2} + \mathcal{O}(\upsilon,\sigma,\omega,\mu)\right) \\
    &+\left(\begin{pmatrix}0 & 1 &0& \sqrt 2 \left(1-\frac{3t_\mu}{\sqrt 2}\right)^{-\frac 1 3}\end{pmatrix} + \mathcal{O}\left(\upsilon,\sigma,\omega,\mu^{1-3\gamma}\right)\right)\cdot \begin{pmatrix}-\upsilon \sin\theta + \mathcal{O}_2(\upsilon)\\\upsilon \cos\theta + \mathcal{O}_2(\upsilon)\\ -\sqrt 3\sin\theta +\mathcal O_1(\upsilon,\sigma)\\  \sqrt 3 \cos\theta + \mathcal{O}(\upsilon,\sigma)\end{pmatrix}\\
    =& -\frac 2 \delta  - \frac{3}{5\sqrt 2} + \delta^{-4}\mathcal O(\upsilon,\sigma,\omega,\mu) \neq 0,
\end{aligned}
\end{equation}
where we have considered $\mu_0 \ll \delta^4$, and therefore from \eqref{eqn:sigma,ups,omega} both $\upsilon_0,\sigma_0\ll\delta^4$ and $\omega_0\ll \delta^4$. 

Relying on \eqref{eqn:leadingtermpotential}, the derivative $\partial_\theta \hat \Theta(\theta,\upsilon,\sigma)$ is given by
\begin{equation}\label{eqn:derhatThetathetabar}
    \begin{aligned}
        \partial_{\theta}\hat\Theta(\theta,\upsilon,\sigma) =&\ \partial_{\hat\theta}\hat V(\delta^2,\hat\theta(\theta,\upsilon,\sigma),\mu)\cdot\left(\frac{3}{5\sqrt 2} + \mathcal{O}(\upsilon,\sigma,\omega,\mu)\right) \\
        &+\left(\begin{pmatrix}0&0&0&1\end{pmatrix}+ \mathcal{O}(\upsilon,\sigma,\omega,\mu^{1-3\gamma})\right)\cdot \begin{pmatrix}-\upsilon \sin\theta + \mathcal{O}_2(\upsilon)\\\upsilon \cos\theta + \mathcal{O}_2(\upsilon)\\ -\sqrt 3\sin\theta +\mathcal O(\upsilon,\sigma)\\  \sqrt 3 \cos\theta + \mathcal{O}(\upsilon,\sigma)\end{pmatrix}\\
        =&\; \sqrt 3\cos\theta\left(1+\mathcal{O}(\upsilon,\sigma,\omega,\mu^{1-3\gamma})\right) = -\sqrt 2\left(1+\mathcal{O}(\upsilon,\sigma,\omega,\mu^{1-3\gamma})\right) \neq 0.
    \end{aligned}
\end{equation}
To have a parameterization as in \eqref{eqn: Curve intersection Ejection-overline r=delta2}, we translate the result in \eqref{eqn:Estimatehattheta} into coordinates $(\rS,\thetaS,\RS,\ThetaS)$ in \eqref{eqn: Change S CM polar}. Using both estimates in \eqref{eqn: Estimation polar coordinates S-CM} and \eqref{eqn:derthetaJhattheta} we obtain 
\begin{equation}
\begin{aligned}
\thetaS(\theta,\upsilon,\sigma) &= \hat \theta( \theta,\upsilon,\sigma) + \mathcal{O}(\mu) = -\frac{\sqrt 2}{3}(1-\delta^3) -\frac{2}{\delta}(\theta-\theta_*) + \mathcal{O}(\upsilon,\sigma,\omega^2,\mu),\\
\partial_{\theta}\thetaS(\theta,\upsilon,\sigma) &= \partial_{\theta}\hat\theta(\theta,\upsilon,\sigma) + \mathcal O(\mu) = -\frac 2 \delta -\frac{3}{5\sqrt 2}+\delta^{-4}\mathcal{O}(\upsilon,\sigma,\omega,\mu) \neq 0.
\end{aligned}
\end{equation}
Since $\omega_0\ll \delta^4$, the image of $\thetaS(\theta,\upsilon,\sigma)$ belongs to a $\delta^4$-neighborhood of $-\frac{\sqrt 2}{3}(1-\delta^3)$, leading to the domain $B(\thetaS_0)$ provided in \eqref{eqn: Curve intersection Ejection-overline r=delta2}. Moreover, the Inverse Function Theorem defines an inverse for $\thetaS( \theta,\upsilon,\sigma)$, which we denote $\theta(\thetaS,\upsilon,\sigma)$, in $B(\thetaS_0)$. The estimates for $\theta(\thetaS,\upsilon,\sigma)$ and its derivative are given by
\begin{equation}\label{eqn:thetabartheta}
\begin{aligned}
     \theta(\thetaS,\upsilon,\sigma) &= \thetaS_* -\frac \delta 2\left(\thetaS + \frac{\sqrt 2}{3}(1-\delta^3)\right) + \delta\mathcal{O}(\upsilon,\sigma,\omega^2,\mu),\\
    \partial_{\thetaS}\theta(\thetaS,\upsilon,\sigma) &= \frac{1}{\partial_{\theta}\thetaS(\theta,\upsilon,\sigma)} = -\frac \delta 2 \left(1+\mathcal{O}\left(\delta\right)\right).
\end{aligned}
\end{equation}
Finally, substituting these expressions in both \eqref{eqn: hat Theta in overline r=delta2} and \eqref{eqn:derhatThetathetabar} leads to
\begin{equation*}
    \begin{aligned}
        \hat\Theta(\thetaS,\upsilon,\sigma) &= \frac{\delta}{\sqrt 2}\left(\thetaS+\frac{\sqrt 2}{3}(1-\delta^3)\right) + \delta\mathcal{O}(\upsilon,\sigma,\omega^2,\mu^{1-2\gamma}),\\
        \partial_{\thetaS}\hat\Theta(\thetaS,\upsilon,\sigma) &= \partial_{\theta}\hat\Theta\cdot \partial_{\thetaS}\theta(\thetaS,\upsilon,\sigma) = \left(-\sqrt 2\left(1+\mathcal{O}\left(\upsilon,\sigma,\omega,\mu^{1-3\gamma}\right)\right)\right)\cdot \left(-\frac \delta 2(1+\mathcal O(\delta)\right) \\
        &= \frac{\delta}{\sqrt 2}\left(1+\mathcal O(\delta,\mu^{1-3\gamma})\right).
    \end{aligned}
\end{equation*}
Translating the result using the estimates in \eqref{eqn: Estimation polar coordinates S-CM} completes the proof.

\section{Proof of Lemma \ref{lemma: tangent vector at Sigma delta,h as a graph}}\label{appendix: From Levi-Civita coordinates to synodical polar coordinates}
Consider the value of the energy for which we prove the existence of a triple intersection in Section \ref{subsec: existence triple intersection}:
\begin{equation}\label{eqn:htrip}
h^* = -\hat\Theta_0^*= -1+\mathcal{O}\left(\mu^{\frac{11\nu-3}{8}}\right)
\end{equation}
so $\xi^* =(2h^* +3)^{-\frac 1 2} + \mathcal{O}(\mu) = 1 + \mathcal{O}\left(\mu^{\frac{11\nu-3}{8}}\right)> 0$ (see \eqref{eqn: Energy xi}). Denote by
\begin{equation*}
    \gamma_{\infty}^{u,<}:=\left\{\left(\theta, \Theta_\infty^{u,<}(\theta)\right)\colon \theta \in(\theta_--\iota,\theta_-+\iota)\right\}\subset \Sigma_\nu^<
\end{equation*}
the $C^1$ graph parameterization of the curve $\Gamma_\infty^{u,<}$ in \eqref{eqn: curve gammainf<}, where $\theta_-$ and $\Theta_\infty^{u,<}(\theta)$ are defined in \eqref{eqn: notation transversality} and $\Sigma_\nu^<$ is defined in \eqref{eqn: Transverse sections}. We apply the change of coordinates $\Upsilon$ in \eqref{eqn: (q,p) - (r,theta,R,Theta)} to translate the curve $\gamma_{\infty}^{u,<}$ into coordinates $(q,p)$. This leads to a curve parameterized in terms of $ \theta$ of the form
\begin{equation}\label{eqn: curve inf in coord (tilde q, p)}
     \gamma_{\infty}^{u,<} =\left(q_\infty^{u,<}(\theta),  p_{\infty}^{u,<}(\theta)\right)
\end{equation}
such that
\begin{equation}\label{eqn: qinf,pinf}
    \begin{aligned}
        q_\infty^{u,<}( \theta) =& \left(q_1^{u,<}( \theta),q_2^{u,<}( \theta)\right) = \mu^\nu(\cos \theta,\sin\theta),\\
         p_\infty^{u,<}(\theta) =& \left( p_1^{u,<}(\theta), p_2^{u,<}(\theta)\right) =  R_\infty^u(\theta)(\cos\theta,\sin\theta) + \frac{\Theta_\infty^u(\theta)}{\mu^\nu}(-\sin\theta,\cos\theta).
    \end{aligned}
\end{equation}
Relying on \eqref{eqn:htrip}, \eqref{eqn: qinf,pinf} and the estimates for both the angular and radial momenta $\Theta_\infty^u(\theta)$ and $ R_\infty^u(\theta)$ (and their derivatives) given in \eqref{eqn: R_infu and Theta_infu}, we obtain the following estimates for $p_\infty^{u,<}(\theta)$ and its derivative
\begin{equation}\label{eqn: approximation p_inf}
\begin{aligned}
     p_{\infty}^{u,<}(\theta) 
    = (-1,0) + \mathcal{O}\left(\mu^{\frac{11\nu-3}{8}}\right),\quad  \partial_\theta\, p_\infty^{u,<}(\theta) 
    = \mathcal{O}\left(\mu^{\frac{11\nu-3}{8}},\mu^{1-3\nu}\right).
\end{aligned}
\end{equation}
%
Once in coordinates $(q,p)$, we apply the diffeomorphism $\psi^{-1}$ in \eqref{eqn: Change of coordinates Levi-Civita} to express the curve $ \gamma_\infty^{u,<}$ in \eqref{eqn: curve inf in coord (tilde q, p)} in Levi-Civita coordinates $(z,w)$, yielding the following parameterization
\begin{equation}\label{eqn: curv inf in coord (z,w)}
    \upgamma_{\infty}^{u,<}(\theta) = \left(z^{u,<}_\infty(\theta), w^{u,<}_\infty(\theta)\right)
\end{equation}
which, from \eqref{eqn:htrip}, \eqref{eqn: qinf,pinf} and \eqref{eqn: approximation p_inf},  satisfies 
\begin{equation}\label{eqn: z_inf, w_inf}
    \begin{aligned}
        z_\infty^{u,<}(\theta) &= (z_1^{u,<}(\theta), z_2^{u,<}(\theta)) =\pm \frac 1 2 \left(\frac{q_2^{u,<}(\theta)}{\sqrt{\mu^\nu - q_1^{u,<}(\theta)}}, \sqrt{\mu^\nu- q_1^{u,<}(\theta)}\right) \\
        &= \pm\frac 1 2 \mu^{\frac \nu 2}\left(\sqrt{1+\cos\theta},\sqrt{1-\cos\theta}\right),\\
        w_\infty^{u,<}(\theta) &= (w_1^{u,<}(\theta),w_2^{u,<}(\theta)) \\
        &= \xi\left(p_1^{u,<}(\theta) z_1^{u,<}(\theta)+p_2^{u,<}(\theta)z_2^{u,<}(\theta),p_2^{u,<}(\theta)z_1^{u,<}(\theta)-p_1^{u,<}(\theta)z_2^{u,<}(\theta)\right) \\
        &= \mp\frac{\mu^{\frac \nu 2}}{2}\left(\sqrt{1+\cos\theta}, -\sqrt{1-\cos\theta}\right) +\mathcal{O}\left(\mu^{\frac{3(5\nu-1)}{8}}\right).\\
    \end{aligned}
\end{equation}
Relying on \eqref{eqn: qinf,pinf}, \eqref{eqn: approximation p_inf} and \eqref{eqn: z_inf, w_inf}, its tangent vector $\partial_{\theta}\,\upgamma_{\infty}^{u,<}(\theta) = \left(\partial_\theta\, z^{u,<}_\infty(\theta), \partial_\theta\,w^{u,<}_\infty(\theta)\right)$ is given by
\begin{equation}\label{eqn: z_inf',w_inf'}
\begin{aligned}
    \partial_\theta\,z_{\infty}^{u,<}(\theta) =&\ (\partial_\theta\,z_1^{u,<}(\theta), \partial_\theta\,z_2^{u,<}(\theta))
    = \frac{\mu^{\frac \nu 2}}{4}\left(\mp \frac{\sin\theta}{\sqrt{1+\cos\theta}}, \pm \frac{\sin\theta}{\sqrt{1-\cos\theta}}\right),\\
    \partial_\theta\,w_{\infty}^{u,<}(\theta) =&\ \left(\partial_\theta\,w_1^{u,<}(\theta),\partial_\theta\,w_2^{u,<}(\theta)\right)
    = \pm\frac{\mu^{\frac \nu 2}}{4}\left(\frac{\sin\theta}{\sqrt{1+\cos\theta}}, \frac{\sin\theta}{\sqrt{1-\cos\theta}}\right) + \mathcal{O}\left(\mu^{\frac{3(5\nu-1)}{8}},\mu^{1-\frac 5 2\nu}\right).
\end{aligned}
\end{equation}
We apply the diffeomorphism  $\Gamma$ defined in \eqref{eqn: straightening (z,w)-(s,u)} to express the curve $ \upgamma_{\infty}^{u,<}$ in \eqref{eqn: curv inf in coord (z,w)} in coordinates $(s,u)$, yielding a parameterization of the form 
\begin{equation}\label{eqn: curv inf in coord (x,y)}
   \boldsymbol\gamma_{\infty}^{u,<}(\theta) = \left(s_\infty^{u,<}(\theta),u_\infty^{u,<}(\theta)\right).
\end{equation}
Relying on \eqref{eqn: z_inf, w_inf} we have
\begin{equation}\label{eqn: sinf,uinf}
    \begin{aligned}
        s_\infty^{u,<}(\theta) &= (s_1^{u,<}(\theta),s_2^{u,<}(\theta))
        = \pm\frac{\mu^{\frac \nu 2}}{\sqrt 2}\left(\sqrt{1+\cos\theta},0\right) + \mathcal{O}\left(\mu^{\frac{3(5\nu-1)}{8}}\right),\\
        u_\infty^{u,<}(\theta) &= (u_1^{u,<}(\theta),u_2^{u,<}(\theta))
        =\pm\frac{\mu^{\frac \nu 2}}{\sqrt 2}\left(0,\sqrt{1-\cos\theta}\right) + \mathcal{O}\left(\mu^{\frac{3(5\nu-1)}{8}}\right).
    \end{aligned}
\end{equation}
From \eqref{eqn: z_inf',w_inf'} its tangent vector $\boldsymbol \gamma_\infty^{u,<}{}'(\theta) = \left(s_\infty^{u,<}{}'(\theta),u_\infty^{u,<}{}'(\theta)\right)$ admits the following estimates
\begin{equation}\label{eqn: sinf',uinf'}
    \begin{aligned}
        \partial_\theta\,s_\infty^{u,<}(\theta) &= (\partial_\theta\,s_1^{u,<}(\theta),\partial_\theta\,s_2^{u,<}(\theta))
        = \mp\frac{\mu^{\frac{\nu}{2}}}{2\sqrt 2}\left(\frac{\sin\theta}{\sqrt{1+\cos\theta}},0\right) + \mathcal{O}\left(\mu^{\frac{3(5\nu-1)}{8}},\mu^{1-\frac{5}{2}\nu}\right),\\
        \partial_\theta\,u_\infty^{u,<}(\theta) &=(\partial_\theta\,u_1^{u,<}(\theta),\partial_\theta\,u_2^{u,<}(\theta))
        =\pm\frac{\mu^{\frac{\nu}{2}}}{2\sqrt 2}\left(0,\frac{\sin\theta}{\sqrt{1-\cos\theta}}\right)+ \mathcal{O}\left(\mu^{\frac{3(5\nu-1)}{8}},\mu^{1-\frac{5}{2}\nu}\right).
    \end{aligned}
\end{equation}
To apply the transition map $ \mathbf f$ in Proposition \ref{prop: transition map close to collision} we have to compute the norm of both components $s_\infty^{u,<}(\theta)$ and $u_\infty^{u,<}(\theta)$ in \eqref{eqn: sinf,uinf},
\begin{equation}\label{eqn: Norm (x,y)}
\begin{aligned}    
    |s_\infty^{u,<}(\theta)| &= \frac{1}{\sqrt 2}\mu^{\frac \nu 2}\sqrt{1+\cos\theta}\left(1+ \mathcal{O}\left(\mu^{\frac{11\nu-3}{8}}\right)\right) = \mu^{\frac{\nu}{2}}\left(1+\mathcal{O}_2(\iota)\right),\\
    |u_\infty^{u,<}(\theta)| &= \frac{1}{\sqrt 2}\mu^{\frac \nu 2}\sqrt{1-\cos\theta}\left(1+ \mathcal{O}\left(\mu^{\frac{11\nu-3}{8}}\right)\right) = \mu^{\frac{\nu}{2}}\mathcal{O}(\iota),
\end{aligned}
\end{equation}
where we have used that $\theta \in (\theta_- - \iota,\theta_-+\iota)$ for $\theta_-$ defined in \eqref{eqn: notation transversality}. Since $|s_\infty^{u,<}(\theta)| > |u_\infty^{u,<}(\theta)|$, we obtain $\boldsymbol{\gamma}_\infty^{u,<}\subset \boldsymbol{\Sigma}^+$ in \eqref{eqn: bold sigma+-}. Moreover, from \eqref{eqn: sinf,uinf} and \eqref{eqn: sinf',uinf'} we estimate their derivatives as
\begin{equation}\label{eqn: dernormsu}
    \begin{aligned}
        |s_{\infty}^{u,<}(\theta)|' =& \frac{s_1^{u,<}(\theta)\partial_\theta s_1^{u,<}(\theta)+s_2^{u,<}(\theta)\partial_\theta s_2^{u,<}(\theta)}{|s_{\infty}^{u,<}(\theta)|}  = \mp \frac{\mu^{\frac \nu 2}}{2\sqrt 2}\left(\sqrt{1-\cos\theta} + \mathcal{O}\left(\mu^{\frac{11\nu-3}{8}},\mu^{1-3\nu}\right)\right),\\
        |u_{\infty}^{u,<}(\theta)|' =& \frac{u_1^{u,<}(\theta)\partial_\theta u_1^{u,<}(\theta)+u_2^{u,<}(\theta)\partial_\theta u_2^{u,<}(\theta)}{|u_{\infty}^{u,<}(\theta)|} = \pm \frac{\mu^{\frac \nu 2}}{2\sqrt 2}\left(\sqrt{1+\cos\theta} + \mathcal{O}\left(\mu^{\frac{11\nu-3}{8}},\mu^{1-3\nu}\right)\right).
    \end{aligned}
\end{equation}
We apply the transition map $\mathbf f$ from Proposition \ref{prop: transition map close to collision} to the curve $\boldsymbol \gamma_{\infty}^{u,<} \subset \boldsymbol{\Sigma}^+$, which yields the following curve
\begin{equation}\label{eqn: curv gamma_inf >}
\begin{aligned}
\boldsymbol\gamma_{\infty}^{u,>} (\theta) =  \mathbf f(\boldsymbol \gamma_\infty^{u,<}(\theta)) = (s_\infty^{u,>}(\theta),u_\infty^{u,>}(\theta)) =& \left(\frac{|u_\infty^{u,<}(\theta)|}{|s_\infty^{u,<}(\theta)|}\cdot s_\infty^{u,<}(\theta),\;  \frac{u_\infty^{u,<}(\theta)}{|u_\infty^{u,<}(\theta)|}\cdot |s_\infty^{u,<}(\theta)|\right) \\
&+ \mathcal{O}_{2}(|s_\infty^{u,>}(\theta)|).
\end{aligned}
\end{equation}
Relying on \eqref{eqn: sinf,uinf} and \eqref{eqn: Norm (x,y)} we obtain
\begin{equation}\label{eqn: (x_inf,y_inf)>}
    \begin{aligned}
        s_\infty^{u,>}(\theta) =& \left(s_1(\theta),\;s_2(\theta)\right) = \pm\frac{\mu^{\frac{\nu}{2}}}{\sqrt 2}\left(\sqrt{1-\cos\theta}, 0\right) + \mathcal{O}\left(\mu^{\frac{3(5\nu-1)}{8}}\right),\\
        u_\infty^{u,>}(\theta) =& \left(u_1(\theta) ,\;u_2(\theta)\right) = \pm\frac{\mu^{\frac{\nu}{2}}}{\sqrt 2}\left(0,\sqrt{1+\cos\theta}\right) + \mathcal{O}\left(\mu^{\frac{3(5\nu-1)}{8}}\right).
    \end{aligned}
\end{equation}
We compute the tangent vector $\partial_\theta\,\boldsymbol \gamma_{\infty}^{u,>}(\theta) = \left(\partial_\theta\,s_\infty^{u,>}(\theta),\partial_\theta\,u_\infty^{u,>}(\theta)\right)$ from \eqref{eqn: sinf,uinf}, \eqref{eqn: sinf',uinf'},  \eqref{eqn: Norm (x,y)} and \eqref{eqn: dernormsu} yielding
\begin{equation}\label{eqn:dersu}
    \begin{aligned}
        \partial_\theta\,s_\infty^{u,>}(\theta) =& \left(s_1'(\theta),\;s_2'(\theta)\right)\\
        =&\ \partial_\theta\,s_{\infty}^{u,<}(\theta)\frac{|u_{\infty}^{u,<}(\theta)|}{|s_\infty^{u,<}(\theta)|} + s_{\infty}^{u,<}(\theta)\left(\frac{|u_{\infty}^{u,<}(\theta)|'}{|s_\infty^{u,<}(\theta)|} -\frac{|u_{\infty}^{u,<}(\theta)|\cdot|s_\infty^{u,<}(\theta)|'}{|s_\infty^{u,<}(\theta)|^2}\right) + \mathcal{O}(\mu^\nu)\\
        =&\pm \frac{\mu^{\frac \nu 2}}{2\sqrt 2}\left(\sqrt{1+\cos\theta},0\right) + \mathcal{O}\left(\mu^{\frac{3(5\nu-1)}{8}},\mu^{1-\frac 5 2\nu}\right),\\
        \partial_\theta\,u_\infty^{u,>}(\theta) =& \left(u_1'(\theta),\;u_2'(\theta)\right)\\
        =&\ \partial_\theta\,u_{\infty}^{u,<}(\theta)\frac{|s_\infty^{u,<}(\theta)|}{|u_{\infty}^{u,<}(\theta)|} + u_{\infty}^{u,<}(\theta)\left(\frac{|s_\infty^{u,<}(\theta)|'}{|u_{\infty}^{u,<}(\theta)|} - \frac{|s_\infty^{u,<}(\theta)|\cdot|u_{\infty}^{u,<}(\theta)|'}{|u_{\infty}^{u,<}(\theta)|^2}\right)+ \mathcal{O}(\mu^\nu)\\
        =&\ \mp\frac{\mu^{\frac \nu 2}}{2\sqrt 2}\left(0,\sqrt{1-\cos\theta}\right) + \mathcal{O}\left(\mu^{\frac{3(5\nu-1)}{8}},\mu^{1-\frac 5 2\nu}\right).
    \end{aligned}
\end{equation}
We apply the changes $\Gamma^{-1}$, $\psi$ in \eqref{eqn: straightening (z,w)-(s,u)} and \eqref{eqn: Change of coordinates Levi-Civita}  on the curve $\boldsymbol \gamma_\infty^{u,>} \subset \boldsymbol{\Sigma}^-$ in \eqref{eqn: curv gamma_inf >} (see \eqref{eqn: bold sigma+-} for the definition of the section $\boldsymbol{\Sigma}^-$) to obtain the curve $ \gamma_{\infty}^{u,>}\subset \Sigma_\nu^>$ (see \eqref{eqn: Transverse sections}) in coordinates $(q,p)$ parameterized in terms of $\theta$ as
\begin{equation}\label{eqn: curv gamma_infty > (q,p)}
     \gamma_\infty^{u,>}(\theta)= \left(q_\infty^{u,>}(\theta),\ p_\infty^{u,>}(\theta)\right) = \left(q_1(\theta),\ q_2(\theta),\ p_1(\theta),\ p_2(\theta)\right)
\end{equation}
where, from the estimates in \eqref{eqn: (x_inf,y_inf)>}, we have
\begin{equation*}
    \begin{aligned}
    q_1(\theta) =&\left(s_1(\theta)+u_1(\theta)\right)^2 - \left(s_2(\theta)+ u_2(\theta)\right)^2 + \mathcal{O}_6(s(\theta), u(\theta))\\
    =&\left(\pm\frac{\mu^{\frac \nu 2}}{\sqrt 2}  \sqrt{1-\cos\theta} + \mathcal{O}\left(\mu^{\frac{3(5\nu-1)}{8}}\right)\right)^2 -\left(\pm\frac{ \mu^{\frac \nu 2}}{\sqrt 2}\sqrt{1+\cos\theta}+\mathcal{O}\left(\mu^{\frac{3(5\nu-1)}{8}}\right)\right)^2 + \mathcal{O}(\mu^{3\nu}) \\
    =& - \mu^\nu\left(\cos\theta + \mathcal{O}\left(\mu^{\frac{11\nu - 3}{8}}\right)\right),\\
    q_2(\theta) =&\;2\left(s_1(\theta)+u_1(\theta)\right)\cdot\left(s_2(\theta)+u_2(\theta)\right)+ \mathcal{O}_6(s(\theta), u(\theta))\\
    =&\; 2\left(\pm \frac{\mu^{\frac \nu 2}}{\sqrt 2} \sqrt{1-\cos\theta} + \mathcal{O}\left(\mu^{\frac{3(5\nu-1)}{8}}\right)\right)\left(\pm \frac{\mu^{\frac \nu 2}}{\sqrt 2} \sqrt{1+\cos\theta}+\mathcal{O}\left(\mu^{\frac{3(5\nu-1)}{8}}\right)\right) +\mathcal{O}(\mu^{3\nu}) \\
    =&\;\mu^\nu\left(\sin\theta + \mathcal{O}\left(\mu^{\frac{11\nu - 3}{8}}\right)\right),\\
    p_1(\theta) =& -\frac{\left(s_1^2(\theta) - u_1^{2}(\theta) \right)- \left(s_2^{2}(\theta) - u_2^{2}(\theta) \right) + \mathcal{O}_6\left(s(\theta), u(\theta)\right)}{\left(s_1(\theta) + u_1(\theta)\right)^2 + \left(s_2(\theta) + u_2(\theta)\right)^2 + \mathcal{O}_6\left(s(\theta), u(\theta)\right)} =  - \frac{\mu^\nu+ \mathcal{O}\left(\mu^{\frac{19\nu-3}{8}}\right)}{\mu^\nu + \mathcal{O}\left(\mu^{\frac{19\nu-3}{8}}\right)} \\
    =&\; -1 + \mathcal{O}\left(\mu^{\frac{11\nu-3}{8}}\right),\\
    p_2(\theta) =&\;\frac{\left(s_1(\theta)+u_1(\theta)\right)\left(u_2(\theta)-s_2(\theta)\right) + \left(s_2(\theta) + u_2(\theta)\right)\left(u_1(\theta)-s_1(\theta)\right) + \mathcal{O}_6\left(s(\theta), u(\theta)\right)}{\left(s_1(\theta) + u_1(\theta)\right)^2 + \left(s_2(\theta) + u_2(\theta)\right)^2 + \mathcal{O}_6\left(s(\theta),u(\theta)\right)}\\
    &= \frac{\mu^{\frac{\nu}{2}}\mathcal{O}\left(\mu^{\frac{3(5\nu-1)}{8}}\right)}{\mu^\nu + \mathcal{O}\left(\mu^{\frac{19\nu-3}{8}}\right)} = \mathcal{O}\left(\mu^{\frac{11\nu-3}{8}}\right).
\end{aligned}
\end{equation*}
Moreover, the tangent vector 
\[ \partial_\theta\gamma_\infty^{u,>}(\theta)= \left(\partial_\theta q_\infty^{u,>}(\theta),\ \partial_\theta p_\infty^{u,>}(\theta)\right) = \left(q_1'(\theta),\ q_2'(\theta),\ p_1'(\theta),\ p_2'(\theta)\right)\]
is obtained from \eqref{eqn: (x_inf,y_inf)>} and \eqref{eqn:dersu} as
\begin{alignat*}{3}
        q_1'(\theta) &=\mu^\nu\left(\sin\theta + \mathcal{O}\left(\mu^{\frac{11\nu-3}{8}},\mu^{1-3\nu}\right)\right),\quad &&q_2'(\theta) &&= \mu^\nu\left(\cos\theta + \mathcal{O}\left(\mu^{\frac{11\nu-3}{8}},\mu^{1-3\nu}\right)\right),\\
        p_1'(\theta) &=\mathcal{O}\left(\mu^{\frac{11\nu-3}{8}},\mu^{1-3\nu}\right),\quad &&p_2'(\theta) &&=\mathcal{O}\left(\mu^{\frac{11\nu-3}{8}},\mu^{1-3\nu}\right).
\end{alignat*}
We apply the change $\Upsilon^{-1}$ in \eqref{eqn: (q,p) - (r,theta,R,Theta)} on the curve $ \gamma_\infty^{u,>}$ in \eqref{eqn: curv gamma_infty > (q,p)} to express it in coordinates $(\theta,\Theta)$ as a graph of the form
\begin{equation*}
    \gamma_\infty^{u,>} = \Big\{(\theta, \Theta_\infty^{u,>}(\theta))\colon \theta \in (\theta_- - \iota,\theta_-+\iota)\Big\},
\end{equation*}
 where $\Theta_\infty^{u,>}(\theta)$ satisfies
 \begin{equation*}
 \begin{aligned}
    \Theta_\infty^{u,>}(\theta) &= q_1(\theta)p_2(\theta) - q_2(\theta)p_1(\theta) =\mu^\nu\left(\sin\theta + \mathcal{O}\left(\mu^{\frac{11\nu-3}{8}}\right)\right),\\
    \partial_\theta\, \Theta_\infty^{u,>}(\theta) &= q_1'(\theta)p_2(\theta) + q_1(\theta)p_2'(\theta) - q_2'(\theta)p_1(\theta)-q_2(\theta)p_1'(\theta) = \mu^\nu\left(\cos\theta + \mathcal{O}\left(\mu^{\frac{11\nu-3}{8}},\mu^{1-3\nu}\right)\right),
\end{aligned}
\end{equation*}
completing the proof.
\printbibliography
\end{document}